\documentclass[10pt]{article} 
\usepackage[preprint]{tmlr}

\usepackage{amsmath,amsfonts,bm}

\def\eqref#1{equation~\ref{#1}}

\def\1{\bm{1}}

\def\ra{{\textnormal{a}}}

\def\rx{{\textnormal{x}}}

\def\rva{{\mathbf{a}}}

\def\erva{{\textnormal{a}}}

\def\ervx{{\textnormal{x}}}

\def\rmA{{\mathbf{A}}}

\def\vmu{{\bm{\mu}}}
\def\vtheta{{\bm{\theta}}}
\def\va{{\bm{a}}}

\def\ve{{\bm{e}}}

\def\vx{{\bm{x}}}

\def\eva{{a}}

\def\mA{{\bm{A}}}

\def\mH{{\bm{H}}}
\def\mI{{\bm{I}}}
\def\mJ{{\bm{J}}}

\def\mX{{\bm{X}}}

\def\mSigma{{\bm{\Sigma}}}

\DeclareMathAlphabet{\mathsfit}{\encodingdefault}{\sfdefault}{m}{sl}
\SetMathAlphabet{\mathsfit}{bold}{\encodingdefault}{\sfdefault}{bx}{n}
\newcommand{\tens}[1]{\bm{\mathsfit{#1}}}
\def\tA{{\tens{A}}}

\def\tX{{\tens{X}}}

\def\gG{{\mathcal{G}}}

\def\sA{{\mathbb{A}}}
\def\sB{{\mathbb{B}}}

\def\sS{{\mathbb{S}}}

\def\emA{{A}}

\newcommand{\etens}[1]{\mathsfit{#1}}

\def\etA{{\etens{A}}}

\newcommand{\E}{\mathbb{E}}

\newcommand{\R}{\mathbb{R}}

\newcommand{\KL}{D_{\mathrm{KL}}}
\newcommand{\Var}{\mathrm{Var}}

\newcommand{\Cov}{\mathrm{Cov}}

\newcommand{\normltwo}{L^2}
\newcommand{\normlp}{L^p}

\newcommand{\parents}{Pa} 

\usepackage{hyperref}
\usepackage{url}

\title{Networked Communication for Decentralised Cooperative Agents in Mean-Field Control}

\author{\name Patrick Benjamin \email phlbenjamin@gmail.com \\
      \addr Department of Computer Science\\
      University of Oxford
      \AND
      \\\;\\
      \name Alessandro Abate \email alessandro.abate@cs.ox.ac.uk \\
      \addr Department of Computer Science\\
      University of Oxford
      }

\def\month{MM}  
\def\year{YYYY} 
\def\openreview{\url{https://openreview.net/forum?id=XXXX}} 

\usepackage{amsthm}

\theoremstyle{plain}
\newtheorem{theorem}{Theorem}[section]

\newtheorem{corollary}[theorem]{Corollary}

\theoremstyle{definition}  
\newtheorem{definition}[theorem]{Definition}
\newtheorem{assumption}[theorem]{Assumption}
\newtheorem{remark}[theorem]{Remark}

\let\AND\relax
\usepackage{algorithm}
\usepackage{algorithmic}
\usepackage{graphicx}
\usepackage{booktabs}
\usepackage{dsfont}
\usepackage{subcaption}

\begin{document}

\maketitle

\begin{abstract} 
The mean-field framework has been used to find approximate solutions to problems involving very large populations of symmetric, anonymous agents, which may be intractable by other methods. The cooperative mean-field control (MFC) problem has received less attention than the non-cooperative mean-field game (MFG), despite the former potentially being more useful as a tool for engineering large-scale collective behaviours. Decentralised communication algorithms have recently been introduced to MFGs, giving benefits to learning speed and robustness. Inspired by this, we introduce networked communication to MFC - where populations arguably have broader incentive to communicate - and in particular to the setting where decentralised agents learn online from a single, non-episodic run of the empirical system. We adapt recent MFG algorithms to this new setting, as well as contributing a novel sub-routine allowing networked agents to estimate the global average reward from their local neighbourhood. Previous theoretical analysis of decentralised communication in MFGs does not extend trivially to MFC. We therefore contribute new theory proving that in MFC the networked communication scheme allows agents to increase social welfare faster than under \textit{both} of the two typical alternative architectures, namely independent and centralised learning. We provide experiments that support this new result across different classes of cooperative game, and also give numerous ablation studies and additional experiments concerning numbers of communication round and robustness to communication failures.
\end{abstract}

\section{Introduction}\label{introduction_section}

The mean-field framework  \citep{lasry2007mean,huangMFG} models a representative agent as interacting not with the rest of the population on a per-agent basis, but instead with a distribution over the other agents, known as the \textit{mean field}. The framework analyses the limiting case when the population consists of an infinite number of symmetric and anonymous agents, that is, they have identical reward and transition functions which depend on the mean-field distribution rather than on the actions of specific other players. The mean-field \textit{control} (MFC) problem is a cooperative scenario where the population seeks to maximise a social welfare criterion such as the average return received by the agents. Alternatively we can consider a non-cooperative scenario called a mean-field \textit{game} (MFG), where each agent seeks to maximise its individual return, to which the solution is a MFG-Nash equilibrium (MFG-NE). 

The MFC social optimum and the MFG-NE can respectively be used as approximate solutions to the associated finite-agent problem/game, with the error in the solution reducing as the number of agents \textit{N} tends to infinity \citep{finite_proofs, doi:10.1137/20M1360700,10.5555/3586589.3586718,major_minor_control,cui2023learningdecentralisedcontrol,Anahtarci2020QLearningIR,yardim2024mean, TOUMI2024111420, hu2024mfoml,chen2024periodic,bayraktar2024learning}
. MFC and MFGs have therefore been used to address the difficulty faced by multi-agent reinforcement learning (MARL), which can struggle to scale computationally as \textit{N} increases \citep{yardim2024exploiting,zeng2024single}. While MFGs have been well-studied and applied to a wide variety of real-world problems \citep{survey_learningMFGs}, MFC has received less attention, despite possibly being more useful for engineering collective behaviours to achieve global objectives, such as in consensus, synchronisation, rendezvous, exploration, coverage or task allocation problems \citep{cui2023learningdecentralisedcontrol}. This paper seek to redress some of this imbalance by adapting recent developments in MFGs for MFC.

Since MFC problems can be interpreted as optimisation problems from the perspective of a social planner, classical approaches involve centralised methods (they also do so for reasons of simplicity, as in MFGs). In this context `centralised' does not necessarily imply global observability of the whole population's actions - which could make computation infeasible given the complexity of the problem - but rather that learning is conducted from the samples of a single representative agent, whose policy updates are assumed to be automatically pushed to the rest of the population by the central node \citep{Fornasier_Solombrino_2014,carmona2019linear,doi:10.1073/pnas.1922204117,survey_learningMFGs,angiuli2021unified,angiuli2023convergence,cui2023multiagent,lee2024mean,denkert2024randomisation}
. For this reason, whilst `centralised learning' is the term used in prior works, we generally refer to `central-agent learning' to reduce confusion. Often the empirical mean field of the actual population is not even used to compute rewards or transitions, with the central learner instead updating an estimate of the mean field based only on its own policy, which is in turn used as input to its reward and transition functions \citep{carmona2019linear,angiuli2021unified,angiuli2023convergence}.

Recent works on MFGs, as in other areas of multi-agent research, have recognised that the existence of a central learner is a strong assumption in complex, real-world settings, as well as representing 
a bottleneck for computation and communication, and a vulnerable single point of failure of the system 
\citep{mainone,decentralised_review,marlreview,distributed_review,policy_mirror_independent, benjamin2024networked, benjamin2024networkedapproximation,jiang2024fully,xu2025targets,AGYEMAN2025106183,10971233}. They advocate instead for the individual agents in the empirical population to learn policies for themselves without relying on a central node. Such works also argue that other strong classical assumptions should similarly be loosened in order to make MFGs applicable to real-world, embodied problems such as swarm robotics. They therefore contend that, aside from decentralised learning, desirable qualities for mean-field algorithms include: learning from the population's empirical mean field (i.e. the mean-field distribution is generated only by the agents' policies, rather than being manipulated 
by the algorithm itself or by an external oracle/simulator); learning online from a single, non-episodic system run (i.e. similar to above, the population is not arbitrarily reset by an external controller); learning without reliance on a model of the system; and using function approximation to allow scalability to high-dimensional observations (including the option to include the mean field in the input to policies). 

Until now, no work on MFC has met all these criteria. Some recent works have considered decentralisation in MFC, but \citet{bayraktar2024learning} requires that decentralised agents optimise for learnt models of the system dynamics (and learning is only fully independent when the population is large but finite rather than infinite), while \citet{cui2023learningdecentralisedcontrol} presents a model-free deep learning algorithm that gives decentralised execution but requires centralised, episodic training. This latter work stipulates that decentralised training can be achieved if all agents can directly observe the mean-field distribution and use the same seed to correlate their actions, though they only provide empirical results for the centralised scenario, while \citet{bayraktar2024learning} provides no empirical results at all. However, assuming decentralised agents have access to this global information is unrealistic, and in the non-cooperative MFG setting \citet{benjamin2024networkedapproximation} have shown that networked communication between decentralised agents allows agents to estimate the global mean field from a local neighbourhood. They also show that proliferating high-performing policies through the population via decentralised communication (in a manner reminiscent of distributed embodied evolutionary algorithms \citep{survivability, Maintaining_Diversity, 10.1007/978-3-030-16692-2_38, cazenille2025signalling, sissodia2025evolutionary}) improves training time and avoidance of local optima, particularly over the case of agents learning entirely independently, but often also over populations with a single central learner. 



Inspired by this non-cooperative MFG work, we introduce networked communication to MFC for the first time, where populations arguably have even more incentive to communicate. This allows us to present a model-free deep learning algorithm that fulfils all of the proposed desiderata, including learning online from a single non-episodic run of the empirical system, and decentralised training without needing to observe global information: we contribute a novel sub-routine for estimating the global average reward from local communication, in addition to the existing sub-routine for estimating the global mean field from \citet{benjamin2024networkedapproximation}.
Previous theoretical analysis of networked communication in the non-cooperative MFG setting does not extend trivially to MFC, so we contribute new theoretical proofs showing that decentralised policy exchange allows networked populations to learn faster than both the independent \textit{and} the central-agent alternatives in the MFC setting, across different classes of cooperative game (coordination and anti-coordination). 
We also demonstrate this finding empirically in numerous games, as well as contributing an empirical study of the algorithms' robustness to communication failures, along with several ablation studies. 
In summary, our contributions include:

\begin{itemize}
\item We provide the first algorithms in MFC for model-free training without any central provision of information or coordination, as well as the first MFC algorithms for online learning from a single, non-episodic run of the empirical system. 
\begin{itemize}
    \item We contribute a novel sub-routine allowing decentralised agents to estimate the global average reward via networked communication, and incorporate an existing sub-routine used in MFGs for estimating the global mean field via local communication.
\end{itemize}
\item We prove theoretically that in this context, decentralised networked communication can improve learning speed over the independent \textit{and} central-agent architectures. 
\item We provide extensive experiments supporting our theoretical results in numerous games, and give ablation studies of various parts of our algorithms, as well as a study of robustness to communication failures. 

\end{itemize}

We provide  further comparison with related work in Sec. \ref{related_work_section}, give preliminaries in Sec. \ref{background_section}, and our algorithms in Sec. \ref{algorithm_section}. We present theoretical results in Sec. \ref{theory_section} and experiments in Sec. \ref{experiments_section}, before suggesting future work in Sec. \ref{conclusion_section}.

\section{Related work}\label{related_work_section}

We discuss here the research most closely related to our present work, focusing on decentralisation and networked communication, and clarifying the differences with prior methods and settings. We refer the reader to \citet{survey_learningMFGs} for a broader survey of MFC.

Numerous works claiming to study decentralisation in MFC take this to mean only that agents do not have access to the specific states of all other agents, and have policies depending on their local state and possibly the mean field, all of which we take as a given in our work. They nevertheless rely on a central learner or coordinator that provides global information to all agents, a dependence that we remove in our work. This applies, for example, to \citet{7368131}, where a `central population coordinator' broadcasts a common signal to all agents, and to \citet{tajeddini2017robust}, which presents a leader-follower setting where a `central population coordinator' estimates the mean-field trajectory. \citet{8879544} similarly requires a central coordinator, and also presents a non-cooperative scenario so does not actually fall under MFC despite being referred to as such.


In \citet{cui2023learningdecentralisedcontrol}, decentralisation applies only during execution, and they offer a centralised-training decentralised-execution method (as also in \citet{cui2023multiagent}). They say that decentralised training could be achieved if the global mean field is observable and all agents use the same seed to correlate their actions, whilst we do not require either assumption for our decentralised training algorithm. They also train episodically whereas we learn online from a single run of the system. Finally, their experiments focus only on coordination games, whereas we additionally explore empirical effects resulting from decentralised training in anti-coordination games, where agents can gain higher rewards by diversifying their behaviour.

\citet{bayraktar2024learning} considers independent, `online' learning for MFC in a setting that is different from ours. Crucially, their method involves agents first estimating a model (reward and transition functions) of the system by conducting `online' updates using samples collected while following exploration policies. Only once having done so do they compute execution policies that are optimal with respect to the estimated model. We argue that having a dedicated exploration phase is infeasible for many real-world applications, and instead present a fully model-free online learning algorithm. Moreover, their setting only permits independent learning if $N$ is large but finite. For infinite populations, a central coordinator is required to supply common noise to aid exploration during the initial phase, and if the optimal policy for the estimated model is not unique, centralised coordination is required to allow the agents to agree on which policy to execute. Our algorithm requires no such special considerations. Finally, their work is purely theoretical, whereas we provide extensive empirical results.


\citet{angiuli2021unified} and \citet{angiuli2023convergence} provide algorithms for MFC learning from a single run, but there it is a single run only of a `representative' player that is used to simulate the mean field, rather than a single run of the empirical population as in our work. Their algorithms are thus inherently centralised, as well as involving two timescales for updating the mean-field approximation, which we argue is unlikely to be a practical paradigm for training in complex real-world systems such as robotic swarms.

Our work is also closely related to \citet{benjamin2024networked} and \citet{benjamin2024networkedapproximation}, which introduce networked communication to the non-cooperative MFG setting. By adapting their communication scheme and learning algorithm, we introduce networked communication to the cooperative MFC setting, where it is arguably more applicable due to broader incentives for agents to communicate policies. Their works focus on coordination games to justify the sharing of policies (though \citet{benjamin2024networkedapproximation} does demonstrate empirically that networked agents outperform independent agents in a non-cooperative anti-coordination game, indicating that self-interested agents do nevertheless have incentive to communicate), whilst we provide extensive theoretical and empirical results on the benefits of policy sharing in MFC for both coordination and anti-coordination games. We leverage Alg. \ref{alg:mean_field_estimation_specific} from \citet{benjamin2024networkedapproximation} for estimating the global mean field from a local neighbourhood, but additionally contribute the novel Alg. \ref{average_reward_alg} for estimating the global average reward from a local neighbourhood for the MFC setting.


\section{Preliminaries}\label{background_section}






\subsection{Mean-field control}\label{prelim_mfc}

We use the following notation. $N$ is the number of agents in a population, with $\mathcal{S}$ and $\mathcal{A}$ representing the finite state and common action spaces. The set of probability measures on a finite set $\mathcal{X}$ is denoted $\Delta_\mathcal{X}$, and $\mathbf{e}_x \in \Delta_\mathcal{X}$ for $x \in \mathcal{X}$ is a one-hot vector with only the entry corresponding to $x$ set to 1, and all others set to 0. For time $t \geq 0$, $\hat{\mu}_t$ = $\frac{1}{N}\sum^N_{i=1}\sum_{s\in\mathcal{S}}$ $\mathds{1}_{s^i_t=s}\mathbf{e}_s$ $\in$ $\Delta_\mathcal{S}$ is a vector of length $|\mathcal{S}|$ denoting the empirical categorical state distribution of the $N$ agents at time $t$. For agent $i\in\{1\dots N\}$, $i$'s policy $\pi^i\in\Pi$ depends on its observation $o^i_t$. We give different forms that this observation can take, and relatedly a more formal definition of the policy, after the following.

\begin{definition}[\textit{N}-player stochastic cooperative control problem with symmetric, anonymous agents] 
This is given by the tuple $\langle$$N$, $\mathcal{S}$, $\mathcal{A}$, $P$, $R$, $\gamma$$\rangle$, where $\mathcal{A}$ is the action space, identical  for each agent, $\mathcal{S}$ is the identical  state space of each agent, such that their initial states are \{$s^i_0$\}$_{i=1}^N \in \mathcal{S}^N$ sampled from some initial distribution $\mu_0\in\Delta_{\mathcal{S}}$, and their policies are \{$\pi^i$\}$_{i=1}^N \in \Pi^N$. $P$ : $\mathcal{S}$ $\times$ $\mathcal{A}$ $\times$ $\Delta_{\mathcal{S}}$ $\rightarrow$ $\Delta_{\mathcal{S}}$ is the transition function and $R$ : $\mathcal{S}$ $\times$ $\mathcal{A}$ $\times$ $\Delta_{\mathcal{S}}$ $\rightarrow$ [0,1] is the reward function, both identical to all agents, and which map each agent's local state and action and the population's empirical distribution to transition probabilities and bounded rewards, respectively, i.e. $\forall i \in \{1,\dots,N\}$: $s^i_{t+1} \sim P(\cdot|s^i_{t},a^i_{t},\hat{\mu}_t)$ and $r^i_{t} = R(s^i_{t},a^i_{t},\hat{\mu}_t)$.
\end{definition} 

For the joint policy $\boldsymbol\pi$ := ($\pi^1,\dots,\pi^N$) $\in \Pi^N$, an individual agent's discounted return is given by:

\begin{definition}[Individual expected discounted return]\label{Discounted_individual_expected_return} For all $i,j \in \{1,\dots,N\}$, $i$'s return is \[V^i(\boldsymbol{\pi},\mu_{\bar{t}}) =  \mathbb{E}\left[\sum^{\infty}_{t={\bar{t}}}\gamma^{t}R(s^i_t,a^i_t,\hat{\mu}_t) \bigg|\substack{s^j_{\bar{t}}\sim\mu_{\bar{t}} \\ a^j_t\sim \pi^j(o^j_t)\\ s^j_{t+1} \sim P(\cdot|s^j_{t},a^j_{t},\hat{\mu}_t)}\right].\]\end{definition}

However, the maximisation objective for this \textit{cooperative} problem is:

\begin{definition}[Population-average expected discounted  return]\label{Discounted_population_average_expected_return}
For $i,j \in \{1,\dots,N\}$ the return is
\[V^{pop}(\boldsymbol{\pi},\mu_{\bar{t}}) = \frac{1}{N}\sum_i^N V^i(\boldsymbol{\pi},\mu_{\bar{t}}) = \mathbb{E}\left[\frac{1}{N}\sum^{\infty}_{t={\bar{t}}}\sum_i^N\gamma^{t}R(s^i_t,a^i_t,\hat{\mu}_t) \bigg|\substack{s^j_{\bar{t}}\sim\mu_{\bar{t}} \\ a^j_t\sim \pi^j(o^j_t)\\ s^j_{t+1} \sim P(\cdot|s^j_{t},a^j_{t},\hat{\mu}_t)}\right].\]
\end{definition}

That is, the solution to the control problem is $\boldsymbol{\pi}^* = \arg\max_{\boldsymbol{\pi}\in\Pi^N} V^{pop}(\boldsymbol{\pi},\mu_{\bar{t}})$. 

At the limit as $N \rightarrow \infty$, the infinite population of agents can be characterised as a limit distribution $\mu \in \Delta_\mathcal{S}$; the infinite-agent setting is termed a MFC problem. The \textit{mean-field flow} $\boldsymbol\mu$ is given by the infinite sequence of mean-field distributions s.t.  $\boldsymbol\mu = (\mu_t)_{t\geq0}$. 

\begin{definition}[Induced mean-field flow] 

We denote by $I(\pi)$ the mean-field flow $\boldsymbol\mu$ induced when all the agents follow $\pi$, where this is generated from $\pi$ by \[\mu_{t+1} (s') = \sum_{s,a}\mu_t (s)\pi (a|o_t)P(s'|s,a,\mu_t).\] The snapshot of this induced flow at $t$ is given by $I(\pi)_t$.
\end{definition}

\begin{definition}[Social welfare]\label{social_welfare_function}
When all agents follow policy $\pi$ giving mean-field flow $\boldsymbol\mu = I(\pi)$, $\pi$'s social welfare is \[W(\pi;I(\pi)) = \mathbb{E}\left[
\sum^{\infty}_{t={\bar{t}}}\gamma^{t}(R(s_t,a_t,I(\pi)_t))\bigg|
\substack{s_{\bar{t}}\sim\mu_{\bar{t}}\\ 
a_t\sim \pi(\cdot|o_t)\\ s_{t+1} \sim P(\cdot|s_{t},a_{t},{I(\pi)_t})}\right].\]
\end{definition}

\begin{definition}[Social optimum]\label{social_optimum} The solution to the MFC problem is a social optimum policy $\pi^*\in\Pi$ that maximises the social welfare function in Def. \ref{social_welfare_function}, i.e. $\pi^* = \arg\max_{\pi\in\Pi} W(\pi; I(\pi))$.\end{definition}

\begin{remark}
Previous works showed that the MFC social optimum $\pi^*$ gives a good approximation for the harder-to-solve finite-agent problem (i.e. if $\boldsymbol\pi=(\pi^*,\dots,\pi^*)$), with the error characterised by $\mathcal{O}(\frac{1}{\sqrt{N}})$ \citep{doi:10.1137/20M1360700,10.5555/3586589.3586718,major_minor_control,cui2023learningdecentralisedcontrol,bayraktar2024learning}.
\end{remark}

When the distribution is the same for all $t$, i.e. $\mu_t=\mu_{t+1}$ $\forall t \geq 0$, we say the mean-field flow is \textit{stationary}, giving a stationary MFC problem. \textit{Non-stationary} problems may require the policy to depend on the mean field such that $o^i_t = (s^i_t, \hat{\mu}_t)$, whereas the observation in the stationary case can be simplified to $o^i_t = s^i_t$. However, since classical approaches to the MFC problem often conceive of a central planner trying to guide the population to a distribution that maximises the expected return, they sometimes have policies that depend on the mean field even in the stationary case \citep{survey_learningMFGs,carmona2021modelfree,cui2023learningdecentralisedcontrol}. 
Therefore \textit{we permit mean field-dependent policies for the sake of generality, but show through our ablation studies that in practice our algorithms require only $\pi^i(a|o^i_t) = \pi^i(a|s^i_t)$} in our experimental tasks, which have \textit{stationary} solutions.

Furthermore, it is unrealistic to assume that decentralised agents with a possibly limited communication radius would have perfect observability of the global mean field $\hat{\mu}_t$. Therefore we allow agents to form a local estimate $\tilde{\hat{\mu}}^i_t$ which can be improved by communication with neighbours, using Alg. \ref{alg:mean_field_estimation_specific} (from Alg. 3 in \citet{benjamin2024networkedapproximation} for the MFG setting). We thus have $o^i_t = (s^i_t, \tilde{\hat{\mu}}^i_t)$. Formally we can now say that when  $o^i_t = (s^i_t, \hat{\mu}_t)$ or $(s^i_t, \tilde{\hat{\mu}}^i_t)$, we have the set of policies defined as $\Pi$ = \{$\pi$ : $\mathcal{S} \times \Delta_\mathcal{S} \rightarrow \Delta_\mathcal{A}$\}, {and the set of Q-functions denoted $\mathcal{Q} = \{q : \mathcal{S} \times \Delta_\mathcal{S} \times \mathcal{A} \rightarrow \mathbb{R}\}$}. (N.b. when $o^i_t = s^i_t$, we instead have $\Pi$ = \{$\pi$ : $\mathcal{S}  \rightarrow \Delta_\mathcal{A}$\} and $\mathcal{Q} = \{q : \mathcal{S}  \times \mathcal{A} \rightarrow \mathbb{R}\}$.)

\begin{figure*}[t]
    \centering
    \includegraphics[width=\textwidth]{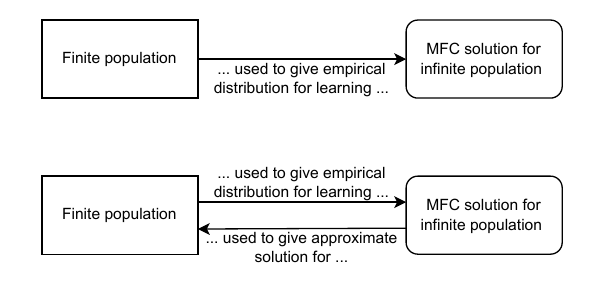}
    \caption{Two possible ways to conceive of our work regarding the relationship between the  infinite- and finite-population control problems, described in Rem. \ref{conceptions_remark}. 
    Note that using the finite empirical population to learn the single-policy MFC social optimum $\boldsymbol\pi=(\pi^*,\dots,\pi^*)$ for the infinite population (Def. \ref{social_optimum}) is \textit{not} the same as directly finding $\boldsymbol{\pi}^* = \arg\max_{\boldsymbol{\pi}\in\Pi^N} V^{pop}(\boldsymbol{\pi},\mu_{\bar{t}}) = (\pi^1,\dots,\pi^N)$, i.e. the tuple of \textit{individual} policies that maximises the expected finite population-average return in Def. \ref{Discounted_population_average_expected_return}, a problem known to be hard \citep{cui2023learningdecentralisedcontrol,b57d2987-f484-384d-ab52-5bac787e11f2}.}
    \label{conceptions_diagram}
\end{figure*}
\vspace{0.3cm}
\begin{remark}\label{conceptions_remark}
We have just seen that an optimal solution to the theoretical MFC problem is a single policy that, when followed by all agents in the infinite population, maximises the population's expected return. We give two ways to conceive of our work, illustrated in Fig. \ref{conceptions_diagram}, which mirror and make more explicit the motivations underpinning other MFC works \citep{cui2023learningdecentralisedcontrol, dayanikli2024deep,zaman2024robust,bayraktar2024learning,10928821}. 
\begin{enumerate}
    \item Firstly, while previous works might make unrealistic assumptions about access to an oracle for the infinite population, we contribute algorithms that allow the solution to a MFC problem to be learnt using the empirical distribution of a decentralised finite population. Note that it is unnecessary (and may be impractical) to assume the decentralised agents always follow a single identical policy throughout training. 
    \item Alternatively, we may have originally been interested in solving a cooperative problem for a large, finite population, but, due to the scalability issues of learning approaches like MARL, were forced to turn to the MFC framework to find a policy that gives an approximate solution to the finite-population problem. We contribute algorithms that allow the deployed finite population to find the MFC solution that in turn approximately solves the original problem, without unrealistic assumptions about centralised training. Under this framing, it may matter less whether all agents follow a single policy in practice (\citet{policy_mirror_independent} and \citet{benjamin2024networked,benjamin2024networkedapproximation} follow a similar logic in MFGs). 
\end{enumerate}

\end{remark}

\subsection{Munchausen Online Mirror Descent}\label{omd_appx}

Recent works have solved MFGs from non-episodic runs of the finite-population empirical system using a form of policy iteration called Online Mirror Descent (OMD) \citep{benjamin2024networkedapproximation}; we adapt this to learn a social optimum in the MFC setting. OMD involves beginning with an initial policy $\pi_0$, and then at each iteration $k$, evaluating the current policy $\pi_k$ with respect to its induced mean-field flow $\boldsymbol{\mu} = I(\pi_k)$ to compute its Q-function $Q_{k+1}$. To stabilise the learning process, we then use a weighted sum over this and past Q-functions, and set $\pi_{k+1}$ to be the softmax over this weighted sum, i.e.  $\pi_{k+1}(\cdot|o) = softmax\left(\frac{1}{\tau_q}\sum^{k+1}_{\kappa=0}{Q_\kappa}(o,\cdot)\right)$. $\tau_{q}$ is a temperature parameter that scales the entropy in Munchausen RL (see Sec. \ref{q_update}) \citep{NEURIPS2020_2c6a0bae}; this is a different temperature to the one agents use when communicating policies, denoted $\tau^{comm}_k$ and discussed in Sec. \ref{policy_adoption_and_communication}.

If the Q-function is approximated non-linearly, it is difficult to compute this weighted sum. The \textit{Munchausen trick} addresses this by computing a single Q-function that mimics the weighted sum using implicit regularisation based on the Kullback-Leibler (KL) divergence between $\pi_{k}$ and $\pi_{k+1}$ \citep{NEURIPS2020_2c6a0bae}. Using this reparametrisation gives Munchausen OMD (MOMD), detailed in Sec. \ref{q_update} \citep{scalable_deep, wu2024populationaware}. MOMD does not bias policies, and has the same convergence guarantees as OMD \citep{hadikhanloo2017learninganonymousnonatomicgames, perolat2021scaling, wu2024populationaware}.

\subsection{Networks}\label{preliminaries_networks}
Our decentralised population exhibits two time-varying graphs. The first is a communication network, by which agents can exchange information:

\begin{definition}[Time-varying communication network] 
The time-varying graph ($\mathcal{G}^{comm}_t$)$_{t\geq0}$ is given by $\mathcal{G}^{comm}_t$ = ($\mathcal{N}, \mathcal{E}^{comm}_t$), where $\mathcal{N}$ is the set of vertices each representing an agent $i \in \{1,\dots,N\}$, and the edge set $\mathcal{E}_{t}$ $\subseteq$ \{(\textit{i},\textit{j}) : \textit{i},\textit{j} $\in$ $\mathcal{N}$, \textit{i} $\neq$ \textit{j}\} is the set of undirected links present at time \textit{t}. A network's \textit{diameter} $d_{\mathcal{G}^{comm}_t}$ is the maximum of the shortest path lengths between any pair of nodes.\end{definition}

In principle, agents can use this same communication network to receive information about others' state in order to estimate the mean field. However, we also define an alternative observation graph that is useful in a specific subclass of environments, which can most intuitively be thought of as those where agents' states are positions in physical space, which include those in our experiments. When this is the case, we usually think of agents' ability to observe each other as depending more abstractly on whether states are visible to each other. This visibility graph is:

\begin{definition}[Time-varying state-visibility graph] 
The time-varying state visibility graph ($\mathcal{G}^{vis}_t$)$_{t\geq0}$ is given by $\mathcal{G}^{vis}_t$ = ($\mathcal{S'}, \mathcal{E}^{vis}_t$), where $\mathcal{S'}$ is the set of vertices representing the environment states $\mathcal{S}$, and the edge set $\mathcal{E}^{vis}_{t}$ $\subseteq$ \{(\textit{m},\textit{n}) : \textit{m},\textit{n} $\in$ $\mathcal{S'}$\} is the set of undirected links present at time \textit{t}, indicating which states are visible to each other.
\end{definition} 

In Sec. \ref{mean_field_estimate_section} we present Alg. \ref{alg:mean_field_estimation_specific}, which forms an initial estimate of the global empirical mean field (to serve as an observation input for agents' Q-/policy-networks) via the visibility graph $\mathcal{G}^{vis}_t$, before refining this estimate via the communication graph $\mathcal{G}^{comm}_t$. \citet{benjamin2024networkedapproximation} discusses an algorithm for more general settings where the visibility graph $\mathcal{G}^{vis}_t$ does not apply.

\section{Learning and estimation algorithms}\label{algorithm_section}

We adapt recent algorithms for the MFG setting, where networked communication is used 1) to form local estimates of the global empirical mean field, and 2) to allow agents to adopt better-performing policies from neighbours to accelerate learning \citep{benjamin2024networkedapproximation}. We adapt these algorithms for cooperative MFC, where decentralised agents must optimise the population-average return instead of their individual one (the decentralised agents may not always follow a common policy while training unless we make strong assumptions on the communication network as in Sec. \ref{theory_section}, so we do not directly optimise social welfare from Def. \ref{social_welfare_function}). 

It is unrealistic to assume that decentralised agents have access to the global average reward, so we find a third use of the communication network in 3) allowing agents to estimate the global average reward $\hat{r}_t$ from a local neighbourhood. We contribute a novel algorithm Alg. \ref{average_reward_alg} for this purpose (Sec. \ref{reward_estimation_section}), and we describe our main learning method Alg. \ref{main_algo} in Sec. \ref{q_update}. Our policy communication algorithm Alg. \ref{policy_communication_algo}, based on that in \citet{benjamin2024networkedapproximation} for the MFG setting, is described in Sec. \ref{policy_adoption_and_communication}. Meanwhile Alg. \ref{alg:mean_field_estimation_specific} for estimating the mean field, which is taken from Alg. 3 in \citet{benjamin2024networkedapproximation} for the MFG setting, is described in Sec. \ref{mean_field_estimate_section}.

\begin{algorithm}[t]
\caption{Average reward estimation and communication}
\label{average_reward_alg}
\begin{algorithmic}[1]
\REQUIRE Time-dependent communication graph $\mathcal{G}^{comm}_t$, rewards $\{r^i_t\}_{i=1}^{N}$, number of communication rounds $C_r$
\STATE $\forall i :$ Initialise reward sets $\hat{\mathcal{R}}^i_{t,1} \leftarrow \{(ID^i, r^i_t)\}$ \label{prepare_unique_reward}
\FOR{$c_r$ in $1, \dots, C_r$}\label{begin_broadcast_reward_set}
    \STATE $\forall i :$ Broadcast $\hat{\mathcal{R}}^i_{t, c_r}$
    \STATE $\forall i : J^i_t \leftarrow \{j \in \mathcal{N} : (i, j) \in \mathcal{E}^{comm}_t\}$
    \STATE $\forall i :$ $\hat{\mathcal{R}}^i_{t, (c_r+1)} \leftarrow \hat{\mathcal{R}}^i_{t, c_r} \cup \bigcup_{j \in J^i_t} \hat{\mathcal{R}}^j_{t, c_r}$
\ENDFOR\label{end_broadcast_reward_set}
\STATE $\forall i :$ $\tilde{\hat{r}}^i_t \leftarrow \frac{1}{|\hat{\mathcal{R}}^i_{t, C_r}|} \sum_{(ID, r) \in \hat{\mathcal{R}}^i_{t, C_r}} r$\label{average_rewards_collected}
\STATE \textbf{return} Estimates of average reward $\left\{\tilde{\hat{r}}^i_t\right\}_{i=1}^N$
\end{algorithmic}
\end{algorithm}

\subsection{Sub-routine for networked estimation of global average reward}\label{reward_estimation_section}

Our novel Alg. \ref{average_reward_alg} involves agents using the communication network $\mathcal{G}^{comm}_t$ to locally estimate the global population-average reward received after a given step in the environment. Maximising the population-average reward ensures agents are solving the cooperative MFC problem instead of the non-cooperative MFG. Agents broadcast their received reward with a unique ID to ensure each reward is only counted once (Line \ref{prepare_unique_reward}). They collect those received from neighbours, and repeat the process of broadcasting and expanding their collections for a further $C_r - 1$ rounds, so as to receive rewards from agents more than one hop away on the network (Lines \ref{begin_broadcast_reward_set}-\ref{end_broadcast_reward_set}). They finally set their estimate of the global average to the average of the rewards they have collected (Line \ref{average_rewards_collected}).

\subsection{Main learning algorithm for updating Q-networks and policies
}\label{q_update}

\begin{algorithm}[t]
    \caption{Decentralised MFC learning from non-episodic system run} \label{main_algo}
    \begin{algorithmic}[1]
    \REQUIRE loop parameters $K, M, L, E, C_e,C_r,C_p$,   learning parameters
    $\gamma, \tau_q, |B|, 
    cl, \nu$, $\{\tau^{comm}_k\}_{k \in \{0,\dots,K-1\}}$
    \REQUIRE initial states $\{s^i_0\}_{i=1}^{N}$; $t \leftarrow 0$ 
    \STATE $\forall i$ : Randomly initialise parameters $\theta^i_0$ of Q-networks $\check{Q}_{\theta^i_0}(o,\cdot)$, and set $\pi^i_0(a|o)= $ softmax$\left(\frac{1}{\tau_q}\check{Q}_{\theta^i_0}(o,\cdot)\right)(a)$\label{initalise_line}
    \FOR{$k \in 0,\dots,K-1$}
        \STATE $\forall i$: Empty $i$'s buffer\label{empty_buffer}
        \FOR{$m \in 0,\dots,M_{}-1$}\label{start_m_loop}
        \STATE $\{o^i_t\}_{i=1}^{N} \gets$ \textbf{EstimateMeanFieldAlg. \ref{alg:mean_field_estimation_specific}}$\left(\mathcal{G}^{vis}_t,\mathcal{G}^{comm}_t,\{s^i_t\}_{i=1}^{N}\right)$
            \STATE Take step $\forall i : a^i_t \sim \pi^i_{k}(\cdot|o^i_t), r^i_{t} = R(s^i_{t},a^i_{t},\hat{\mu}_t),s^i_{t+1} \sim P(\cdot|s^i_{t},a^i_{t},\hat{\mu}_t)$; $t \leftarrow t + 1$\label{step1}
            \STATE $\{\tilde{\hat{r}}^i_t\}_{i=1}^N \gets$ \textbf{EstimateAverageRewardAlg. \ref{average_reward_alg}}$\left(\mathcal{G}^{comm}_t, \{r^i_t\}_{i=1}^{N}\right)$\label{get_average_reward_estimate}
            \STATE $\forall i$: Add $\left(o^i_t,a^i_t,\tilde{\hat{r}}^i_t,o^i_{t+1}\right)$ to $i$'s buffer\label{store_transition}   
        \ENDFOR\label{end_m_loop}
        \FOR{$l \in 0,\dots,L-1$}\label{beginlearningwithbuffer}
            \STATE $\forall i : $ Sample batch $B^i_{k,l}$ from $i$'s buffer\label{sample_from_buffer}
            \STATE Update $\theta$ to minimise $\hat{\mathcal{L}}(\theta,\theta')$ as in Def. \ref{loss}\label{minimise_loss_line}
            \STATE If $l\mod$ $\nu = 0$, set $\theta'\leftarrow\theta$ \label{update_target_line}
        \ENDFOR \label{endlearningwithbuffer}
        \STATE $\check{Q}_{\theta^i_{k+1}}(o,\cdot) \leftarrow \check{Q}_{\theta^i_{k,L}}(o,\cdot)$ 
        \STATE $\forall i$ : $\pi^i_{k+1}(a|o) \gets $ softmax$\left(\frac{1}{\tau_q}\check{Q}_{\theta^i_{k+1}}(o,\cdot)\right)(a)$\label{policy_update_line}
        \STATE $\left(\{\pi^i_{k+1}\}_i, \{s^i_t\}_i, t\right)$ $\gets$ \textbf{CommunicatePolicyAlg. \ref{policy_communication_algo}}$\left(\mathcal{G}^{comm}_t,\{\pi^i_{k+1}\}_i, \{s^i_t\}_i, t\right)$\label{run_CommunicatePolicyAlg_line}
    \ENDFOR
    \STATE \textbf{return}  policies \{$\pi^i_K$\}$_{i=1}^{N}$
    \end{algorithmic}
\end{algorithm}

Our novel Alg. \ref{main_algo}, adapted from non-cooperative Alg. 1 in \citet{benjamin2024networkedapproximation}, contains the core method for online MFC learning using the empirical mean field in a non-episodic system run. Our MOMD-based method (Sec. \ref{omd_appx}) works as follows. Each agent $i$ approximates its Q-function $\check{Q}_{\theta^i_k}(o,\cdot)$ with its own neural network parametrised by $\theta^i_k$. Agent $i$'s policy is determined by \[\pi_{\theta^i_k}(a|o)= \text{softmax}\left(\frac{1}{\tau_{q}}\check{Q}_{\theta^i_k}(o,\cdot)\right)\left(a\right).\] We denote this as $\pi^i_{k}(a|o)$ for simplicity when appropriate. Each agent maintains a buffer (with size $M$) of collected transitions of the form $\left(o^i_t,a^i_t,\tilde{\hat{r}}^i_t,o^i_{t+1}\right)$, where $\tilde{\hat{r}}^i_t$ is $i$'s local estimate of the global average reward obtained by running Alg. \ref{average_reward_alg} (Line \ref{get_average_reward_estimate}). At each iteration $k$, agents empty their buffer ({Line} \ref{empty_buffer}) before collecting $M$ new transitions in the environment ({Lines} \ref{start_m_loop}-\ref{end_m_loop}). Each decentralised agent then trains its Q-network $\check{Q}_{\theta^i_k}$ via $L$ updates ({Lines} \ref{beginlearningwithbuffer}-\ref{endlearningwithbuffer}) as follows. 

For stability, $i$ also maintains a target network $\check{Q}_{\theta^{i,'}_{k,l}}$ with the same architecture but parameters $\theta^{i,'}_{k,l}$ copied from $\theta^{i}_{k,l}$ less regularly than $\theta^{i}_{k,l}$ themselves are updated, i.e. only every $\nu$ learning iterations ({Line} \ref{update_target_line}).  At each iteration $l$, the agent samples a random batch $B^i_{k,l}$ of $|B|$ transitions from its buffer ({Line} \ref{sample_from_buffer}). It then trains its Q-network using stochastic gradient descent to minimise the loss in Def \ref{loss} ({Line} \ref{minimise_loss_line}). The trained Q-network determines $i$'s updated policy (Line \ref{policy_update_line}).

\begin{definition}[Q-network empirical loss]\label{loss}
The training loss to be minimised is given by 
\begin{align*}
    \hat{\mathcal{L}}(\theta,\theta') = \frac{1}{|B|}\sum_{transition \in B^i_{k,l}}\left|\check{Q}_{\theta^i_{k,l}}(o_t,a_t) - T\right|^2,
    \end{align*}
\begin{align*}
    \text{where} \;\;\;\;T = \tilde{\hat{r}}_t + 
    \left[\tau_{q}\ln\pi_{\theta^{i,'}_{k,l}}(a_t|o_t)\right]^0_{cl} + \gamma\sum_{a\in\mathcal{A}}\pi_{\theta^{i,'}_{k,l}}(a|o_{t+1})\left(\check{Q}_{\theta^{i,'}_{k,l}}(o_{t+1},a) - \tau_{q}\ln\pi_{\theta^{i,'}_{k,l}}(a|o_{t+1})\right).\end{align*}
\end{definition}

For $cl < 0$, $[\cdot]^0_{cl}$ is a clipping function used in Munchausen RL to prevent numerical issues if the policy is too close to deterministic, as the log-policy term is otherwise unbounded \citep{NEURIPS2020_2c6a0bae,wu2024populationaware}.

\subsection{Sub-routine for communicating and refining policies}\label{policy_adoption_and_communication}

\begin{algorithm}[!t]
    \caption{Policy communication and selection} \label{policy_communication_algo}
    \begin{algorithmic}[1]
    \REQUIRE Time-dependent communication graph $\mathcal{G}^{comm}_t$, loop parameters $E, C_p$,   learning parameters
    $\gamma$, $\{\tau^{comm}_k\}_{k \in \{0,\dots,K-1\}}$
    \REQUIRE policies \{$\pi^i_{k+1}$\}$_{i=1}^{N}$; states \{$s^i_t$\}$_{i=1}^{N}$; $t$ 
    
        \STATE $\forall i :$ $\sigma^i_{k+1} \gets 0$\label{start_eval_line}
        \FOR{$e \in 0,\dots,E-1$ evaluation steps}
        \STATE $\{o^i_t\}_{i=1}^{N} \gets$ \textbf{EstimateMeanFieldAlg. \ref{alg:mean_field_estimation_specific}}$\left(\mathcal{G}^{vis}_t,\mathcal{G}^{comm}_t,\{s^i_t\}_{i=1}^{N}\right)$
        \STATE Take step $\forall i : a^i_t \sim \pi^i_{k}(\cdot|o^i_t), r^i_{t} = R(s^i_{t},a^i_{t},\hat{\mu}_t),s^i_{t+1} \sim P(\cdot|s^i_{t},a^i_{t},\hat{\mu}_t)$\label{step2}
        \STATE $\forall i :$ $\sigma^i_{k+1} \gets \sigma^i_{k+1} + \gamma^{e}\cdot r^i_{t}$ 
        \STATE $t \leftarrow t + 1$
        \ENDFOR\label{end_eval_line}
        \FOR{$C_p$ rounds}\label{start_comm}
            \STATE $\forall i :$ Broadcast $\sigma^i_{k+1}, \pi^{i}_{k+1}$\label{broadcast_line}
            \STATE $\forall i : J^i_t \leftarrow i\cup\{j \in \mathcal{N} : (i,j) \in \mathcal{E}^{comm}_{t}$\}\label{receive_line}
            \STATE $\forall i:$ Select $\mathrm{adopted}^i \sim$ Pr$\left(\mathrm{adopted}^i = j\right)$ = $\frac{\exp{(\sigma^j_{k+1}}/\tau^{comm}_k)}{\sum_{x\in J^i_t}\exp{(\sigma^x_{k+1}}/\tau^{comm}_k)}$ $\forall j \in J^i_t$ \label{softmax_adoption_prob}
            \STATE $\forall i : \sigma^i_{k+1} \leftarrow \sigma^{\mathrm{adopted}^i}_{k+1}, \pi^i_{k+1} \leftarrow \pi^{\mathrm{adopted}^i}_{k+1}$\label{adopt_line}
            \STATE $\{o^i_t\}_{i=1}^{N} \gets$ \textbf{EstimateMeanFieldAlg. \ref{alg:mean_field_estimation_specific}}$\left(\mathcal{G}^{vis}_t,\mathcal{G}^{comm}_t,\{s^i_t\}_{i=1}^{N}\right)$
            \STATE Take step $\forall i : a^i_t \sim \pi^i_{k}(\cdot|o^i_t), r^i_{t} = R(s^i_{t},a^i_{t},\hat{\mu}_t),s^i_{t+1} \sim P(\cdot|s^i_{t},a^i_{t},\hat{\mu}_t)$; $t \leftarrow t + 1$\label{step3}
            \ENDFOR\label{end_comm}
     \STATE \textbf{return} (policies \{$\pi^i_{k+1}$\}$_{i=1}^{N}$, states \{$s^i_t$\}$_{i=1}^{N}$, $t$)
    \end{algorithmic}
\end{algorithm}

Alg. \ref{policy_communication_algo} (based on Alg. 1 in \citet{benjamin2024networkedapproximation} for MFGs) uses the communication network $\mathcal{G}^{comm}_t$ to spread policy updates that are estimated to be better performing through the population, allowing faster learning than in the independent and central-agent cases.

Alg. \ref{policy_communication_algo} is run after agents have independently updated their policies according to their newly trained Q-networks at each iteration $k$ of the main learning algorithm (Line \ref{run_CommunicatePolicyAlg_line}, Alg. \ref{main_algo}). In Alg. \ref{policy_communication_algo}, agents obtain an approximation of their \textit{individual} discounted expected return $\{{V}^i(\boldsymbol{\pi},\mu_t)\}_{i=1}^{N}$ (Def. \ref{Discounted_individual_expected_return}, i.e. \textit{not} the population-average return, which would not give differentiation between the different updated policies). They do so by collecting individual rewards for $E$ steps (not added to the training buffer), and calculating the discounted sum of rewards over these finite steps, setting this value to $\sigma^i_{k+1}$ ({Lines} \ref{start_eval_line}-\ref{end_eval_line}). We can characterise this approximation of the infinite-step return as \{$\sigma^i_{k+1}$\}$_{i=1}^{N}$ = $\{\hat{V}^i(\boldsymbol{\pi}_{k+1},\mu_t; E)\}_{i=1}^{N}$. 

They then broadcast their Q-network parameters along with $\sigma^i_{k+1}$ ({Line} \ref{broadcast_line}). Receiving these from their neighbours $J^i_t$ on the network, agents select which set of parameters to adopt by taking a softmax over their own and the received estimate values $\sigma^j_{k+1}$ $\forall j \in J^i_t$, defined as follows ({Lines} \ref{receive_line}-\ref{adopt_line}): \[\mathrm{adopted}^i \sim \text{Pr}\left(\mathrm{adopted}^i = j\right) = \frac{\exp{(\sigma^j_{k+1}}/\tau^{comm}_k)}{\sum_{x\in J^i_t}\exp{(\sigma^x_{k+1}}/\tau^{comm}_k)}.\] They repeat this broadcast and adoption process for $C_p$ rounds (distinct from the $C_r$/$C_e$ communication rounds for the other sub-routines). 

\subsection{Sub-routine for networked estimation of global empirical mean-field}\label{mean_field_estimate_section}

Networked agents use Alg. \ref{alg:mean_field_estimation_specific} (this is Alg. 3 from \citet{benjamin2024networkedapproximation} for the {MFG} setting) to locally estimate the global empirical mean field, to serve as an observation input for their Q-/policy-networks. Recall that we include this added observation and sub-routine for generality, especially for non-stationary problems. However it is often not necessary, particularly in stationary problems like those in our experiments, where agents can find the social optimum while only observing $o^i_t = s^i_t$, and therefore would not need to estimate the mean field.

Alg. \ref{alg:mean_field_estimation_specific} involves agents using the visibility graph $\mathcal{G}^{vis}_t$ to count the number of agents in locations that fall within the visibility radius ({Line} \ref{visible_line}). For $C_e$ communication rounds, agents can supplement this local count with those received from neighbours over the communication network $\mathcal{G}^{comm}_t$, in order to count agents that do not fall within the visibility radius ({Lines} \ref{begin_MF_count_comm}-\ref{end_MF_count_comm}). We assume agents know the population's total size $N$, and therefore can distribute the uncounted agents uniformly over the states that remain unaccounted for after the communication rounds ({Lines} \ref{get_counted}-\ref{get_uncounted_states}). Agents now have a vector containing a true or estimated count for every state; this is converted to an estimated empirical mean field by dividing all counts by $N$ ({Lines} \ref{convert_counted}-\ref{distribute_uncounted}).

\begin{algorithm}[t]
\caption{Mean-field estimation and communication for environments with $\mathcal{G}^{vis}_t$}
\label{alg:mean_field_estimation_specific}
\begin{algorithmic}[1]
\REQUIRE Time-dependent visibility graph $\mathcal{G}^{vis}_t$, time-dependent communication graph $\mathcal{G}^{comm}_t$, states $\{s^i_t\}_{i=1}^{N}$, number of communication rounds $C_e$
\STATE $\forall i,s :$ Initialise count vector $\hat{\upsilon}^i_{t,1}[s]$ with $\emptyset$
\STATE $\forall i$, $\forall s' \in \mathcal{S'} : (s^i_t, s') \in \mathcal{E}^{vis}_t$ : $\hat{\upsilon}^i_{t,1}[s']$ $\leftarrow$ $\sum_{j \in \{1,\dots,N\} : s^j_t = s'} 1$\label{visible_line}
\FOR{$c_e$ $\in$ $1,\dots,C_e$}\label{begin_MF_count_comm}
    \STATE $\forall i :$ Broadcast $\hat{\upsilon}^i_{t,c_e}$
    \STATE $\forall i : J^i_t \leftarrow i \cup \{j \in \mathcal{N} : (i,j) \in \mathcal{E}^{comm}_t\}$
    \STATE $\forall i,s :$ Initialise new count vector $\hat{\upsilon}^i_{t,(c_e+1)}[s]$ with $\emptyset$
    \STATE $\forall i,s$ and $\forall j \in J^i_t : \hat{\upsilon}^i_{t,(c_e+1)}[s] \leftarrow \hat{\upsilon}^j_{t,c_e}[s]$ if $\hat{\upsilon}^j_{t,c_e}[s] \neq \emptyset$
\ENDFOR\label{end_MF_count_comm}
\STATE $\forall i :$ $counted\_agents^i_t \leftarrow \sum_{s \in \mathcal{S} : \hat{\upsilon}^i_t[s] \neq \emptyset} \hat{\upsilon}^i_t[s]$\label{get_counted}
\STATE $\forall i :$ $uncounted\_agents^i_t \leftarrow N - counted\_agents^i_t$
\STATE $\forall i :$ $unseen\_states^i_t \leftarrow \sum_{s \in \mathcal{S} : \hat{\upsilon}^i_t[s] = \emptyset} 1$\label{get_uncounted_states}
\STATE \label{convert_counted} $\forall i,s$ where $\hat{\upsilon}^i_t[s]$ is not $\emptyset$ : $\tilde{\hat{\mu}}^i_t[s] \leftarrow \frac{\hat{\upsilon}^i_t[s]}{N}$ 
\STATE \label{distribute_uncounted} $\forall i,s$ where $\hat{\upsilon}^i_t[s]$ is $\emptyset$ : $\tilde{\hat{\mu}}^i_t[s] \leftarrow \frac{uncounted\_agents^i_t}{N \times unobserved\_states^i_t}$
\STATE \textbf{return}  $\{$(states $s^i_t$, mean-field estimates $\tilde{\hat{\mu}}^i_t$)$\}_{i=1}^{N}$
\end{algorithmic}
\end{algorithm}

\section{Theoretical results}\label{theory_section}

\subsection{Introduction}

We follow the definitions of the central-agent and independent-learning architectures from closely related works that learn MFGs online from a non-episodic run of the empirical system \citep{policy_mirror_independent,benjamin2024networked,benjamin2024networkedapproximation}; both architectures can each be seen as special cases of our networked algorithm:  

\begin{itemize}
    \item In the \textbf{central-agent} case, only arbitrary central agent $i=1$ updates a Q-network and automatically pushes this to all other agents in place of the decentralised policy communication in Line \ref{run_CommunicatePolicyAlg_line} of Alg. \ref{main_algo}. Additionally, the true global mean-field distribution and average reward are always used in place of the local estimates, i.e. $\tilde{\hat{\mu}}^i_t$ = ${\hat{\mu}}_t$ and $\tilde{\hat{r}}^i_t = \hat{r}$.
    \item In the \textbf{independent} case, there are never any links in $\mathcal{G}^{comm}_t$ or $\mathcal{G}^{vis}_t$, i.e. $\mathcal{E}^{comm}_t = \mathcal{E}^{vis}_t = \emptyset$. 
\end{itemize}

We prove theoretically that the policy communication and adoption scheme allows networked agents to increase their returns faster than these alternatives (with the central-agent paradigm being potentially unrealistic and vulnerable in any case). Rem. \ref{intuitive_explanation} suggests informal reasons for our formal results to aid intuitive understanding. 

\vspace{0.1cm}

\begin{remark}\label{intuitive_explanation}
Like many cooperative learning paradigms, both the independent {and} central-agent alternatives to our networked architecture may suffer from the credit-assignment problem, in that it is not clear how agents' local state $s^i_t$ and local action $a^i_t$ contributed to the (locally estimated) \textit{average} reward $\tilde{\hat{r}}^i_t$ \citep{10777009,cazenille2025signalling}. Agents may receive low individual reward $r^i_t$ by taking action $a^i_t$ given $o^i_t$, but would nevertheless learn that doing so was `good' if the rest of the population took highly rewarded actions at the same step giving high average reward $\tilde{\hat{r}}^i_t$. By drawing spurious relations, an agent's updated policy $\pi^i_{k+1}(a|o)$ may negatively impact (or simply not advance) the goal of maximising social welfare. Including the (estimated) empirical mean field in the observation $o^i_t = (s^i_t, \tilde{\hat{\mu}}^i_t)$ might mitigate this slightly by indicating which mean fields gave high average rewards. However, this does not solve the issue of allowing learners to distinguish between helpful or unhelpful local actions $a^i_t$, whether those learners are centralised or not, since actions can affect rewards in ways other than simply by helping to reach a certain mean field. By updating policies with respect to average return but then  spreading updates through the population which are estimated to give a higher \textit{individual} return, despite this being a cooperative problem, we reduce the credit-assignment problem by replicating updated policies that should contribute positively to the population-average return, and filtering out those that do not.

Moreover, even if we assumed credit assignment were not a problem, there is randomness in the Q-network update: agents have stochastic policies and thus may collect a wide variety of transitions to add to their individual buffers, from which they sample randomly when training Q-networks. There may therefore be considerable variance in the quality of their estimated Q-functions, leading in turn to variance in the quality of policy updates. At each iteration of the central-agent algorithm, in \textit{expectation} the central learner will by definition have an average-quality update, and its updated policy will be pushed to the entire population whether or not it performs well, also giving large variance in the quality of updates. Our decentralised networked approach permits beneficial parallelisation in place of this single-learner method, by generating a whole population of possible updates, from which the one(s) estimated to be best-performing can be selected via a process akin to the comparison of fitness functions in evolutionary algorithms. These are then spread around the population, biasing networked populations towards better performing updates.
\end{remark}
\vspace{0.25cm}


We give the theoretical analysis separately for two important subclasses of cooperative game usually found in MFC, which have different reward structures and therefore can incentivise different population behaviour: 
\begin{enumerate}
    \item \textit{coordination games}, where the social welfare is increased by agents aligning their strategies, such as in consensus/synchronisation/rendezvous tasks;
    \item \textit{anti-coordination games}, where the social welfare is increased by the population exhibiting diverse strategies, such as in exploration, coverage or task allocation games.
\end{enumerate}

These subclasses cover a large proportion of cooperative objectives in anonymous, symmetric settings with large populations. We emphasise that the fact that agents would in principle benefit from having diverse policies in anti-coordination games does not contradict the classical MFC framework that simplifies the infinite population problem by finding the single policy to be shared by all agents. In the symmetric (i.e. identical reward and transition functions) MFC limit, an optimal solution can be realised by having the infinite agents all follow the single socially optimal policy, even for reward functions that favour diversity. A very large number of works on both MFC and MFGs conduct experiments on anti-coordination games, particularly dispersal and exploration tasks, despite assuming that the population follows a shared single policy learnt by a central node \citep{doi:10.1073/pnas.1922204117, survey_learningMFGs, lee2024mean}. 
We make the distinction between coordination and anti-coordination games to aid theoretical analysis of our decentralised policy adoption scheme compared with entirely independent learning: while it is intuitive that adopting independently-updated policies from neighbours via the communication scheme would be beneficial in coordination games, we also show theoretically and empirically that the adoption scheme provides a benefit in anti-coordination games, though this requires separate analysis.


To define the two types of game, we first introduce the following functions. $\mathbb{I}[\cdot]$ is the indicator function, which equals 1 if the condition inside is true and 0 otherwise. $b: \Pi \to \mathbb{R}_{\geq 0}$ is a \textit{base return function} that quantifies a policy’s inherent ability to receive rewards regardless of how many other agents follow the same strategy. For example, if agents are rewarded for agreeing on one of a number of targets at which to meet, then policies that visit none of the designated targets will have lower returns than those that do, whether agents are aligned or not.  $f_c: \mathbb{N} \to \mathbb{R}_{> 0}$ (resp. $f_d: \mathbb{N} \to \mathbb{R}_{> 0}$) is a \textit{coordination (resp. anti-coordination) scaling function}. It has minimum $f_c(1) > 0$ (resp. $f_d(0) > 0$), and increases monotonically with the number of agents whose policies match (resp. are different from) $i$'s. 
 

\begin{definition}[Coordination game]\label{coordination_game_definiton}
The agents' return can be decomposed as follows, $\forall i,j \in \{1,\dots,N\}$:
$V^i(\boldsymbol{\pi},\mu_{\bar{t}}) = h\left(b(\pi^i),f_c\left(\sum_{j \in \{1,\dots,N\}}\mathbb{I}\left[\pi^i = \pi^j\right]\right)\right)$, where $h: \mathbb{R}_{\geq 0} \times \mathbb{R}_{> 0} \to \mathbb{R}_{\geq 0}$ is a function that composes $b(\cdot)$ and $f_c(\cdot)$ and is monotonic in both arguments, i.e. an increase in either the policy's intrinsic ability to attain rewards, or the extent to which it is aligned with other agents' policies, gives a higher return. 
\end{definition}

\begin{definition}[Anti-coordination game]\label{anti-coordination_game_definiton}
The agents' return can be decomposed as follows, $\forall i,j \in \{1,\dots,N\}$:
$V^i(\boldsymbol{\pi},\mu_{\bar{t}}) = h\left(b(\pi^i),f_d\left(N - \sum_{j \in \{1,\dots,N\}}\mathbb{I}\left[\pi^i = \pi^j\right]\right)\right)$, where $h: \mathbb{R}_{\geq 0} \times \mathbb{R}_{> 0} \to \mathbb{R}_{\geq 0}$ is a function that composes $b(\cdot)$ and $f_d(\cdot)$ and is monotonic in both arguments, i.e. an increase in either the policy's intrinsic ability to attain rewards, or the extent to which it is different from other agents' policies, gives a higher return. 
\end{definition}

Note that in our setting, where policy parameters are directly communicated and adopted among the population, we focus on exact equality of policies for simplicity of the theory. However, these definitions could be made more general and inclusive by instead considering similarity kernels or label mappings of strategically relevant parts of policies.

\subsection{Analysis}

Our sub-routines involve time-varying networks sharing different types of information at different points in the algorithm, meaning that theoretical analysis can potentially grow complicated. We seek to simplify this analysis to make it more intuitive and useful by focusing on the benefit of the decentralised policy exchange scheme in Alg. \ref{policy_communication_algo}. This is because our ablation studies of Algs. \ref{average_reward_alg} (average reward estimation) and \ref{alg:mean_field_estimation_specific} (mean field estimation) in Sec. \ref{results_and_discussion} indicate that the policy exchange scheme is the dominant factor in driving the benefit of the networked paradigm in our experimental settings. Moreover, recall that Alg. \ref{alg:mean_field_estimation_specific} is only necessary when we allow population-dependent policies i.e. $o^i_t = (s^i_t, \tilde{\hat{\mu}}^i_t)$, whereas for stationary problems, including all those in our experiments and many others, using a mean field observation or estimation is not actually required for finding the optimal policy.

For simplicity of the theory, we make several assumptions. We explore the conditions under which these assumptions apply in practice, and discuss how even when loosening the assumptions, they still provide useful heuristic insight as to how our networked communication scheme affords benefits over the central-agent and independent-learning architectures. We do not enforce the assumptions in our experiments, and our empirical results nevertheless follow our theoretical theorems in all but some specific instances that we discuss. 







The first assumption simplifies the theory by presuming that it is only the decentralised policy communication scheme that creates a difference in learning between the networked and central-agent cases, by assuming that the estimated mean fields and average rewards are equivalent to the true ones used in the central-agent case. Note that populations with fully connected networks will in any case always be able to accurately estimate $\hat{r}$ and ${\hat{\mu}}_t$ by Algs. \ref{average_reward_alg} and \ref{alg:mean_field_estimation_specific}, even for $C_r = 1$ and $C_e=0$. This is may apply reasonably commonly in practice depending on the scenario; for example, if the network is defined by a broadcast radius (as in our experiments), then the network will be fully connected whenever that radius is at least large enough to cover the area that all the agents fall within. Moreover, as just mentioned, our ablation studies suggest that the policy communication scheme is the dominant factor in our experimental settings anyway, with the estimated mean field not required at all in the broad class of stationary problems. We leave analysis of the theoretical impact of worsening mean-field and average-reward estimations for future work.

\begin{assumption}\label{assume_all_other_estimations_being_equal}
Assume that Algs. \ref{average_reward_alg}  and \ref{alg:mean_field_estimation_specific} allow networked agents to obtain accurate estimations of the true population-average rewards and global empirical mean field respectively, i.e. $\forall i \;\tilde{\hat{\mu}}^i_t$ = ${\hat{\mu}}_t$ and $\tilde{\hat{r}}^i_t = \hat{r}$. 
\end{assumption}

Recall that at each iteration $k$ of Alg. \ref{main_algo}, after individually updating their policies in Line \ref{policy_update_line}, the population has the policies \{$\pi^i_{k+1}$\}$_{i=1}^{N}$. There is randomness in these individual policy updates, stemming from the random sampling of each agent's individually collected buffer. In Lines \ref{start_eval_line}-\ref{end_eval_line} of Alg. \ref{policy_communication_algo}, agents estimate the individual infinite-step discounted returns $\{V^i(\boldsymbol{\pi},\mu_0)\}_{i=1}^{N}$ (Def. \ref{Discounted_individual_expected_return}) of their updated policies by computing \{$\sigma^i_{k+1}$\}$_{i=1}^{N}$: the $E$-step discounted return with respect to the {empirical} mean field generated when agents follow policies \{$\pi^i_{k+1}$\}$_{i=1}^{N}$
. 

We next assume that the populations' policies are all pair-wise distinct after the updates in Line \ref{policy_update_line} and before the policy communication. This ensures that policies that are estimated to receive higher returns (and are thus adopted) are being evaluated as higher-performing due to receiving higher base returns, rather than simply because of how aligned or distinct they already happen to be with regard to other policies. This avoids scenarios where, for example, significantly suboptimal policies that are shared across multiple agents after the update (in the case of a coordination game) end up spreading through the population by communication at the expense of a more promising but less common policy, decelerating rather than accelerating improvement. In practice, this assumption is highly likely to apply in most situations in any case. Even if agents start a given iteration with identical policies, their different random seeds are likely to mean that they collect different sample transitions to add to their reinitialised buffers. Even if their buffers end up containing identical transitions, their different random seeds are likely to mean that they sample differently from their buffers, leading to slightly different updates to their policy networks.

\begin{assumption}\label{no_alignment_after_update}
Assume that directly after the policy updates in Line \ref{policy_update_line} (Alg. \ref{main_algo}), before any policy transfer as in the networked or central-agent algorithms, all policies are pair-wise distinct due to the randomness in these updates, i.e. $\forall i,j \in \{1,\dots,N\}\; \pi^i_{k+1} \neq \pi^j_{k+1}$. This means the function $f_c$ attains its minimum $f_c(1)$, and $f_d$ attains its maximum $f_d(N-1)$. 
\end{assumption}

We now assume that the finite-step estimations of the returns give sufficiently accurate comparisons between policies, so that better policies are indeed the ones that get adopted in expectation.

\begin{assumption}\label{approximation_ordering_assumption}
Assume that \{$\sigma^i_{k+1}$\}$_{i=1}^{N}$ are sufficiently good estimations so as to respect the ordering of the true infinite discounted individual returns $\{V^i(\boldsymbol{\pi}_{k+1},\mu_0)\}_{i=1}^{N}$, 
i.e. \[V^i(\boldsymbol{\pi}_{k+1},\mu_0) > V^j(\boldsymbol{\pi}_{k+1},\mu_0) \;\;\iff\;\; \sigma^i_{k+1} > \sigma^j_{k+1} \;\;\;\;\;\;\;\;\;\;\forall i,j \in \{1,\dots,N\}.\]
\end{assumption}

In practice, even if Assumption \ref{approximation_ordering_assumption} does not strictly hold, the softmax parameter $\tau^{comm}_{k}$ allows a smooth degradation as the ordering of the approximations worsens with respect to the ordering of the true values. That is, if instead of the exact correct policy ordering we have that better policies are simply \textit{more likely} to be given higher estimated evaluations, then the softmax means that these policies remain \textit{more likely} to spread, and a better policy may still be adopted even if it is not evaluated as being better.


The next assumption presumes that the networked population reaches consensus on a single policy within each $k$ iteration. We use it only in Thm. \ref{faster_learning_theorem_centralised}, and we do so to give general and intuitive comparison with the central-agent population which always shares a single policy. Incomplete consensus would give different levels of alignment/diversity, such that the relative performance of the central-agent and networked architectures might otherwise depend on the specific reward function of the task, and whether base return or alignment/diversity is more important in that reward function.

\begin{assumption}\label{single_policy_assumption}
Assume that after the $C_p$ rounds in Lines \ref{start_comm}-\ref{end_comm} (Alg. \ref{policy_communication_algo}), in which agents exchange and adopt policies from neighbours, the networked population is left with a single policy such that $\forall i,j \in \{1,\dots,N\}$ $\pi^i_{k+1} = \pi^j_{k+1}$. 
\end{assumption}

While this may sound like a strong assumption, we phrase it like this so as not to make overly strong restrictions on the communication network instead - we intentionally leave it so that Assumption \ref{single_policy_assumption} can be fulfilled in numerous ways. Most simply we can think of Assumption \ref{single_policy_assumption} holding if: 
\begin{enumerate}
    \item we set $\tau^{comm}_{k}$ close to 0 for all $k$, such that the softmax essentially becomes a max function; and 
     \item the communication network $\mathcal{G}^{comm}_t$ is static and connected during the $C_p$ communication rounds, where $C_p$ is at least as large as the network diameter $d_{\mathcal{G}^{comm}_t}$. 
\end{enumerate}

Under these conditions, previous results on max-consensus algorithms show that all agents in the network will converge on the highest value $\sigma^{max}_{k+1}$ (and hence the unique associated $\pi^{max}_{k+1}$) within a number of rounds equal to the diameter $d_{\mathcal{G}^{comm}_t}$ \citep{5348437,benjamin2024networked}. If we assumed more strongly that the network was always \textit{fully} connected, policy consensus would be achieved within a single communication round.
 
Policy consensus can be achieved even outside of these conditions, including if the network is dynamic and not connected at every step. The \textit{union} of a collection of graphs \{$\mathcal{G}_t, \mathcal{G}_{t+1}, \cdots, \mathcal{G}_{t+\omega}\}$ ($\omega \in \mathbb{N}$) is the graph with vertices and edge set equalling the union of the vertices and edge sets of the graphs in the collection \citep{1205192}. A collection is \textit{jointly connected} if its members' union is connected. Now, instead of assuming that the communication network is static and connected, we assume instead only that the sequence of networks contains one or more jointly connected collections. Then max-consensus is reached within $C_p$ if $C_p$ is large enough that the number of jointly connected collections occurring within $C_p$ is equal to the largest diameter of the union of any such collection.
  
Thus Assumption \ref{single_policy_assumption} may not hold if $C_p$ is not large enough or if parts of the population remain isolated. However, we do not enforce this assumption in our experiments, where we use $C_p = 1$ to show the benefit of even just one communication round, yet we still see networked populations significantly outperforming central-agent populations across anti-coordination games. In coordination games, while networked populations that are more connected (due to having larger communication radii) usually perform similarly to or better than central-agent populations, those that are less connected occasionally perform less well than the central-agent populations. This is probably due to Assumption \ref{single_policy_assumption} being empirically more likely to be violated in less connected populations, which in turn is more of an issue in coordination games (where consensus is more likely to be beneficial) than in anti-coordination games (where some lack of consensus does not prevent, or even helps, networked populations to outperform central-agent ones in practice).

The next assumption presumes that if a certain policy, when followed by all members of a finite population, is better than another policy when the latter is followed by all members of a finite population, then the same quality ordering will apply when members of infinite populations follow each policy. We require this in order to relate our analysis of learning in the empirical finite population back to the mean field limit when comparing with central-agent learning. Since the finite population can be arbitrarily large, and in many environments when all agents follow the same policy the finite population-average return  will converge smoothly to the infinite population social welfare, this assumption will naturally hold in many scenarios. For example, a policy that is better than another at getting a population of 500 or 5,000,000 agents to cluster in a particular location will also be better than the other policy at getting an infinite population to gather at the location. Nevertheless this order preservation is not a completely general phenomenon, and strict inequalities can vanish or reverse in the limit, especially in models with thresholds or discontinuities in the dependence of rewards or transitions on the mean field, so we state it as an explicit condition. 

\begin{assumption}\label{finite_to_infinite_ordering_assumption}
Say we have two different policies that could be shared by the whole population such that $\boldsymbol\pi^x$ = ($\pi^x,\dots,\pi^x$) and $\boldsymbol\pi^y$ = ($\pi^y,\dots,\pi^y$). We assume that:
\begin{align*}
    V^{pop}(\boldsymbol{\pi}^x,\mu_0)>V^{pop}(\boldsymbol{\pi}^y,\mu_0) \iff W(\pi^x,I(\pi^x))>W(\pi^y,I(\pi^y)).\end{align*}
\end{assumption}

We have now given all the assumptions for our first theorem. Assumption \ref{single_policy_assumption} assumes that after the $C_p$ policy exchange rounds in Lines \ref{start_comm}-\ref{end_comm} of Alg. \ref{policy_communication_algo}, the networked population is left with a single policy. Call this consensus policy $\pi^{\mathrm{net}}_{k+1}$, and its associated finitely approximated return $\sigma^{\mathrm{net}}_{k+1}$. Recall that the central-agent case is where the Q-network update of arbitrary agent $i=1$ is automatically pushed to all the others instead of the policy evaluation and exchange in Line \ref{run_CommunicatePolicyAlg_line} of Alg. \ref{main_algo}; this is equivalent to a networked case where policy consensus is reached on a \textit{random} one of the policies \{$\pi^i_{k+1}$\}$_{i=1}^{N}$. Call this policy \textit{arbitrarily} given to the whole population $\pi^{\mathrm{cent}}_{k+1}$, and its associated finitely approximated return $\sigma^{\mathrm{cent}}_{k+1}$.


We can now give our first theorem, namely that in expectation networked populations will increase their returns at least as fast as central-agent ones.

\begin{theorem}\label{faster_learning_theorem_centralised}
Let us set $\tau^{comm}_{k} \in \mathbb{R}_{> 0}$
. In coordination and anti-coordination games where Assumptions \ref{assume_all_other_estimations_being_equal}-
    \ref{finite_to_infinite_ordering_assumption} apply, we have $\mathbb{E}[W(\pi^{\mathrm{net}}_{k+1},I(\pi^{\mathrm{net}}_{k+1}))] \geq \mathbb{E}[W(\pi^{\mathrm{cent}}_{k+1},I(\pi^{\mathrm{cent}}_{k+1}))]$. 
\end{theorem}

\begin{remark}
    Assumption \ref{no_alignment_after_update} presumes that all policies are pairwise distinct after the updates, but does not restrict their returns in the same way. If we additionally make the very weak assumption that at least one of these distinct policies in each $k$ iteration has a base return that is distinct from the others (which is likely to hold in all but the most trivial environments), the inequality in the theorem above will be strict, i.e. \textit{networked learning will always be faster in expectation}.
\end{remark}


\begin{proof}
Recall that before the communication rounds in Line \ref{start_comm} (Alg. \ref{policy_communication_algo}), the  randomly updated policies \{$\pi^i_{k+1}$\}$_{i=1}^{N}$ have  associated estimated returns \{$\sigma^i_{k+1}$\}$_{i=1}^{N}$. Denote the mean and maximum of this set $\sigma^{\mathrm{mean}}_{k+1}$ and $\sigma^{\mathrm{max}}_{k+1}$ respectively. Since $\pi^{\mathrm{cent}}_{k+1}$ is chosen arbitrarily from \{$\pi^i_{k+1}$\}$_{i=1}^{N}$, it will obey $\mathbb{E}[\sigma^{\mathrm{cent}}_{k+1}] = \sigma^{\mathrm{mean}}_{k+1}$ $\forall k$, though there will be high variance. Conversely, for the networked case the softmax adoption scheme (Line \ref{softmax_adoption_prob}, Alg. \ref{policy_communication_algo}), which for $\tau^{comm}_{k} \in \mathbb{R}_{> 0}$ gives non-uniform adoption probabilities for distinct $\sigma$ values, means by definition that some communicated policies are more likely to be adopted than others if they have distinct finitely estimated returns (those with higher $\sigma^i_{k+1}$ are more likely to be adopted at each communication round). Thus the consensus $\pi^{\mathrm{net}}_{k+1}$ that gets adopted by the whole networked population will obey $\mathbb{E}[\sigma^{\mathrm{net}}_{k+1}] > \sigma^{\mathrm{mean}}_{k+1}$ if at least one policy receives a distinct return from the others, or $\mathbb{E}[\sigma^{\mathrm{net}}_{k+1}] \geq \sigma^{\mathrm{mean}}_{k+1}$ in the rare circumstance that all policies receive the same return. If $\tau^{comm}_{k+1} \rightarrow 0$, it will obey $\mathbb{E}[\sigma^{\mathrm{net}}_{k+1}] = \sigma^{\mathrm{max}}_{k+1}$ $\forall k$. As such: 
\begin{equation}\label{sigma_ordering}
    \mathbb{E}[\sigma^{\mathrm{net}}_{k+1}] \geq \mathbb{E}[\sigma^{\mathrm{cent}}_{k+1}].
\end{equation}

In Eq. \ref{sigma_ordering} and the remaining equations of the proof, bear in mind that the equality will be strict if at least one policy receives a distinct return from the others. 

Refer to the agent whose update originally gave rise to $\pi^{\mathrm{net}}_{k+1}$ and $\sigma^{\mathrm{net}}_{k+1}$ as agent $(i, \mathrm{net})$; we equivalently also have the arbitrary agent $(j, \mathrm{cent})$. Prior to consensus being attained in each case, the joint policy can be written as $\boldsymbol{\pi}^{(i, \mathrm{net}; j, \mathrm{cent})}$ $:= (\pi^1, \dots, \pi^{i-1}, \pi^{(i, \mathrm{net})}, \pi^{i+1}, \dots, \pi^{j-1}, \pi^{(j, \mathrm{cent})}, \pi^{j+1}, \dots,$ $\pi^N)$.

Given Eq. \ref{sigma_ordering}, and by Assumption \ref{approximation_ordering_assumption} on the quality of finite-step estimations, we know that directly after the policy update in Line \ref{policy_update_line} (Alg. \ref{main_algo}), \textit{prior to the consensus being reached}, we have:
\begin{equation}\label{finite_varied_ordering_equation}
\begin{aligned}
    \mathbb{E}\left[V^{(i,\mathrm{net})}(\boldsymbol{\pi}^{(i, \mathrm{net}; j, \mathrm{cent})}_{k+1},\mu_t)\right] \geq \mathbb{E}\left[V^{(j, \mathrm{cent})}(\boldsymbol{\pi}^{(i, \mathrm{net}; j, \mathrm{cent})}_{k+1},\mu_t)\right].
\end{aligned}
\end{equation}
   
We now need to show that this ordering is maintained in the case that each policy is given to the whole population. 

By  Assumption \ref{no_alignment_after_update} we know that straight after the random policy updates there is no alignment among policies, i.e. in a coordination game we have $f_c^{(i,\mathrm{net})}$ = $f_c^{(j,\mathrm{cent})} = \min f_c$, and in an anti-coordination game we have $f_d^{(i,\mathrm{net})}$ = $f_d^{(j,\mathrm{cent})} = \max f_d$. Therefore if Eq. \ref{finite_varied_ordering_equation} pertains, by Def. \ref{coordination_game_definiton} it must be because:
\begin{equation}\label{better_base_policy}
    \mathbb{E}[b(\pi^{(i,\mathrm{net})})]\geq\mathbb{E}[b(\pi^{(j,\mathrm{cent})})],
\end{equation}
i.e. because the base policy quality is higher for $\pi^{(i,\mathrm{net})}$ than for $\pi^{(j,\mathrm{cent})}$.

By Assumption \ref{single_policy_assumption} on policy consensus, we know that in the networked and central-agent cases the joint policies respectively become $\boldsymbol{\pi}^\mathrm{net} := (\pi^{\mathrm{net}}, \pi^{\mathrm{net}}, \pi^{\mathrm{net}},\dots)$ and $\boldsymbol{\pi}^\mathrm{cent} := (\pi^{\mathrm{cent}}, \pi^{\mathrm{cent}}, \pi^{\mathrm{cent}},\dots)$. We therefore end up with maximum alignment in both cases, such that $f_c^{\mathrm{net}}$ = $f_c^{\mathrm{cent}} = \max f_c$ in a coordination game, and  $f_d^{\mathrm{net}}$ = $f_d^{\mathrm{cent}} = \min f_d$ in an anti-coordination game. Due to this, along with Eqs. \ref{finite_varied_ordering_equation} and \ref{better_base_policy}, we have 
\begin{equation}\label{finite_consenus_ordering_equation}
    \mathbb{E}\left[V^i(\boldsymbol{\pi}^{\mathrm{net}}_{k+1},\mu_t)\right] \geq \mathbb{E}\left[V^j(\boldsymbol{\pi}^{\mathrm{cent}}_{k+1},\mu_t)\right].
\end{equation}
In turn we have:
\begin{equation}\label{average_consenus_ordering_equation}
    \mathbb{E}\left[V^{pop}(\boldsymbol{\pi}^{\mathrm{net}}_{k+1},\mu_t)\right] \geq \mathbb{E}\left[V^{pop}(\boldsymbol{\pi}^{\mathrm{cent}}_{k+1},\mu_t)\right],
\end{equation}
which by Assumption \ref{finite_to_infinite_ordering_assumption} gives 
\[
\mathbb{E}[W(\pi^{\mathrm{net}}_{k+1},I(\pi^{\mathrm{net}}_{k+1}))] \geq \mathbb{E}[W(\pi^{\mathrm{cent}}_{k+1},I(\pi^{\mathrm{cent}}_{k+1}))],
\]
namely the result. 
\end{proof}

We now give results showing that learning is at least as fast in the networked case than in the independent case - empirically we find networked learning always to be strictly faster. We give separate theorems for coordination and anti-coordination games. Since we cannot necessarily expect the independent agents to share a single policy $\pi_{k+1}$ after the update in each iteration of learning, we give these results in terms of the population-average return (Def. \ref{Discounted_population_average_expected_return})
instead of the single-policy social welfare (Def. \ref{social_welfare_function}) as before. 

Again, we assume for simplicity of the theory that it is only the policy communication scheme that creates a difference in learning between the networked and independent cases, i.e. we assume that networked agents receive the same estimates of the mean field and average reward as independent agents. As mentioned above, our ablation studies suggest this is the dominant factor in our experimental settings anyway, with the estimated mean field not required at all in the broad class of stationary problems. Nevertheless, in practice the networked estimates of the (mean field and) average reward will be better than the independent ones, giving an additional performance increase over the independent case. Thus loosening this assumption is likely to actually enhance the effects identified in the theorems.

\begin{assumption}\label{assume_estimations_as_bad_as_independent}
Assume that the estimated global mean field and average reward in the networked case are the same as the independent case, i.e. $\forall i,j \;\tilde{\hat{\mu}}^{(i,net)}_t$ = $\tilde{\hat{\mu}}^{(j,ind)}_t$ and $\tilde{\hat{r}}^{(i,net)}_t = r^i_t$. 
\end{assumption}

We refer to the joint policy in the networked case after communication round $c$ as $\boldsymbol\pi^{\mathrm{net}}_{k+1,c} = \left(\pi^{(1,\mathrm{net})}_{k+1,c},\dots,\pi^{(N,\mathrm{net})}_{k+1,c}\right)$, and the joint policy in the independent case as $\boldsymbol\pi^{\mathrm{ind}}_{k+1} = \left(\pi^{(1,\mathrm{ind})}_{k+1},\dots,\pi^{(N,\mathrm{ind})}_{k+1}\right)$. 

We can now give our second theorem, namely that in expectation networked populations will increase their returns at least as fast as independent ones in coordination games with only a single round of communication in each iteration.

\begin{theorem}\label{theorem_independent_coordination}
Let us again set $\tau^{comm}_{k} \in \mathbb{R}_{> 0}$
. In a coordination game, given Assumptions  
\ref{no_alignment_after_update},
\ref{approximation_ordering_assumption}
 and \ref{assume_estimations_as_bad_as_independent}, for $c=0$, $\mathbb{E}\left[V^{pop}(\boldsymbol{\pi}^{\mathrm{net}}_{k+1,c+1},\mu_t)\right] \geq \mathbb{E}\left[V^{pop}(\boldsymbol{\pi}^{\mathrm{ind}}_{k+1},\mu_t)\right]$. 
\end{theorem}

\begin{remark}
    Assumption \ref{no_alignment_after_update} presumes that all policies are pairwise distinct after the updates, but does not restrict their returns in the same way. If we additionally make the very weak assumption that at least one of the distinct policies in each $k$ iteration has a distinct base return from the others (as is generally likely to be the case), the inequality in the theorem above will be strict, i.e. \textit{networked learning will always be faster in expectation}.
\end{remark}

\begin{proof}
Let us consider two scenarios. Firstly let us imagine that within the communication round, agents swap policies, but no policy drops out of the population, such that if agent $i$ adopts policy $\pi^j_{k+1}$, there exists an agent $i'$ that adopts policy $\pi^i_{k+1}$, and so on. That way we end up with the same policies in the population as before the change, but with each one possibly carried by different arbitrary agents. This is equivalent to if no communication had taken place, meaning that in this scenario $V^{pop}(\boldsymbol{\pi}^{\mathrm{net}}_{k+1,c+1},\mu_t) = V^{pop}(\boldsymbol{\pi}^{\mathrm{ind}}_{k+1},\mu_t)$.

Let us now consider an alternative scenario. The softmax adoption scheme (Line \ref{softmax_adoption_prob}, Alg. \ref{policy_communication_algo}), which for $\tau^{comm}_{k} \in \mathbb{R}_{> 0}$ gives non-uniform adoption probabilities for distinct $\sigma$ values, means by definition that some communicated policies are more likely to be adopted than others if they have distinct finitely estimated returns. Thus in expectation the number of distinct policies in the population will decrease if at least one policy has a distinct return from the others (of course, there is a possibility of this still happening even if no policy has a distinct return from the others). Let us start by saying for simplicity that during the first communication round a single $\pi^{(j,\mathrm{net})}_{k+1,c}$ is replaced by $\pi^{(i,\mathrm{net})}_{k+1,c}$, such that for $c=0$
\begin{align*}
\boldsymbol\pi^{\mathrm{net}}_{k+1,c} &= \left(\pi^{(1,\mathrm{net})}_{k+1,c},\dots,\pi^{(\textbf{i},\mathrm{net})}_{k+1,c},\dots, \pi^{(\textbf{j},\mathrm{net})}_{k+1,c}, \dots \pi^{(N,\mathrm{net})}_{k+1,c}\right),\\
    \text{and} \quad \boldsymbol\pi^{\mathrm{net}}_{k+1,c+1} &= \left(\pi^{(1,\mathrm{net})}_{k+1,c+1},\dots,\pi^{(\textbf{i},\mathrm{net})}_{k+1,c+1},\dots, \pi^{(\textbf{i},\mathrm{net})}_{k+1,c+1}, \dots \pi^{(N,\mathrm{net})}_{k+1,c+1}\right).\end{align*}

For this to have occurred, we know that \[\mathbb{E}[\sigma^{(i,\mathrm{net})}_{k+1,c}] > \mathbb{E}[\sigma^{(j,\mathrm{net})}_{k+1,c}],\] and therefore by Assumption \ref{approximation_ordering_assumption} that 
\begin{equation}\label{return_order_coord_ind}
    \mathbb{E}\left[V^{(i,\mathrm{net})}(\boldsymbol\pi^{\mathrm{net}}_{k+1,c},\mu_t)\right] > \mathbb{E}\left[V^{(j, \mathrm{net})}(\boldsymbol\pi^{\mathrm{net}}_{k+1,c},\mu_t)\right].
\end{equation}

By  Assumption \ref{no_alignment_after_update} we know that straight after the random policy updates there is no alignment among policies, i.e. in a coordination game we have $f_c^{(i,\mathrm{net})}$ = $f_c^{(j,\mathrm{net})} = \min f_c$. Therefore if Eq. \ref{return_order_coord_ind} pertains, by Def. \ref{coordination_game_definiton} it must be because:
\begin{equation}\label{better_base_policy_coord_ind}
    \mathbb{E}[b(\pi^{(i,\mathrm{net})})]>\mathbb{E}[b(\pi^{(j,\mathrm{net})})],
\end{equation}
i.e. because the base policy quality is higher for $\pi^{(i,\mathrm{net})}$ than for $\pi^{(j,\mathrm{net})}$. For this reason we have, for $c=0$:
\begin{equation}\label{better_vpop}\mathbb{E}\left[V^{pop}(\boldsymbol{\pi}^{\mathrm{net}}_{k+1,c+1},\mu_t)\right] > \mathbb{E}\left[V^{pop}(\boldsymbol{\pi}^{\mathrm{net}}_{k+1,c},\mu_t)\right].\end{equation}

Additionally, replacing $\pi^{(j,\mathrm{net})}_{k+1,c}$ with a second copy of $\pi^{(i,\mathrm{net})}_{k+1,c}$ will increase the alignment ($f_c$) of $\pi^{(i,\mathrm{net})}_{k+1,c}$ such that $\mathbb{E}\left[V^{(i,\mathrm{net})}(\boldsymbol\pi^{\mathrm{net}}_{k+1,c+1},\mu_t)\right] > \mathbb{E}\left[V^{(i, \mathrm{net})}(\boldsymbol\pi^{\mathrm{net}}_{k+1,c},\mu_t)\right]$,  increasing the improvement even further. This effect is even greater if more than one policy is replaced. 

Since the independent case is equivalent to the networked case when $C_p = 0$, we can say that $\boldsymbol\pi^{\mathrm{ind}}_{k+1} = \boldsymbol\pi^{\mathrm{net}}_{k+1,0}$. This gives the result, i.e. \[\mathbb{E}\left[V^{pop}(\boldsymbol{\pi}^{\mathrm{net}}_{k+1,c+1},\mu_t)\right] \geq \mathbb{E}\left[V^{pop}(\boldsymbol{\pi}^{\mathrm{ind}}_{k+1},\mu_t)\right],\]
where this inequality is only not strict if the first scenario always applies rather than the second.
\end{proof}

To prove the benefit of the networked case over the independent case in anti-coordination games, we use a final additional assumption. 
This presumes that the base return is not yet fully maximised, and that the benefit to an agent's overall return of increasing its base return by adopting a neighbour's better-performing policy outweighs the resulting decrease in diversity. This establishes the conditions under which our policy adoption scheme is able to advantage networked agents over those whose policies are always independent. This assumption applies in most non-trivial scenarios (at least at the beginning of training), namely where the goal of the task is not simply for agents to have distinct policies that are otherwise inconsequential, and thus where the benefit of diverse behaviour can only be fully felt once agents have a certain level of aptitude at accomplishing the given task. For example, in all of the anti-coordination games in our experiments, agents are always penalised for moving, and only start to receive higher rewards if they are stationary. Therefore in these anti-coordination games agents will receive higher returns by \textit{aligning} on policies that prioritise stationarity, than by maintaining diverse policies that have high levels of movement. Of course once base return is maximised and the assumption no longer holds, one can consider terminating policy communication and adoption to avoid decreases in diversity (one may also be ready to stop training entirely at this point, as the population is likely to be reaching the optimal average return). Please see Sec. \ref{results_and_discussion} for further discussion of the applicability of this assumption in practice.

\begin{assumption}\label{base_outweighs_anti_coordination_assumption}
Assume that an increase in the base return function outweighs a decrease in the policy diversity, namely $h(b + \Delta b, f_d - \Delta f_d) > h(b, f_d), \; \forall \Delta b > 0, \Delta f_d > 0$, and that the agents have not yet maximised their base return function i.e. $b(\pi^{i}_{k+1}) < \sup_{\pi\in\Pi}b(\pi)\;\; \forall i\in\{1,\dots,N\}$.
\end{assumption}

We now give our final theorem, namely that in anti-coordination games, in expectation networked populations will increase their returns at least as fast as independent ones with only a single round of communication in each iteration.

\begin{theorem}\label{theorem_independent_anti-coordination}
Let us once again set $\tau^{comm}_{k} \in \mathbb{R}_{> 0}$
. In an anti-coordination game, given Assumptions  
\ref{no_alignment_after_update},
\ref{approximation_ordering_assumption},
\ref{assume_estimations_as_bad_as_independent}
and \ref{base_outweighs_anti_coordination_assumption}, for $c=0$, $\mathbb{E}\left[V^{pop}(\boldsymbol{\pi}^{\mathrm{net}}_{k+1,c+1},\mu_t)\right] \geq \mathbb{E}\left[V^{pop}(\boldsymbol{\pi}^{\mathrm{ind}}_{k+1},\mu_t)\right]$.
\end{theorem}

\begin{proof}
The proof begins similarly to that for a coordination game. Let us consider two scenarios. Firstly let us imagine that within the communication round, agents swap policies, but no policy drops out of the population, such that if agent $i$ adopts policy $\pi^j_{k+1}$, there exists an agent $i'$ that adopts policy $\pi^i_{k+1}$, and so on. That way we end up with the same policies in the population as before the change, but with each one possibly carried by different arbitrary agents. This is equivalent to if no communication had taken place, meaning that in this scenario $V^{pop}(\boldsymbol{\pi}^{\mathrm{net}}_{k+1,c+1},\mu_t) = V^{pop}(\boldsymbol{\pi}^{\mathrm{ind}}_{k+1},\mu_t)$.

Let us now consider an alternative scenario. The softmax adoption scheme (Line \ref{softmax_adoption_prob}, Alg. \ref{policy_communication_algo}), which for $\tau^{comm}_{k} \in \mathbb{R}_{> 0}$ gives non-uniform adoption probabilities for distinct $\sigma$ values, means by definition that some communicated policies are more likely to be adopted than others if they have distinct finitely estimated returns. Thus in expectation the number of distinct policies in the population will decrease if at least one policy has a distinct return from the others. Say for simplicity that during the first communication round a $\pi^{(j,\mathrm{net})}_{k+1,c}$ is replaced by $\pi^{(i,\mathrm{net})}_{k+1,c}$, such that for $c=0$
\begin{align*}\boldsymbol\pi^{\mathrm{net}}_{k+1,c} &= \left(\pi^{(1,\mathrm{net})}_{k+1,c},\dots,\pi^{(\textbf{i},\mathrm{net})}_{k+1,c},\dots, \pi^{(\textbf{j},\mathrm{net})}_{k+1,c}, \dots \pi^{(N,\mathrm{net})}_{k+1,c}\right),\\
    \text{and} \quad \boldsymbol\pi^{\mathrm{net}}_{k+1,c+1} &= \left(\pi^{(1,\mathrm{net})}_{k+1,c+1},\dots,\pi^{(\textbf{i},\mathrm{net})}_{k+1,c+1},\dots, \pi^{(\textbf{i},\mathrm{net})}_{k+1,c+1}, \dots \pi^{(N,\mathrm{net})}_{k+1,c+1}\right).\end{align*}
For this to have occurred, we know that \[\mathbb{E}[\sigma^{(i,\mathrm{net})}_{k+1,c}] > \mathbb{E}[\sigma^{(j,\mathrm{net})}_{k+1,c}],\] and therefore by Assumption \ref{approximation_ordering_assumption} that 
\begin{equation}\label{return_order_anticoord_ind}
    \mathbb{E}\left[V^{(i,\mathrm{net})}(\boldsymbol\pi^{\mathrm{net}}_{k+1,c},\mu_t)\right] > \mathbb{E}\left[V^{(j, \mathrm{net})}(\boldsymbol\pi^{\mathrm{net}}_{k+1,c},\mu_t)\right].
\end{equation}

By  Assumption \ref{no_alignment_after_update} we know that straight after the random policy updates there is no alignment among policies, i.e. in the anti-coordination game we have $f_d^{(i,\mathrm{net})}$ = $f_d^{(j,\mathrm{net})} = \max f_d$, while by Assumption \ref{base_outweighs_anti_coordination_assumption} we know that the agents have not yet maximised their base return function.  Therefore if Eq. \ref{return_order_anticoord_ind} pertains, by Def. \ref{coordination_game_definiton} it must be because:
\begin{equation}\label{better_base_policy_anticoord_ind}
    \mathbb{E}[b(\pi^{(i,\mathrm{net})})]>\mathbb{E}[b(\pi^{(j,\mathrm{net})})],
\end{equation}
i.e. because the base policy quality is higher for $\pi^{(i,\mathrm{net})}$ than for $\pi^{(j,\mathrm{net})}$. 

Assumption \ref{base_outweighs_anti_coordination_assumption} assumes that any increase in the base quality of the policy will outweigh the decrease in diversity that will come from having more than one agent following $\pi^{(i,\mathrm{net})}_{k+1,c+1}$. Therefore we have, for $c=0$:
\[\mathbb{E}\left[V^{pop}(\boldsymbol{\pi}^{\mathrm{net}}_{k+1,c+1},\mu_t)\right] > \mathbb{E}\left[V^{pop}(\boldsymbol{\pi}^{\mathrm{net}}_{k+1,c},\mu_t)\right].\] These steps apply similarly if more than one policy is replaced. 

Since the independent case is equivalent to the networked case when $C_p = 0$, we can say that $\boldsymbol\pi^{\mathrm{ind}}_{k+1} = \boldsymbol\pi^{\mathrm{net}}_{k+1,0}$. This gives the result, i.e. \[\mathbb{E}\left[V^{pop}(\boldsymbol{\pi}^{\mathrm{net}}_{k+1,c+1},\mu_t)\right] \geq \mathbb{E}\left[V^{pop}(\boldsymbol{\pi}^{\mathrm{ind}}_{k+1},\mu_t)\right],\]where this inequality is only not strict if the first scenario always applies rather than the second.
\end{proof}

\section{Experiments}\label{experiments_section}

\subsection{Experimental setup}\label{}

We present experiments from grid worlds, following the gold standard in similar works on MFGs and MFC \citep{survey_learningMFGs}. We give results from six tasks similar to those found in prior works, defined by the agents' reward/transition functions and relating to agents' positions relative to other agents. Two 
are coordination games and four 
are anti-coordination games, where in each case the reward function reflects a coordination/anti-coordination ($f_c/f_d$) element alongside other elements that may be crucial for receiving reward, reflected in the policies' base quality $b(\pi)$ (Sec. \ref{theory_section}). In all cases, rewards are normalised in [0,1] after they are computed. 

The two coordination games are:

\begin{itemize}
    \item \textbf{Cluster.} This game is also used in \citet{benjamin2024networked, benjamin2024networkedapproximation}. Agents are encouraged to gather together by the reward function $R(s^i_{t},a^i_{t},\hat{\mu}_t) =$ log$(\hat{\mu}_t(s^i_{t}))$. That is, agent $i$ receives a reward that is logarithmically proportional to the fraction of the population that is co-located with it at time $t$. We give the population no indication where they should cluster, agreeing this themselves over time.

    \item \textbf{Target selection.} This game is also used in \citet{benjamin2024networked, benjamin2024networkedapproximation}. Unlike in the above `cluster' game, the agents are given options of locations at which to gather, and they must reach consensus among themselves. If the agents are co-located with one of a number of specified targets $\phi \in \Phi$ (in our experiments we place one target in each of the four corners of the grid), and other agents are also at that target, they get a reward proportional to the fraction of the population found there; otherwise they receive a penalty of -1. In other words, the agents must coordinate on which of a number of mutually beneficial points will be their single gathering place. Define the magnitude of the distances between $x,y$ at $t$ as $dist_t(x,y)$. The reward function is given by $R(s^i_{t},a^i_{t},\hat{\mu}_t) = r_{targ}(r_{coord}(\hat{\mu}_t(s^i_{t})))$, where \[r_{targ}(x)=\begin{cases}
          x \quad &\text{if} \, \exists \phi \in \Phi \text{ s.t. } dist_t(s^i_{t},\phi) = 0 \\
          -1 \quad &  \text{otherwise,} \\
     \end{cases} \]
\[r_{coord}(x)=\begin{cases}
          x \quad &\text{if} \, \hat{\mu}_t(s^i_{t}) > 1/N \\
          -1 \quad &  \text{otherwise.} \\
     \end{cases}
\]
\end{itemize}

The anti-coordination games are:

\begin{itemize}
    \item \textbf{Disperse.} This game is also used in \citet{benjamin2024networkedapproximation} and is similar to the `exploration' tasks in \citet{scalable_deep,wu2024populationaware} and other MFG works. In our version agents are rewarded for being located in more sparsely populated areas but only if they are stationary, to avoid trivial random policies. The reward function is given by $R(s^i_{t},a^i_{t},\hat{\mu}_t) = r_{stationary}(-\log(\hat{\mu}_t(s^i_{t})))$, where \[r_{stationary}(x)=\begin{cases}
          x \quad &\text{if} \, a^i_{t} \text{ is `remain stationary'}  \\
          -1 \quad &  \text{otherwise.} \\
     \end{cases} \]

    \item \textbf{Target coverage.} The population is rewarded for spreading across a certain number of targets, as long as agents are stationary at the target. As in the `target selection' game, we have targets $\phi \in \Phi$, where in our experiments we place one target in each of the four corners of the grid. Again define the magnitude of the distances between $x,y$ at $t$ as $dist_t(x,y)$. The reward function is given by \[R(s^i_{t},a^i_{t},\hat{\mu}_t) = r_{stationary}\left(r_{targ}\left(-\log(\hat{\mu}_t(s^i_{t}))\right)\right),\] where $r_{stationary}$ and $r_{targ}$ are as defined above.
    
    \item \textbf{Beach bar.} Such games are very common in MFG works \citep{continuous_fictitious_play,survey_learningMFGs,cui2023multiagent,wu2024populationaware}. In our version agents are rewarded for being stationary in sparsely populated locations as close as possible to a target $\phi_b$, located in the centre of the grid. The maximum possible distance from the target is denoted \textit{maxDist}. The reward is given by  \begin{align*}
    R(s^i_{t},a^i_{t},\hat{\mu}_t) = r_{stationary}\left(maxDist - dist_t(s^i_{t},\phi_b) - \log(\hat{\mu}_t(s^i_{t}))\right),\end{align*} where $r_{stationary}$ is as defined above.
    
    \item \textbf{Shape formation.} The population is rewarded for spreading around a ring shape, accomplished by encouraging agents to be a distance of 3 (chosen arbitrarily to fit the grid) from a centre point $\phi_c$. The reward is given by  \[R(s^i_{t},a^i_{t},\hat{\mu}_t) = r_{stationary}\left(r_{ring}\left( -\log(\hat{\mu}_t(s^i_{t}))\right)\right),\] where $r_{stationary}$ is as defined above, and 

\[r_{ring}(x)=\begin{cases}
          x \quad &\text{if} \; dist_t(s^i_{t},\phi_c) = 3 \\
          -1 \quad &  \text{otherwise.} \\
     \end{cases} \]
\end{itemize}

In these spatial environments, we define both the communication network $\mathcal{G}^{comm}_t$ and the visibility graph $\mathcal{G}^{vis}_t$ by the physical distance from $i$. We show plots for various transmission radii, given as fractions of the maximum distance in the grid. Note that the networked population with the largest radius is always fully connected, and therefore these agents are always able to accurately estimate $\hat{r}$ and ${\hat{\mu}}_t$ even for $C_r = 1$ and $C_e=0$. That is, when we set $C_r = C_e>0$ their observations are equivalent to those that the central-agent population would receive, albeit that policies are updated and spread differently.

We evaluate our experiments according to a finite-step estimation of the  population-average discounted return (Def. \ref{Discounted_population_average_expected_return}) over $M$ steps within each outer $k$ loop, i.e. $\hat{V}^{pop}(\boldsymbol{\pi}_k,\mu_t; M)$. Experiments were conducted on a Linux-based machine with 2 x Intel Xeon Gold 6248 CPUs (40 physical cores, 80 threads total, 55 MiB L3 cache). We use the JAX framework to accelerate and vectorise our code. We run five trials with different random seeds for each experiment, and plot the mean and standard deviation of the mean across the seeds. Random seeds are set in our code in a fixed way dependent on the trial number to allow easy replication of experiments. Our code is included in the publicly available supplementary material for reproducibility.

\subsection{Hyperparameters}\label{hyperparameters_section}

\begin{table*}
    \centering \vspace*{-0.9cm}
    \caption{Hyperparameters}\label{Hyperparameters}
    \vspace*{-0.4cm}
    \hspace*{-1.2cm}
    \begin{tabular}{p{1.9cm}|p{1.0cm}|p{14.5cm}}\\ 
  Hyperparam.  & Value     & Comment \\ \hline
  Trials & 5  & We run 5 trials with different random seeds for each experiment. {We plot the mean and standard deviation of the mean for each metric across the seeds.}   \\ \hline
Gridsize     & 20x20 &  -\\ \hline

    Population     & 500 & We chose 500 for our demonstrations to show that our algorithm can handle large populations, indeed often larger than those demonstrated in other mean-field works, especially for grid-world environments, while also being feasible to simulate with respect to time and computation constraints \citep{mfrl_yang,rl_stationary,subramanian2022multi,pomfrl,approximately_entropy,subjective_equilibria,DMFG,cui2023multiagent,general_framework,benjamin2024networked, benjamin2024networkedapproximation,wu2024populationaware}. For example, the MFC work in \citet{carmona2019linear} uses 10 agents; the work on decentralised execution for MFC by \citet{cui2023learningdecentralisedcontrol} uses 200 agents. \\ \hline
    Number of neurons in input layer & 440 & The agent's position is represented by two concatenated one-hot vectors, indicating the agent's row and column. The mean-field distribution is a flattened vector of the same size as the grid. As such, the input size is $[(2 \times \mathrm{dimension}) + (\mathrm{dimension}^2)]$.
    \\ \hline
    Neurons per hidden layer & 256 & We draw inspiration from common rules of thumb when selecting the number of neurons in hidden layers, e.g. it should be between the number of input neurons and output neurons / it should be 2/3 the size of the input layer plus the size of the output layer / it should be a power of 2 for computational efficiency. Using these rules of thumb as rough heuristics, we select the number of neurons per hidden layer by rounding the size of the input layer down to the nearest power of 2. The layers are all fully connected.
    \\ \hline
    Hidden layers & 2 & We achieved sufficient learning speed with just 2 hidden layers, but further optimising the number of layers may lead to better results.
     \\ \hline
     Activation function &   ReLU     & This is a common choice in deep RL.

     \\ \hline
    $K$     & 150       & $K$ is chosen to be large enough to see convergence in most networked cases.\\ \hline 
    $M$     & 20       &  We tested $M$ in \{20,50,100\} and found that the lowest value was sufficient to achieve convergence while minimising training time. It may be possible to converge with even smaller choices of $M$. 
    \\ \hline 
    
     $L$     &   20    &  We tested $L$ in \{20,50,100\} and found that the lowest value was sufficient to achieve convergence while minimising training time. It may be possible to converge with even smaller choices of $L$.
     \\ \hline 
    $E$     &   20    & We tested $E$ in \{20,50,100\}, and choose the lowest value to show the benefit to convergence even from very few evaluation steps. It may be possible to reduce this value further and still achieve similar results. 
     \\\hline $C_p$     &   1 (10/50)   & As in \citet{benjamin2024networked,benjamin2024networkedapproximation}, we choose a value of 1 for most experiments to show the convergence benefits brought by even a single communication round, even in networks that may have limited connectivity. We also conduct additional studies to show the effect of additional rounds in Figs. \ref{10commfig} and \ref{50commfig}.
    \\\hline 
    $C_r$     &   1 (10/50)   & Similar to $C_p$, we choose this value to show our algorithm's ability to appropriately estimate the average reward even with only a single round, even in networks that may have limited connectivity. We conduct additional studies to show the effect of additional rounds in Figs. \ref{10commfig} and \ref{50commfig}.
    \\\hline 
    $C_e$     &   1 (10/50)   & Similar to $C_p$, we choose this value to show the ability of our algorithm to appropriately estimate the mean field even with only a single communication round, even in networks that may have limited connectivity. We also conduct additional studies to show the effect of additional rounds in Figs. \ref{10commfig} and \ref{50commfig}.
    \\\hline $\gamma$     & 0.9 & Standard choice across RL literature.     \\\hline 
    $\tau_q$     &   0.03   & We follow \citet{NEURIPS2020_2c6a0bae} and \citet{benjamin2024networkedapproximation}, which tested a range of values.\\\hline 

    $|B|$   &  32    & This is a common choice of batch size that trades off noisy updates and computational efficiency. \\\hline 

    $cl$     &   -1   & We use the same value as in \citet{NEURIPS2020_2c6a0bae} and \citet{benjamin2024networkedapproximation}.\\\hline 

    $\nu$     &  $L-1$    & We follow \citet{benjamin2024networkedapproximation}, which is similar to \citet{scalable_deep}. \\\hline

    Optimiser     &   Adam   & As in \citet{NEURIPS2020_2c6a0bae}, we use the Adam optimiser with initial learning rate 0.01. \\\hline

    $\tau^{comm}_k$    &   cf. comment    & We follow \citet{benjamin2024networkedapproximation}, where $\tau^{comm}_k$ increases linearly from 0.001 to 1 across the $K$ iterations. Further optimising this inverse annealing process may lead to better results; we provide an ablation study in Fig. \ref{tau_ablation_fig}.
    \\ \hline \hline 
    \end{tabular}    
    \end{table*}

See Table \ref{Hyperparameters} for our hyperparameter choices. We can group our hyperparameters into those controlling the size of the experiment, those controlling the size of the Q-network, those controlling the number of iterations of each loop in the algorithms and those affecting the learning/policy updates or policy adoption.

In our experiments we generally want to demonstrate that our communication-based algorithm learns faster than the central-agent and independent architectures
, even when the Q-function / mean field / average reward are poorly estimated as is likely to be the case in complex real-world scenarios. There is a similar motivation in the related works on networked communication in the MFG setting by \citet{benjamin2024networked,benjamin2024networkedapproximation}. Moreover we want to show that there is a benefit even to a small amount of communication, so that communication rounds themselves do not excessively add to time complexity. As such, we generally select hyperparameters at the lowest end of those we tested during development, to show that our algorithms are particularly successful and robust given what might otherwise be considered `undesirable' hyperparameter choices. 

\subsection{Results and discussion}\label{results_and_discussion}

\begin{figure*}[t]
    \centering
    \begin{subfigure}[b]{0.49\textwidth}
        \centering
        \includegraphics[width=\textwidth]{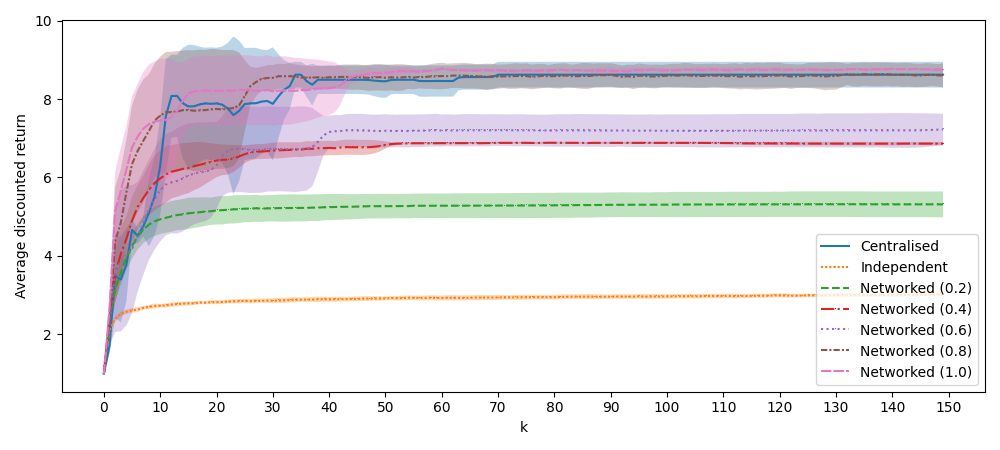}
        \caption{`Cluster' game.}
        \label{}
    \end{subfigure}
    \begin{subfigure}[b]{0.49\textwidth}
        \centering
        \includegraphics[width=\textwidth]{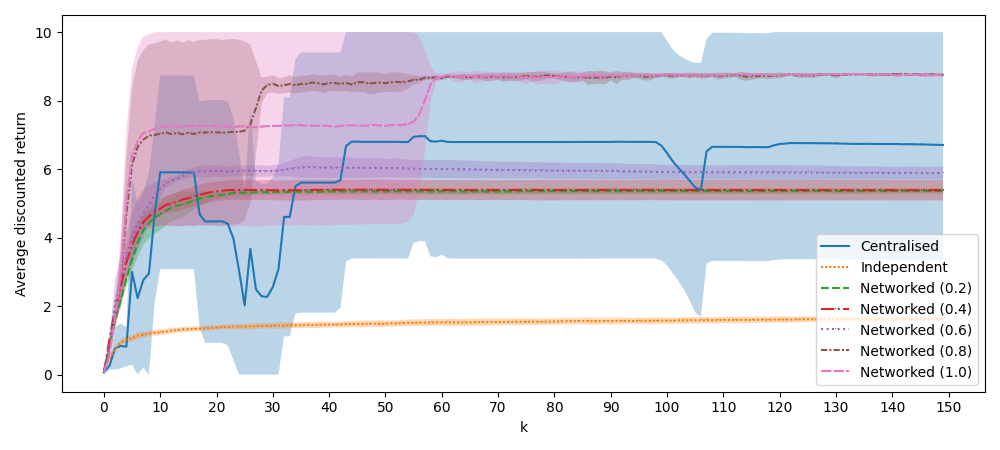}
        \caption{`Target selection' game.}
        \label{}
    \end{subfigure}
    \begin{subfigure}[b]{0.49\textwidth}
        \centering
        \includegraphics[width=\textwidth]{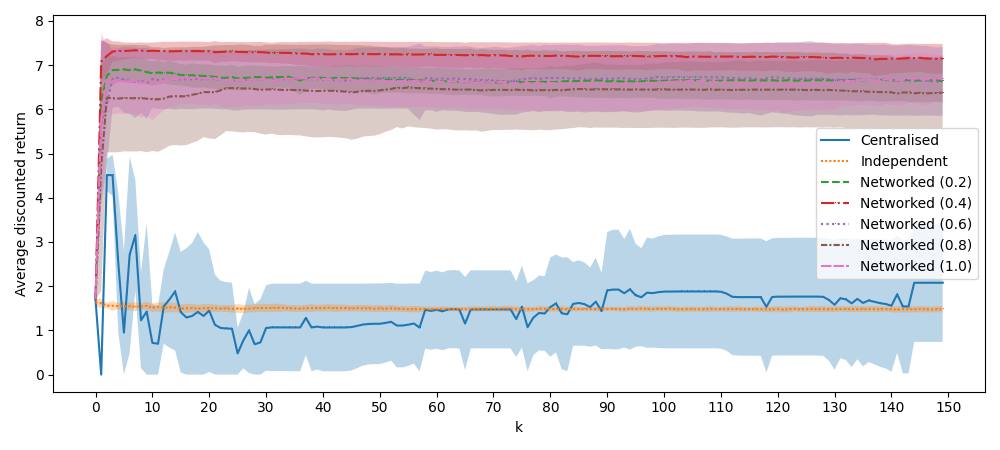}
        \caption{`Disperse' game.}
        \label{}
    \end{subfigure}
    \begin{subfigure}[b]{0.49\textwidth}
        \centering
        \includegraphics[width=\textwidth]{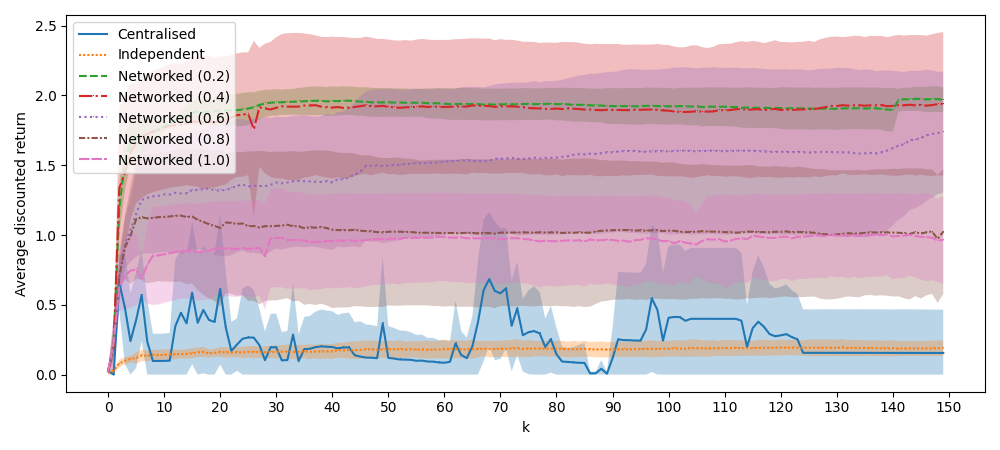}
        \caption{`Target coverage' game.}
        \label{}
    \end{subfigure}
    
    \begin{subfigure}[b]{0.49\textwidth}
        \centering
        \includegraphics[width=\textwidth]{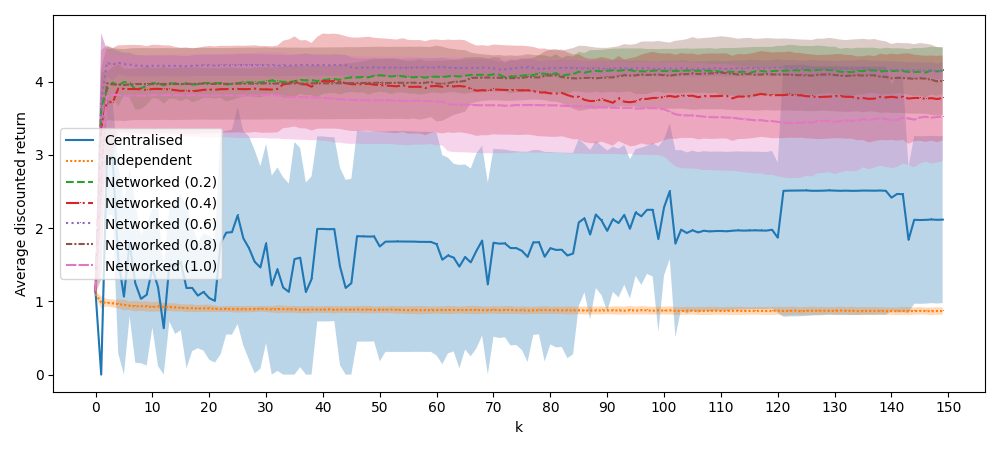}
        \caption{`Beach bar' game.}
        \label{}
    \end{subfigure}
    \begin{subfigure}[b]{0.49\textwidth}
        \centering
        \includegraphics[width=\textwidth]{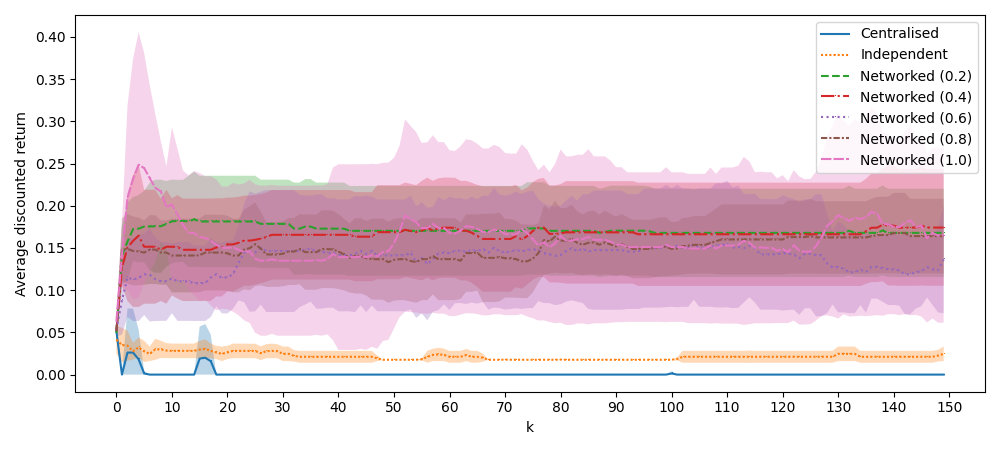}
        \caption{`Shape formation' game.}
        \label{}
    \end{subfigure}
    \caption{Standard settings with $C_e=C_r=C_p=1$. In all games networked agents of all broadcast radii significantly outperform the independent (orange) populations, and in most games they also outperform the central-agent (blue) populations, reflecting our theoretical results. The central-agent populations also have markedly higher variance than networked ones in several games, since the central learner pushes an arbitrary updated policy to the whole population regardless of its quality, leading to large fluctuations in performance, whereas our communication scheme biases networked populations towards better performing updates.} 
    \label{standard}
\end{figure*}

Fig. \ref{standard} gives results for our standard experimental settings involving 500 agents, each with their own Q-network. When networked agents communicate, they have only a \textit{single} communication round. Fig. \ref{standard} shows that in all of our games, networked populations of {all} broadcast radii significantly outperform independent (orange) agents, which hardly appear to increase their returns, if at all. Networked populations of {all} broadcast radii also significantly outperform the central-agent (blue) populations in all but the two coordination games, where only networked agents of the smaller radii (green, 0.2; red, 0.4; purple, 0.6) underperform them (due to these less connected populations being more likely to experience violations of Assumption \ref{single_policy_assumption} on policy consensus, which is a disadvantage in scenarios where alignment is beneficial). Indeed, in the anti-coordination games the central-agent populations perform similarly to purely independent ones in hardly appearing to increase their returns, performing even worse than independent agents in the `shape formation' game. The central-agent populations also have markedly higher variance than networked ones in several games (`target selection', `disperse', `beach bar'). This reflects our theoretical analysis in Sec. \ref{theory_section} that the central learner pushes an arbitrary updated policy to the whole population regardless of its quality, leading to large fluctuations in performance, whereas our communication scheme biases networked populations towards better performing updates. 




In the `target coverage' game, and sometimes the other anti-coordination games to a lesser extent, networked agents of smaller broadcast radii appear to outperform those of larger radii, i.e. the ordering is reversed from that of the coordination games, albeit not necessarily significantly so. This reflects the point up to which our Assumption \ref{base_outweighs_anti_coordination_assumption} (increase in base return outweighs decrease in diversity in anti-coordination games, and base return is not yet maximised), holds in practice, which we discuss in the following. 

The first part of Assumption \ref{base_outweighs_anti_coordination_assumption} strictly holds throughout the `disperse', `target coverage' and `shape formation' anti-coordination games: agents get no reward for diversity unless they are stationary (and also unless they are in one of the correct locations in the latter two cases). This means that any increase in base return (likelihood of being stationary or in the right location) achieved by policy adoption does indeed outweigh the loss of diversity. The first part of Assumption \ref{base_outweighs_anti_coordination_assumption} mostly holds in the `beach bar' game, apart from in a small window for agents that are stationary close to the bar target, with the window defined by the size of the population and hence the potential magnitude of the $\log(\hat{\mu}_t(s^i_{t})$ term in the reward function. Inside this window, increasing base return by moving even closer to the target, at the cost of being in a more crowded area, would not necessarily be beneficial. Regardless, in all these games the networked populations of all broadcast radii significantly outperform the independent agents, which do not appear to be able to learn at all without the helpful bias towards policies with better base returns enabled by the communication scheme. 

However, among these networked populations, the base return quickly reaches its capacity, i.e. agents learn to be primarily stationary in one of the right locations, such that the second part of Assumption \ref{base_outweighs_anti_coordination_assumption} no longer holds. This is not an issue when comparing with the independent populations, which have not maximised their base returns and therefore perform worse, but it does give rise to the reverse ordering of returns which we see among networked populations of different radii and hence connectivities. Once base return is maximised, policies that are estimated to receive higher returns may be less aligned with other policies than those other policies are with each other (at least regarding the strategically relevant parts of policies which are rewarded for greater diversity, e.g. these policies visit the less congested locations), or they simply visited the less congested locations by chance during the finite evaluation steps. Either way, more adoption of policies now becomes a disadvantage, since it reduces diversity without an additional positive impact on base return. Therefore architectures that give less communication now perform better by preserving diversity. Populations with lower broadcast radii usually have less connected networks, especially if sub-populations become isolated from each other, which is more likely in our `target coverage' game than the others since the target locations are as far apart as possible from each other. Therefore these populations have less communication than those with larger broadcast radii and so may perform better, even while all networked populations outperform the independent agents that have not maximised their base returns.

This intuition also helps to understand why networked populations outperform central-agent populations in these anti-coordination games, especially when policy consensus is not enforced for the networked populations
. The ultimate choice of consensus level might depend on whether one is using the empirical population as a practical way of learning the social optimum for a MFC problem (Def. \ref{social_optimum}), where a single policy $\pi^*$ is desired to be given to an infinite population, or whether one is solving the MFC problem to approximate the solution to a finite-agent control problem (Def. \ref{Discounted_population_average_expected_return}) involving the same number of agents as the empirical population from which one is learning. In the latter case some policy diversity may be accepted/desired if it affords a better approximation to the $N$-agent solution.

We provide numerous additional experiments and ablation studies. We list these below, but please find the full discussion of results in the caption for each figure. Of particular note, the ablation studies of Algs. \ref{average_reward_alg} (estimating global average reward) and \ref{alg:mean_field_estimation_specific} (estimating global empirical mean field) suggest that in our experimental settings the policy communication scheme (Alg. \ref{policy_communication_algo}) is the dominant factor in the better performance of networked populations over the other architectures.


\begin{itemize}
    \item Robustness to communication failure - Fig. \ref{comm_failure_fig}.
    \item Increased communication rounds - Figs. \ref{10commfig} and \ref{50commfig}.
    \item Ablation study with population-independent policies - Fig. \ref{population_independent_policies_fig}.
    \item Ablation study of Alg. \ref{alg:mean_field_estimation_specific} for estimating the empirical mean field - Fig. \ref{global_mean_field_fig}.
    \item Ablation study for observation of true/estimated average reward (agents only see their individual reward) - Fig. \ref{no_average_reward_fig}.
    \item Ablation study for Alg. \ref{average_reward_alg} for estimating the true global average reward (all agents receive true global average reward) - Fig. \ref{all_global_average_fig}.
    \item Ablation study of the choice of $\tau^{comm}_k$ - Fig. \ref{tau_ablation_fig}.
\end{itemize}

\begin{figure*}[h]
    \centering
    \begin{subfigure}[b]{0.49\textwidth}
        \centering
        \includegraphics[width=\textwidth]{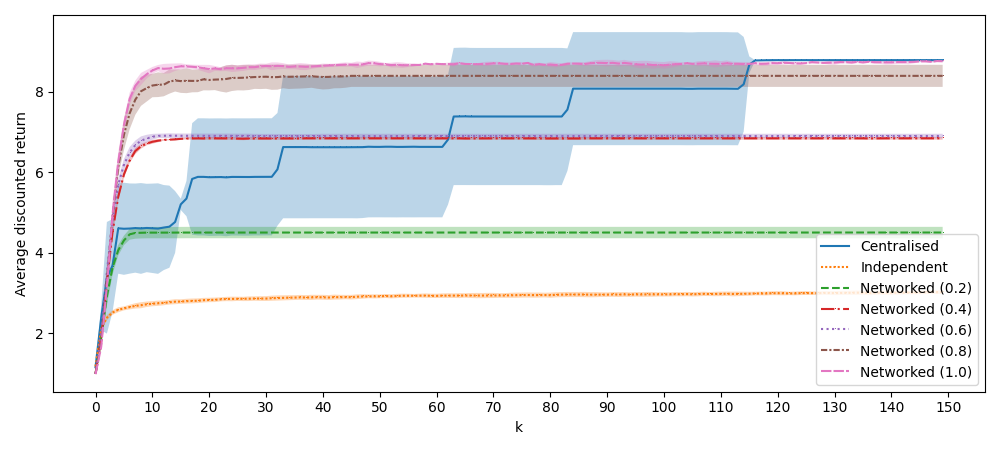}
        \caption{`Cluster' game.}
        \label{}
    \end{subfigure}
    \begin{subfigure}[b]{0.49\textwidth}
        \centering
        \includegraphics[width=\textwidth]{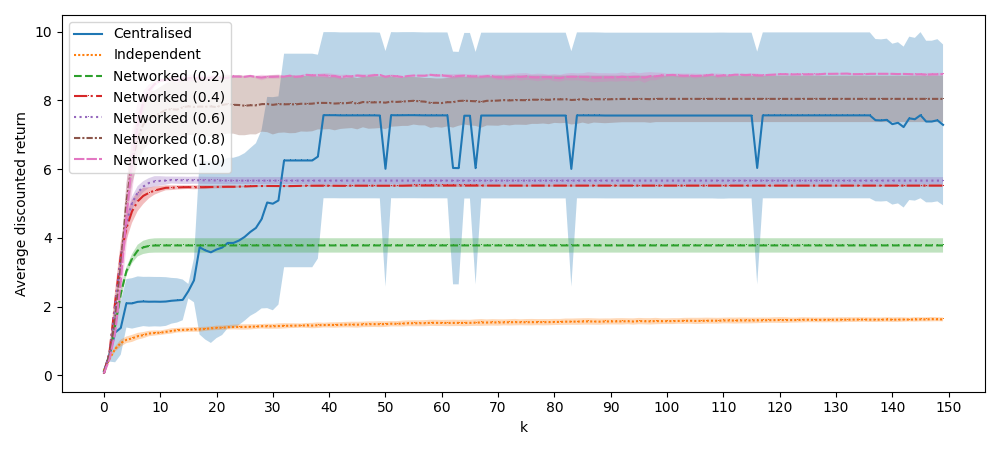}
        \caption{`Target selection' game.}
        \label{}
    \end{subfigure}
    \begin{subfigure}[b]{0.49\textwidth}
        \centering
        \includegraphics[width=\textwidth]{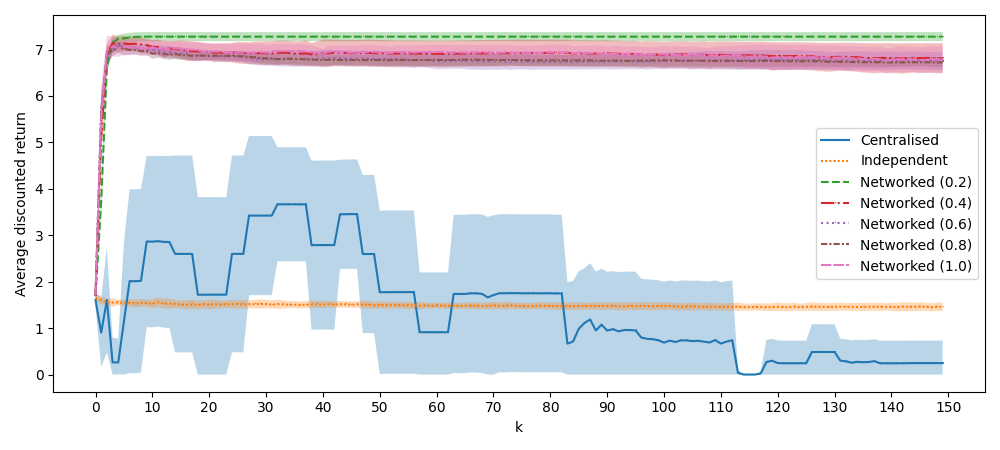}
        \caption{`Disperse' game.}
        \label{}
    \end{subfigure}
    \begin{subfigure}[b]{0.49\textwidth}
        \centering
        \includegraphics[width=\textwidth]{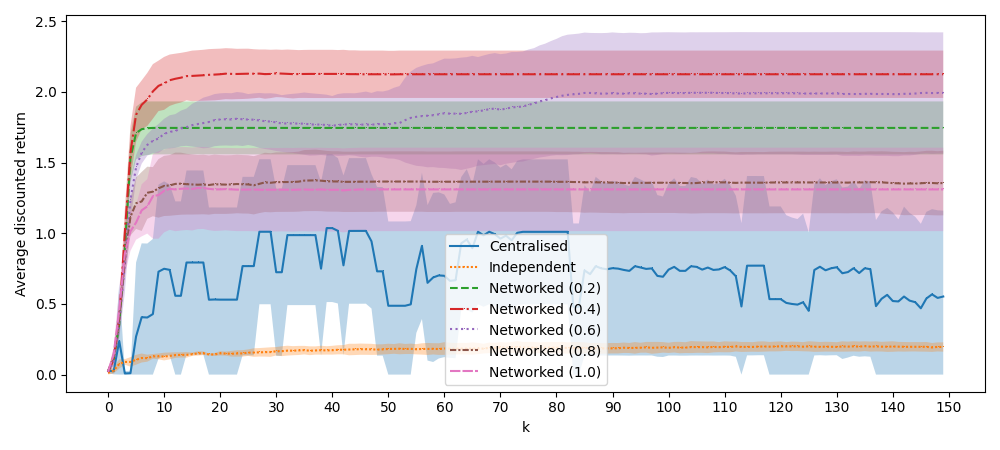}
        \caption{`Target coverage' game.}
        \label{}
    \end{subfigure}
    
    \begin{subfigure}[b]{0.49\textwidth}
        \centering
        \includegraphics[width=\textwidth]{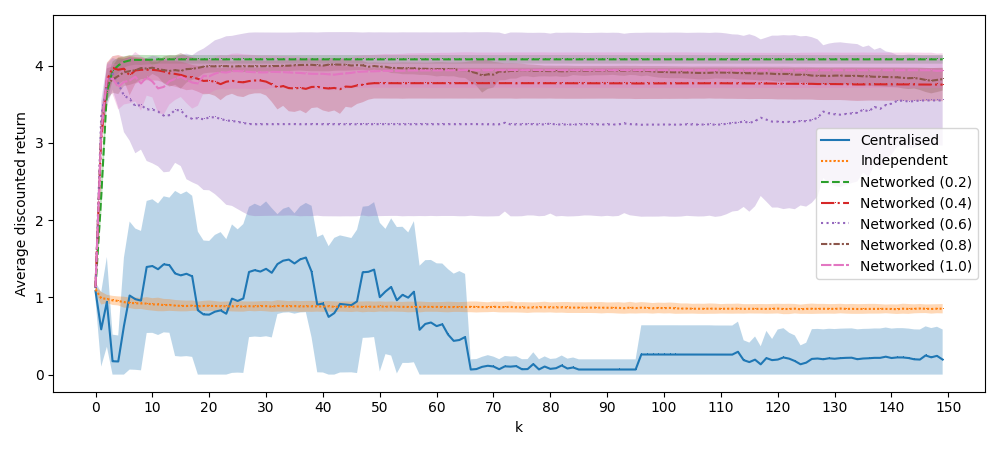}
        \caption{`Beach bar' game.}
        \label{}
    \end{subfigure}
    \begin{subfigure}[b]{0.49\textwidth}
        \centering
        \includegraphics[width=\textwidth]{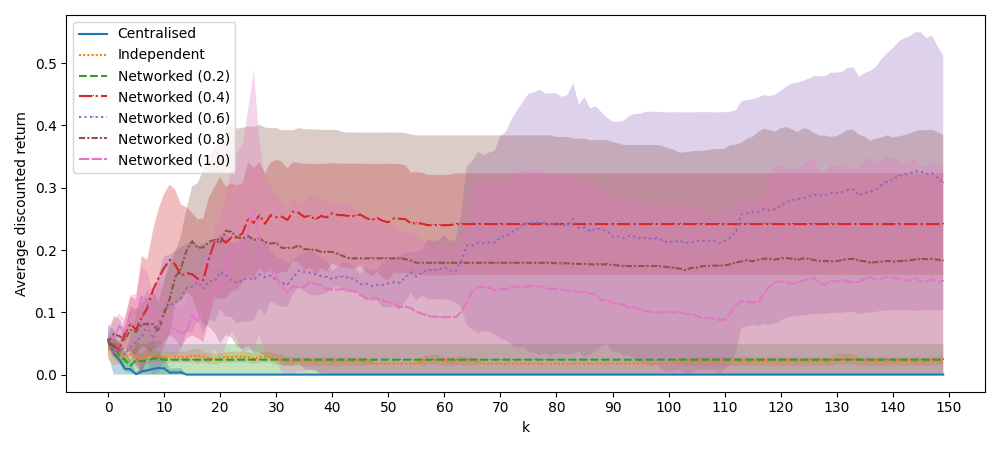}
        \caption{`Shape formation' game.}
        \label{}
    \end{subfigure}
    \caption{All communication links suffer a 90\% probability of failure, including in the central-agent case, where the link between the central learner and the rest of the population may fail. $C_e=C_r=C_p=1$. The central-agent population, which in the standard setting matched networked performance only in the `cluster' game, now learns slower even in this game, due to suffering from the single point of failure. Our networked scheme appears robust to the failures in all games, with only small differences compared to performance in the standard setting. In fact, several broadcast radii appear to perform better in the `shape formation' game with these failures than without (though not significantly so), probably because they permit greater diversity policies while still having an advantage over purely independent learners (as discussed in the body of Sec. \ref{results_and_discussion}). However, the smallest broadcast radius (green, 0.2) does drop in performance in this game, which might be expected given it now acts similarly to the independent case. Networked populations appear to have less variance in this setting than in the standard setting, at least in the first four games. This is possibly because the communication failures prevent both particularly high and particularly low performing policies from spreading fast in the population, preventing large performance fluctuations and smoothing learning progress. Meanwhile a central-agent population still has large variance even with communication failures, due to enforcing the adoption of an arbitrarily-chosen consensus policy - in some games variance is higher in this setting (though in some it may be marginally lower). This points to an additional benefit of our networked scheme over the central-agent case.}
    \label{comm_failure_fig}
\end{figure*}

\newpage

\begin{figure*}[t]
    \centering
    \begin{subfigure}[b]{0.49\textwidth}
        \centering
        \includegraphics[width=\textwidth]{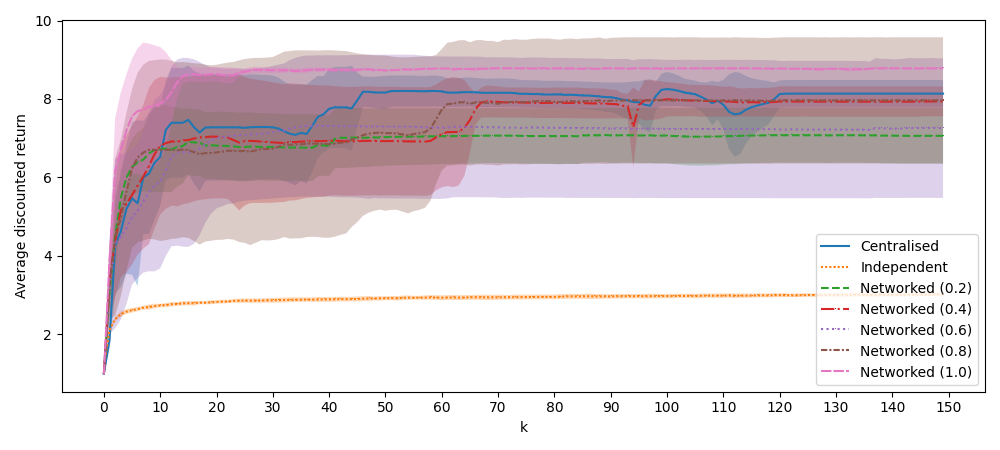}
        \caption{`Cluster' game.}
        \label{}
    \end{subfigure}
    \begin{subfigure}[b]{0.49\textwidth}
        \centering
        \includegraphics[width=\textwidth]{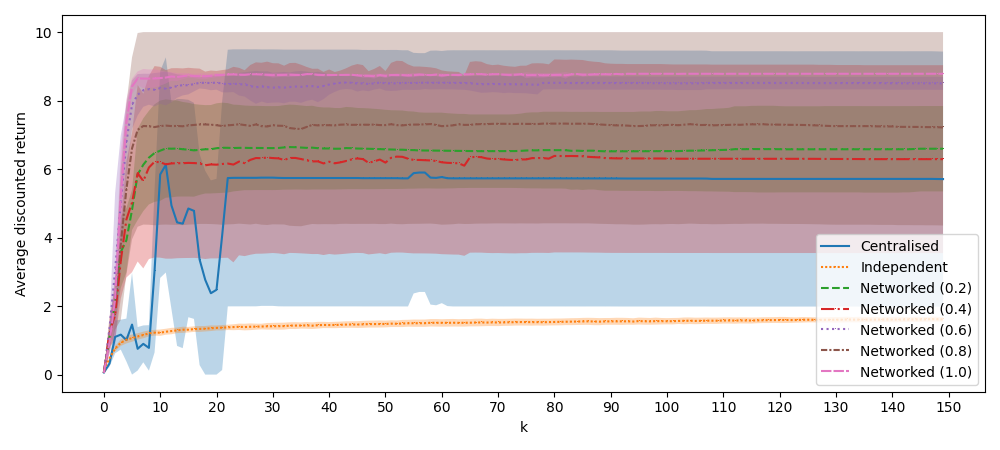}
        \caption{`Target selection' game.}
        \label{}
    \end{subfigure}
    \begin{subfigure}[b]{0.49\textwidth}
        \centering
        \includegraphics[width=\textwidth]{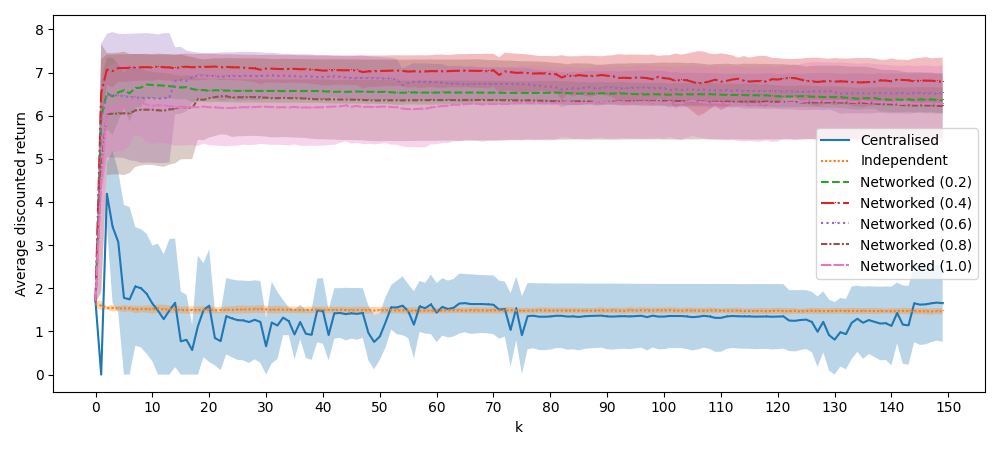}
        \caption{`Disperse' game.}
        \label{}
    \end{subfigure}
    \begin{subfigure}[b]{0.49\textwidth}
        \centering
        \includegraphics[width=\textwidth]{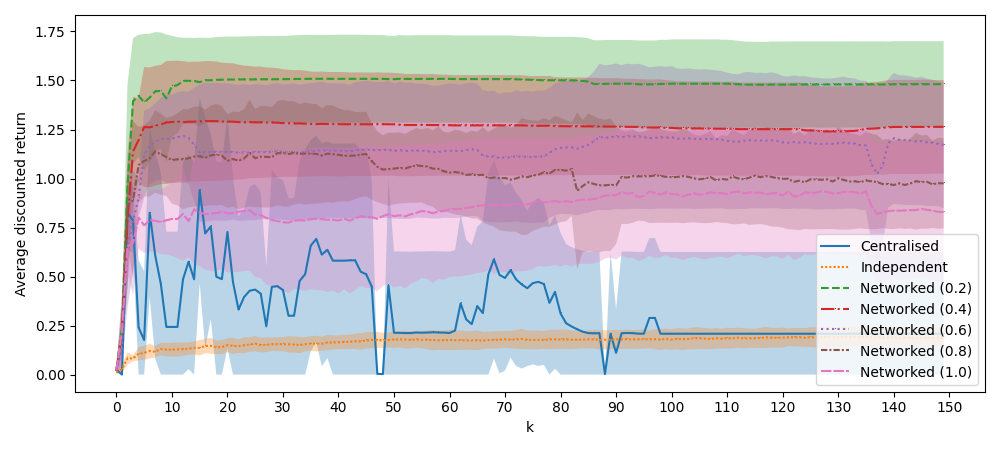}
        \caption{`Target coverage' game.}
        \label{}
    \end{subfigure}
    
    \begin{subfigure}[b]{0.49\textwidth}
        \centering
        \includegraphics[width=\textwidth]{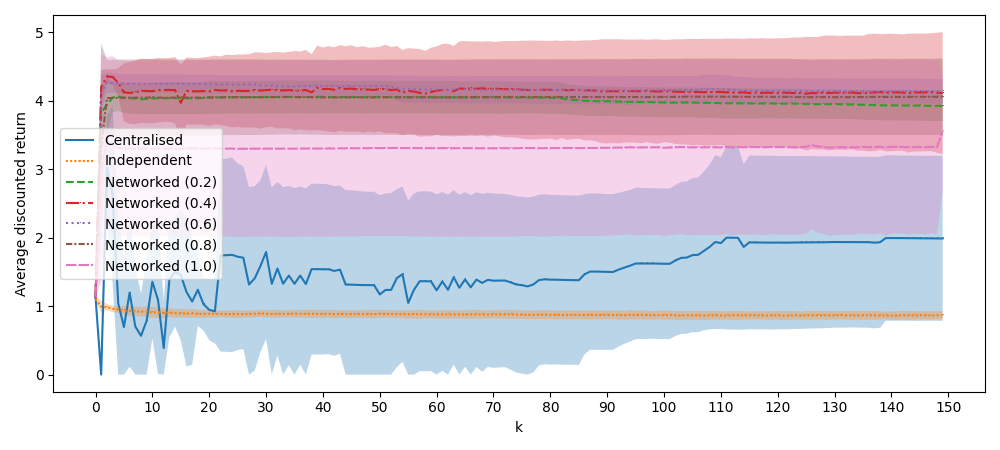}
        \caption{`Beach bar' game.}
        \label{}
    \end{subfigure}
    \begin{subfigure}[b]{0.49\textwidth}
        \centering
        \includegraphics[width=\textwidth]{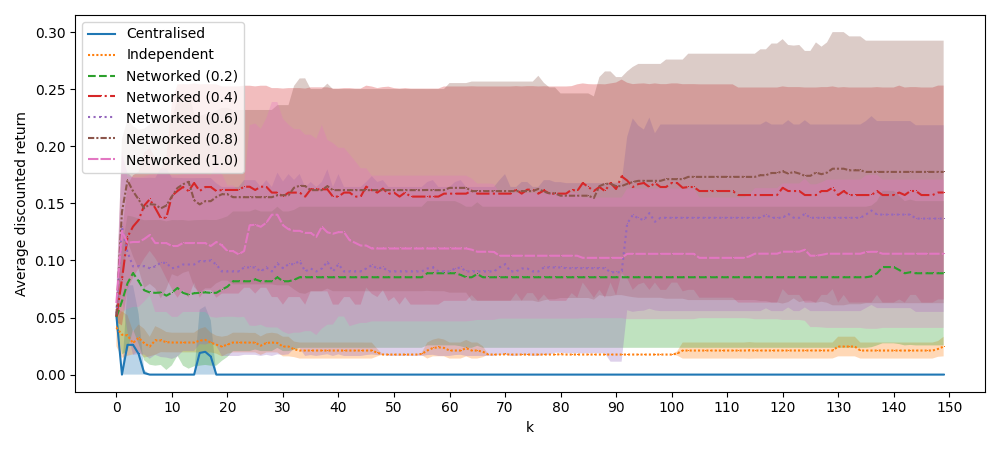}
        \caption{`Shape formation' game.}
        \label{}
    \end{subfigure}
    \caption{Standard algorithms but $C_e=C_r=C_p=$ \textbf{10}. As is expected, in the coordination games the networked agents with lower broadcast radii now receive returns almost as high as those with larger radii, albeit at the cost of greater variance (having more communication rounds leads to greater policy consensus in the population at each iteration of the outer loop, and there may be some noise in the quality of these consensus policies). In the `target selection' game, now all networked populations appear to outperform the central-agent (orange) population, though again with high variance. In the anti-coordination `target coverage' game, the smaller broadcast radii (green, 0.2; red, 0.4; purple, 0.6) receive slightly lower returns than before, since the additional communication rounds now make policy alignment more likely, reducing $f_d$ as per Def \ref{anti-coordination_game_definiton}.  The same is true of the smallest radius population (green, 0.2) in the `shape formation' game, which receives a lower return than before. This reflects the discussion in Sec. \ref{results_and_discussion} regarding the detrimental effect of additional policy adoption once the maximum base return has been achieved in anti-coordination games. Nevertheless, \textit{all} networked populations receive higher returns than the independent agents in all games, and also than the central-agent population in all but the `cluster' game. This shows that in our experimental settings there is a very large benefit to a single communication round, with limited benefit to increasing the algorithms' time complexity with additional communication rounds. 
    }
    \label{10commfig}
\end{figure*}

\newpage

\begin{figure*}[t]
    \centering
    \begin{subfigure}[b]{0.49\textwidth}
        \centering
        \includegraphics[width=\textwidth]{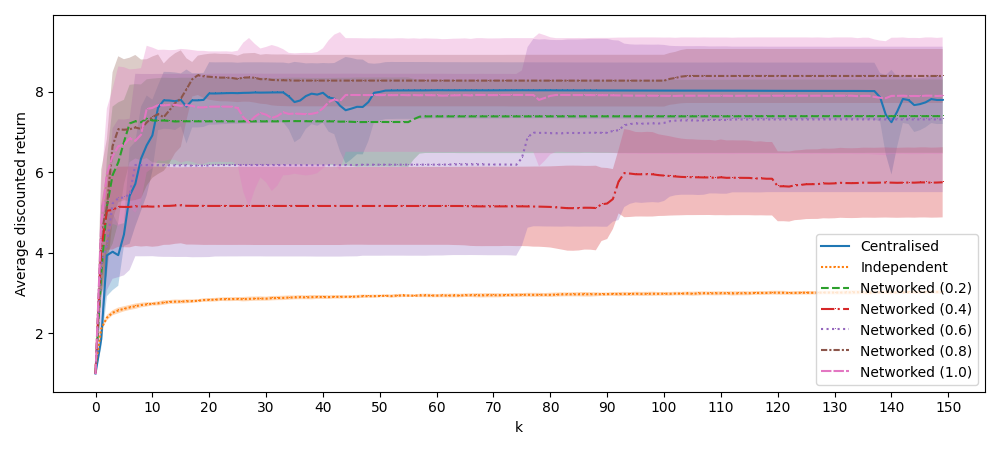}
        \caption{`Cluster' game.}
        \label{}
    \end{subfigure}
    \begin{subfigure}[b]{0.49\textwidth}
        \centering
        \includegraphics[width=\textwidth]{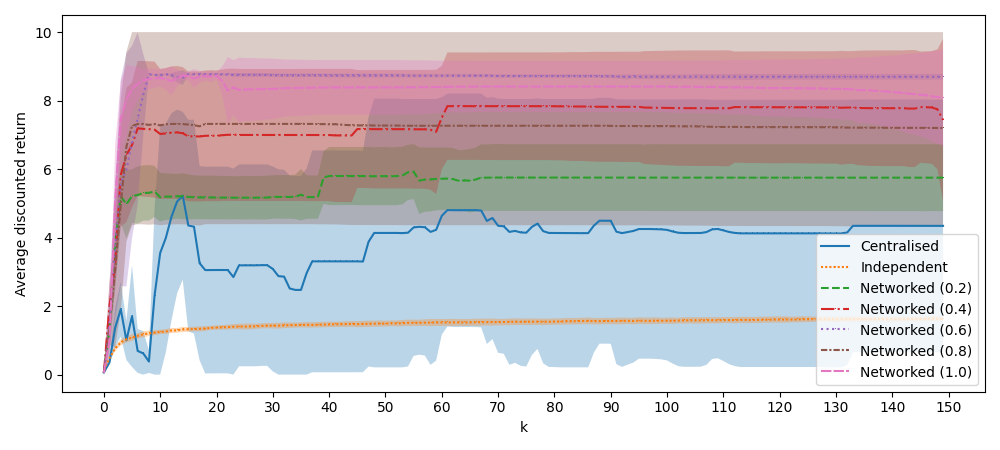}
        \caption{`Target selection' game.}
        \label{}
    \end{subfigure}
    \begin{subfigure}[b]{0.49\textwidth}
        \centering
        \includegraphics[width=\textwidth]{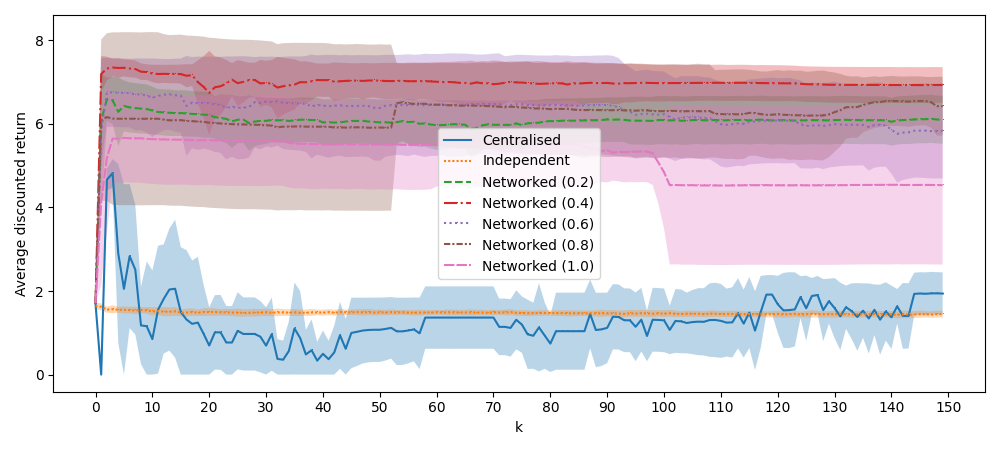}
        \caption{`Disperse' game.}
        \label{}
    \end{subfigure}
    \begin{subfigure}[b]{0.49\textwidth}
        \centering
        \includegraphics[width=\textwidth]{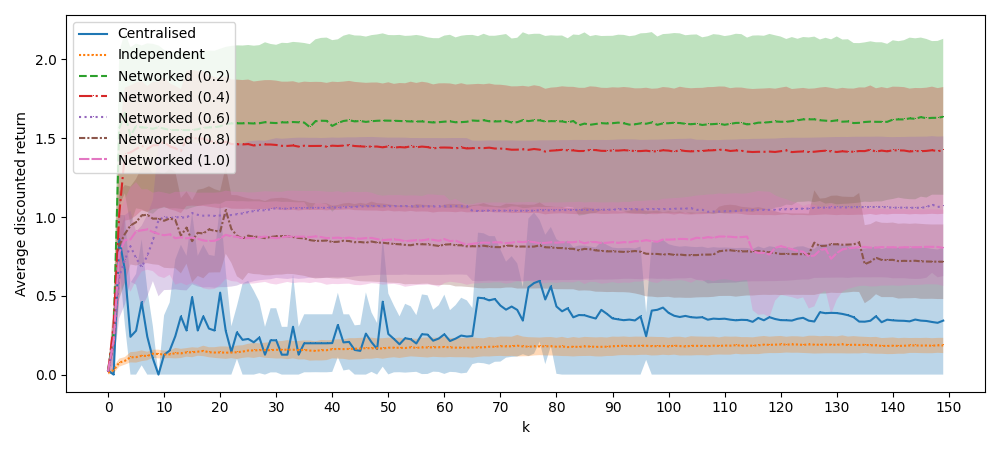}
        \caption{`Target coverage' game.}
        \label{}
    \end{subfigure}
    
    \begin{subfigure}[b]{0.49\textwidth}
        \centering
        \includegraphics[width=\textwidth]{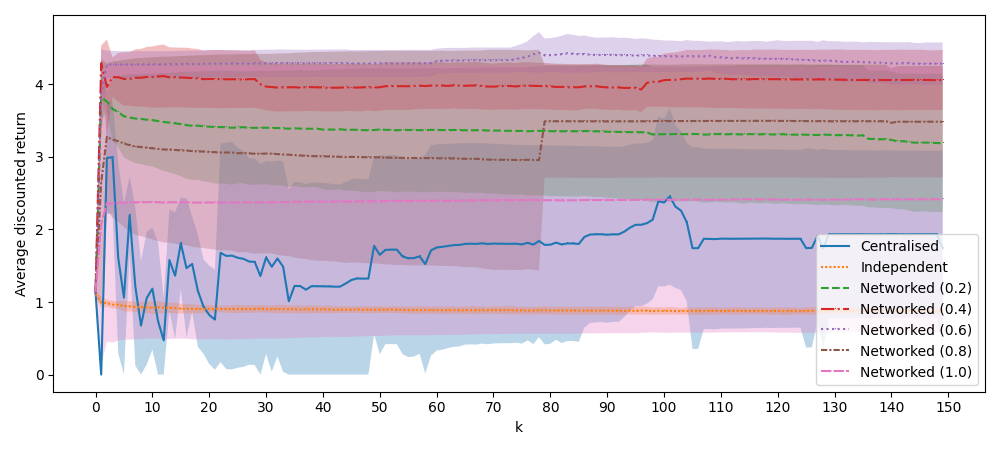}
        \caption{`Beach bar' game.}
        \label{}
    \end{subfigure}
    \begin{subfigure}[b]{0.49\textwidth}
        \centering
        \includegraphics[width=\textwidth]{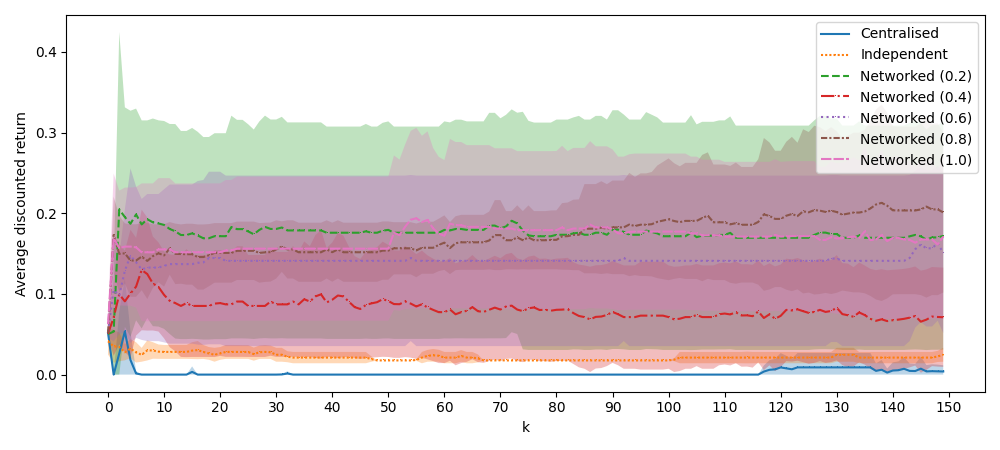}
        \caption{`Shape formation' game.}
        \label{}
    \end{subfigure}
    \caption{Standard algorithms but $C_e=C_r=C_p=$ \textbf{50}. Having 50 communication rounds does not appear to significantly change networked performance compared to 10  rounds (Fig. \ref{10commfig}), with most increases or decreases in average return appearing within the margin of error. Most notably, the largest broadcast radius (pink, 1.0) receives slightly lower return now than with 10 rounds in the `disperse' game, while pink (1.0), brown (0.8) and green (0.2) receive lower returns and have higher variance now in the `beach bar' game. As in the case of $C_e=C_r=C_p=10$, additional communication rounds make policy alignment more likely, reducing $f_d$ as per Def \ref{anti-coordination_game_definiton}. This reflects the discussion in Sec. \ref{results_and_discussion} regarding the detrimental effect of additional policy adoption once the maximum base return has been achieved in anti-coordination games. Nevertheless, \textit{all} networked populations receive higher returns than the independent agents in all games, and also than the central-agent population in all but the `cluster' game. This shows that in our experimental settings there is a very large benefit to a single communication round, with limited benefit to increasing the algorithms' time complexity with additional communication rounds. 
    }
    \label{50commfig}
\end{figure*}

\newpage

\begin{figure*}[t]
    \centering
    \begin{subfigure}[b]{0.49\textwidth}
        \centering
        \includegraphics[width=\textwidth]{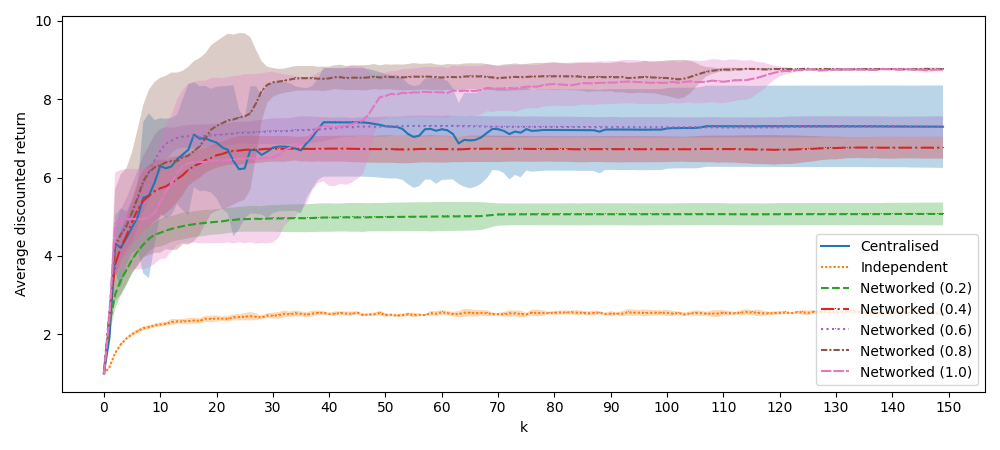}
        \caption{`Cluster' game.}
        \label{}
    \end{subfigure}
    \begin{subfigure}[b]{0.49\textwidth}
        \centering
        \includegraphics[width=\textwidth]{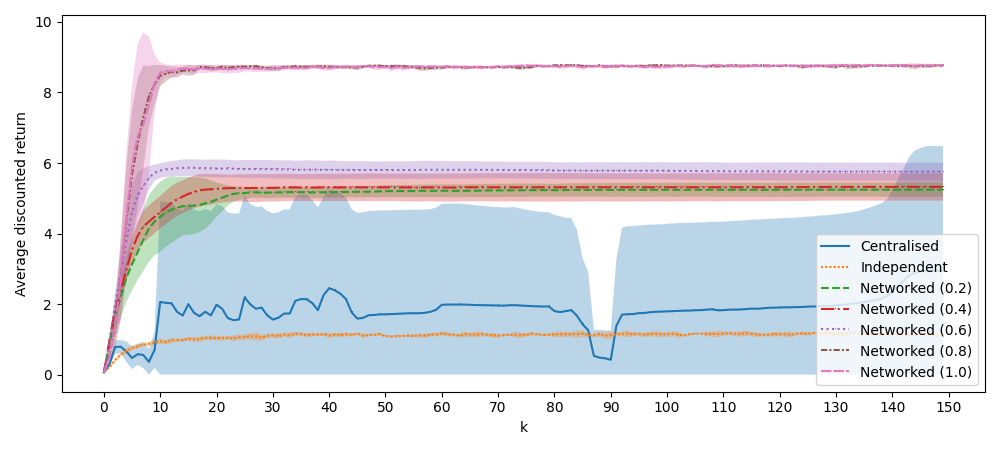}
        \caption{`Target selection' game.}
        \label{}
    \end{subfigure}
    \begin{subfigure}[b]{0.49\textwidth}
        \centering
        \includegraphics[width=\textwidth]{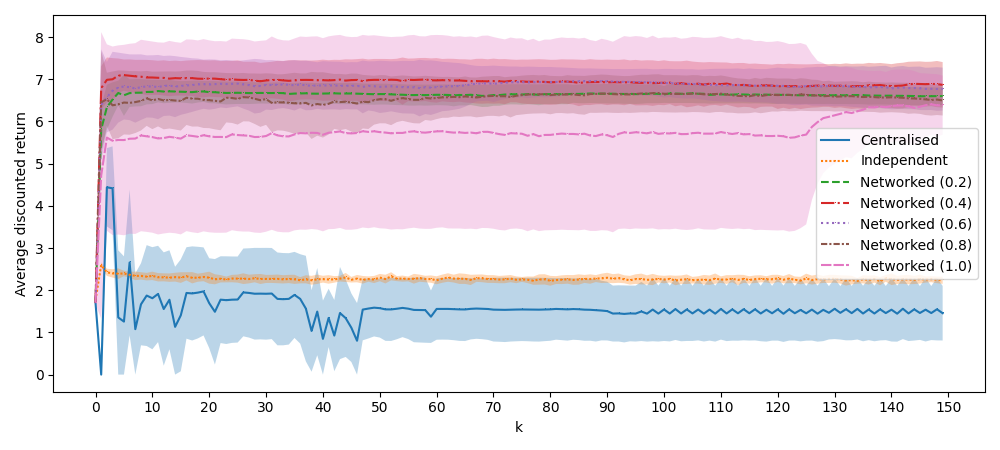}
        \caption{`Disperse' game.}
        \label{}
    \end{subfigure}
    \begin{subfigure}[b]{0.49\textwidth}
        \centering
        \includegraphics[width=\textwidth]{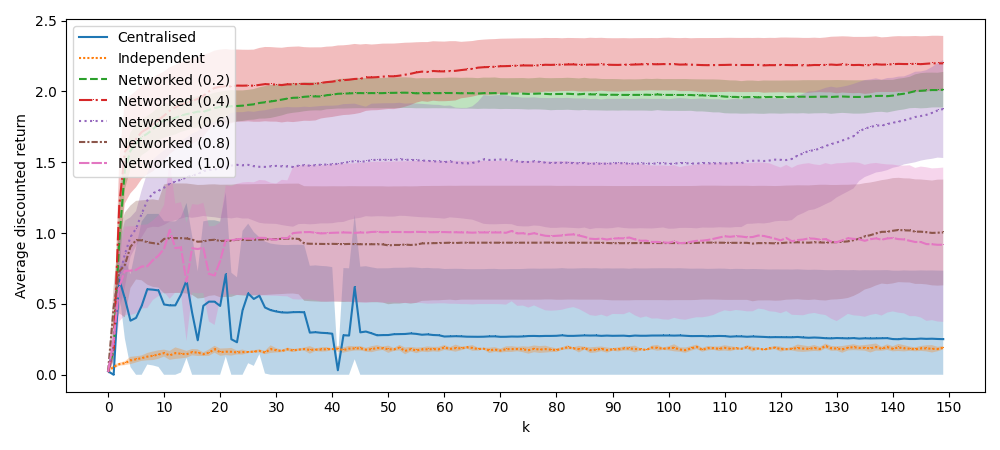}
        \caption{`Target coverage' game.}
        \label{}
    \end{subfigure}
    \begin{subfigure}[b]{0.49\textwidth}
        \centering
        \includegraphics[width=\textwidth]{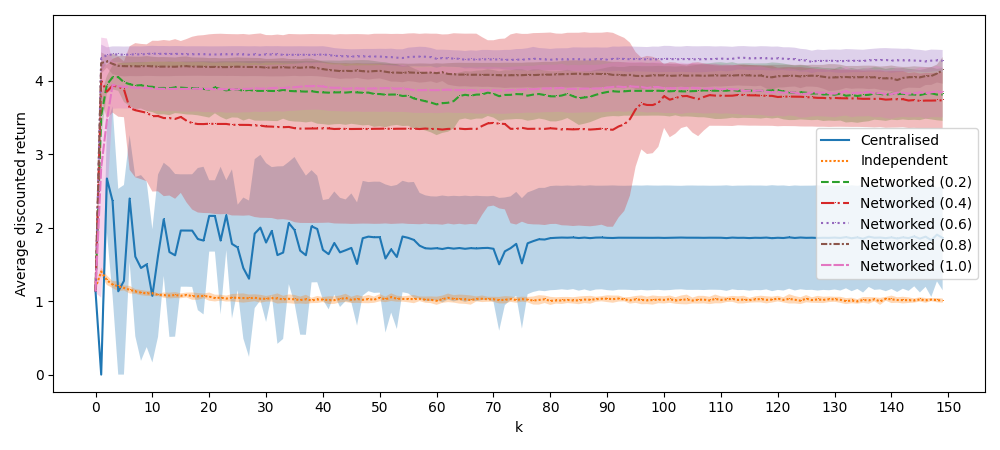}
        \caption{`Beach bar' game.}
        \label{}
    \end{subfigure}
    \begin{subfigure}[b]{0.49\textwidth}
        \centering
        \includegraphics[width=\textwidth]{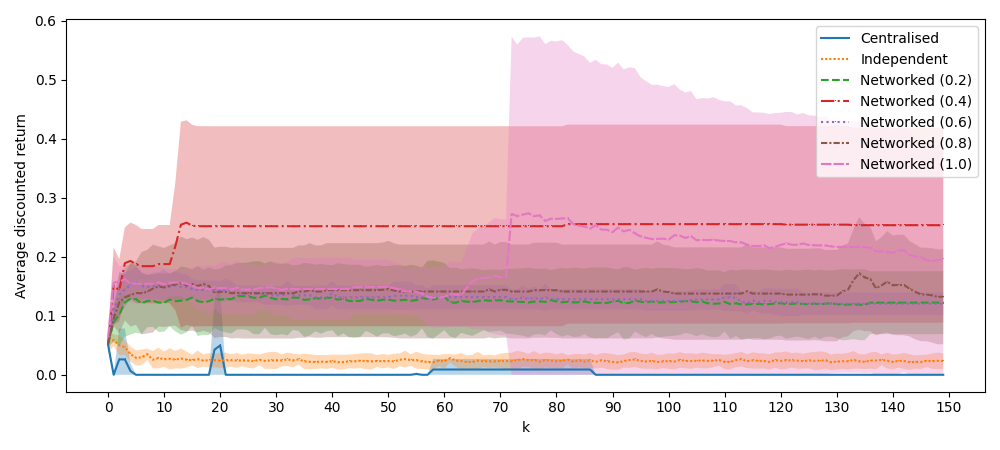}
        \caption{`Shape formation' game.}
        \label{}
    \end{subfigure}
    \caption{Ablation study on population-\textit{independent} policies. No agents, including centralised and networked ones, observe the empirical mean field, and all receive a vector of zeros in its place (so as to keep the neural networks the same size as in the standard setting). $C_r=C_p=1$. Networked populations do not appear to perform substantially differently to the standard population-dependent setting, though some radii (red, 0.4; pink, 1.0) appear to perform slightly better in the `shape formation' game. This is likely because all of our games have stationary solutions, such that observing the mean field is not actually necessary, even if it could potentially be useful (see Sec. \ref{prelim_mfc} for discussion of the conception of MFC as a central planner trying to guide the population to a distribution that maximises the expected return). Indeed, in the coordination games, and particularly the `target selection' game, the central-agent population receives a lower return in this setting, whereas our networked populations are robust to this change.}
    \label{population_independent_policies_fig}
\end{figure*}

\newpage

\begin{figure*}[t]
    \centering
    \begin{subfigure}[b]{0.49\textwidth}
        \centering
        \includegraphics[width=\textwidth]{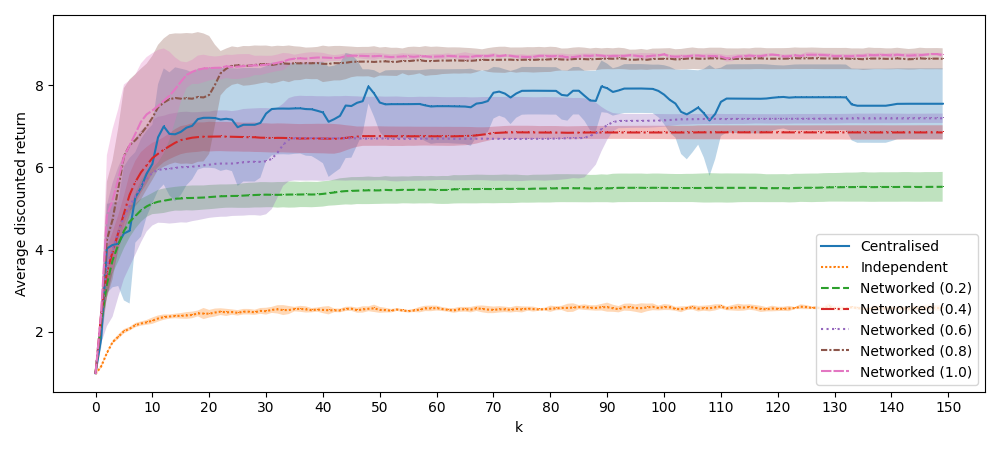}
        \caption{`Cluster' game.}
        \label{}
    \end{subfigure}
    \begin{subfigure}[b]{0.49\textwidth}
        \centering
        \includegraphics[width=\textwidth]{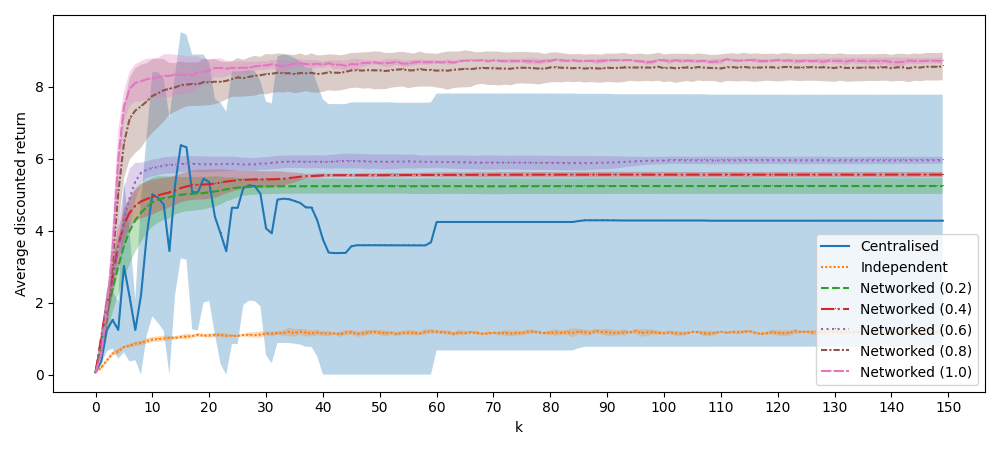}
        \caption{`Target selection' game.}
        \label{}
    \end{subfigure}
    \begin{subfigure}[b]{0.49\textwidth}
        \centering
        \includegraphics[width=\textwidth]{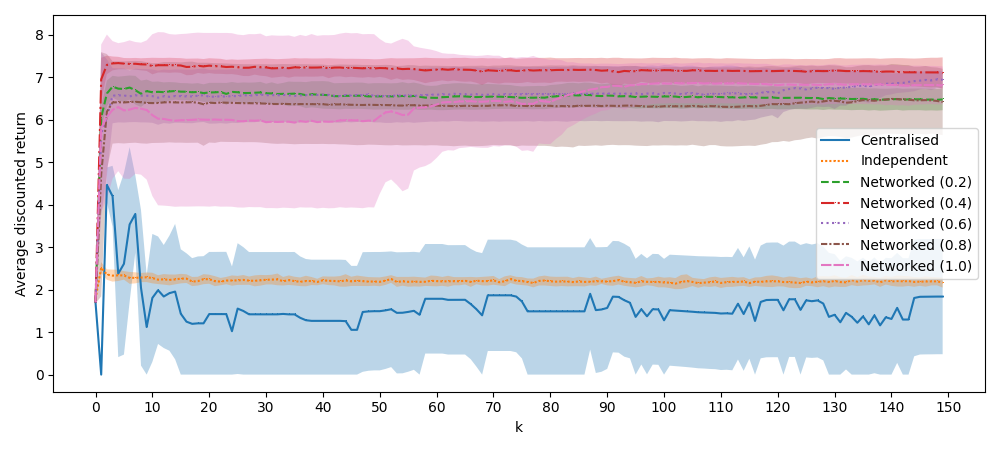}
        \caption{`Disperse' game.}
        \label{}
    \end{subfigure}
    \begin{subfigure}[b]{0.49\textwidth}
        \centering
        \includegraphics[width=\textwidth]{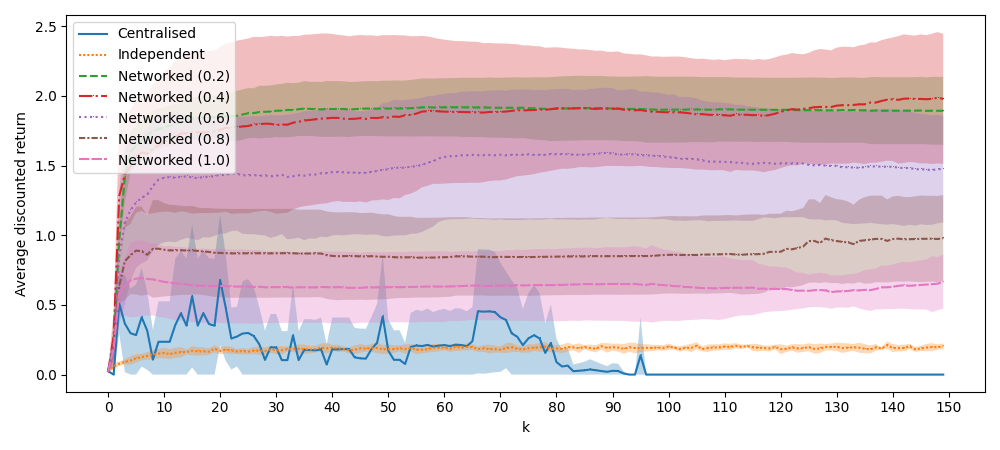}
        \caption{`Target coverage' game.}
        \label{}
    \end{subfigure}
    
    \begin{subfigure}[b]{0.49\textwidth}
        \centering
        \includegraphics[width=\textwidth]{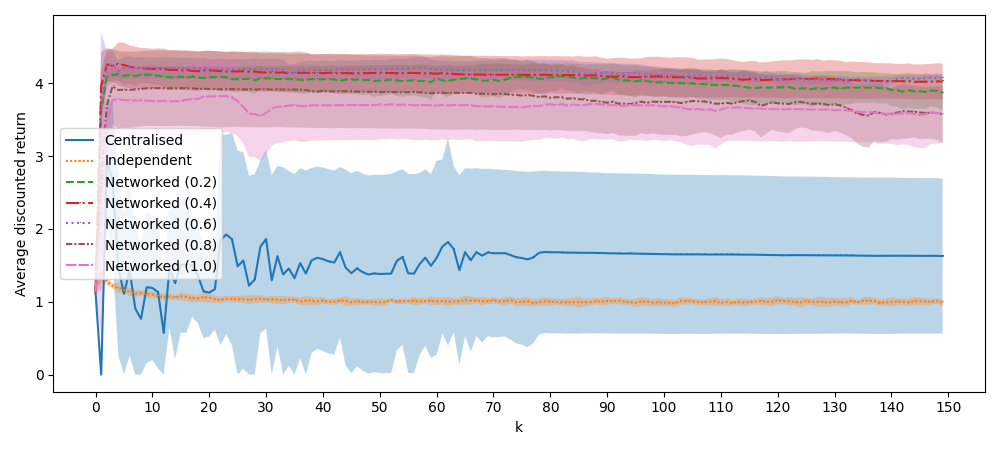}
        \caption{`Beach bar' game.}
        \label{}
    \end{subfigure}
    \begin{subfigure}[b]{0.49\textwidth}
        \centering
        \includegraphics[width=\textwidth]{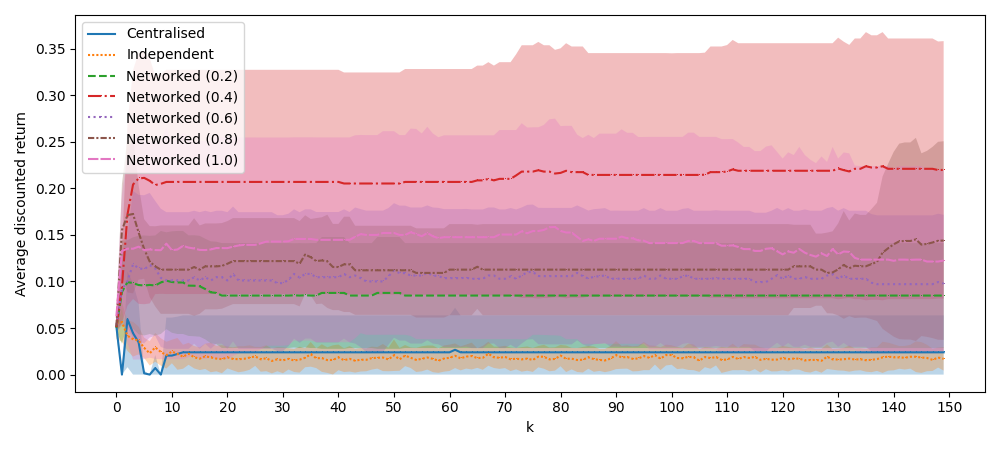}
        \caption{`Shape formation' game.}
        \label{}
    \end{subfigure}
    \caption{Ablation study of Alg. \ref{alg:mean_field_estimation_specific} for estimating the empirical mean field - all agents, including independent ones, directly receive the true global empirical mean field. $C_r=C_p=1$. This does not appear to change performance in the networked populations (apart from greater variance here in the `shape formation' game), nor does it help independent agents. This may be evidence that Alg. \ref{alg:mean_field_estimation_specific} enables networked agents to accurately estimate the global mean field from local observations. However, our ablation study on population-independent policies (Fig. \ref{population_independent_policies_fig}) suggests that not observing the mean field does not markedly disadvantage agents in our experimental settings in any case (apart from for the central-agent populations in the coordination games). This is likely because all of our games have stationary solutions, such that observing the mean field is not necessary. Therefore, in order to confirm the efficacy of Alg. \ref{alg:mean_field_estimation_specific} for estimating the mean field, further evidence is perhaps needed in MFC settings that require population-dependent policies, though \citet{benjamin2024networkedapproximation} has already confirmed this for non-stationary games in the non-cooperative MFG setting.}
    \label{global_mean_field_fig}
\end{figure*}

\newpage

\begin{figure*}[t]
    \centering
    \begin{subfigure}[b]{0.49\textwidth}
        \centering
        \includegraphics[width=\textwidth]{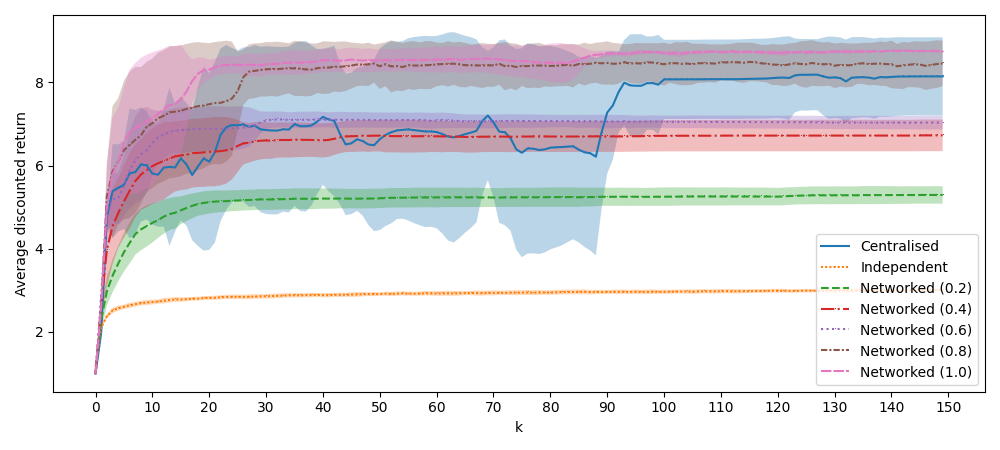}
        \caption{`Cluster' game.}
        \label{}
    \end{subfigure}
    \begin{subfigure}[b]{0.49\textwidth}
        \centering
        \includegraphics[width=\textwidth]{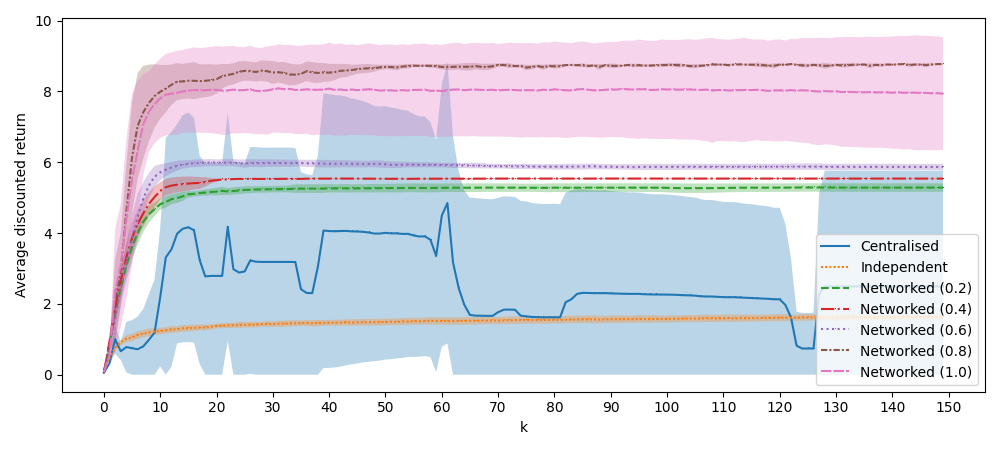}
        \caption{`Target selection' game.}
        \label{}
    \end{subfigure}
    \begin{subfigure}[b]{0.49\textwidth}
        \centering
        \includegraphics[width=\textwidth]{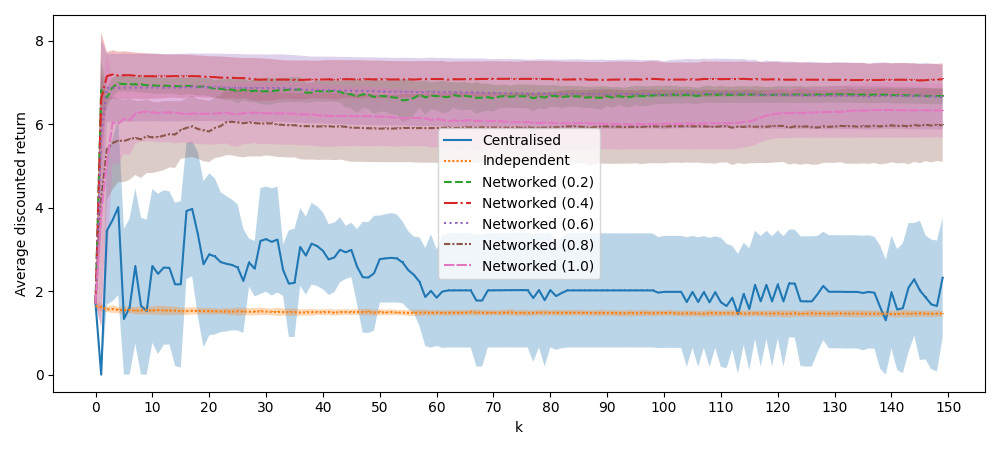}
        \caption{`Disperse' game.}
        \label{}
    \end{subfigure}
    \begin{subfigure}[b]{0.49\textwidth}
        \centering
        \includegraphics[width=\textwidth]{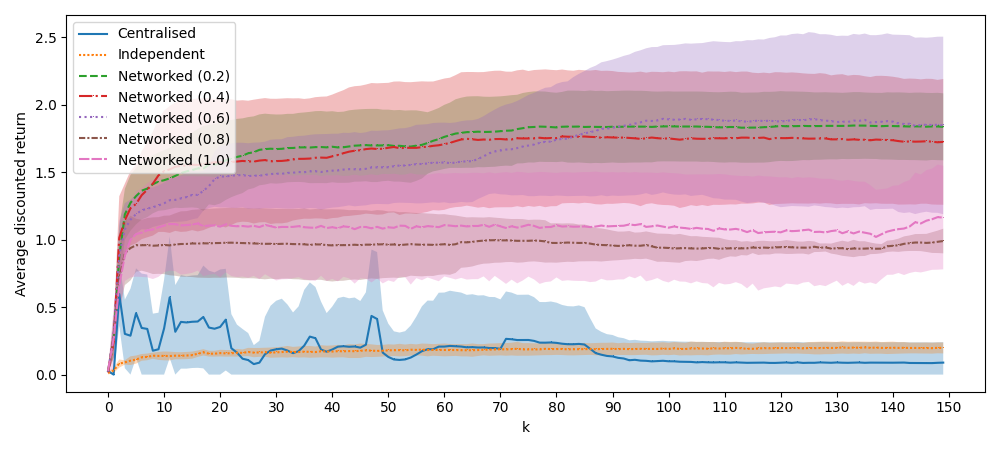}
        \caption{`Target coverage' game.}
        \label{}
    \end{subfigure}
    \begin{subfigure}[b]{0.49\textwidth}
        \centering
        \includegraphics[width=\textwidth]{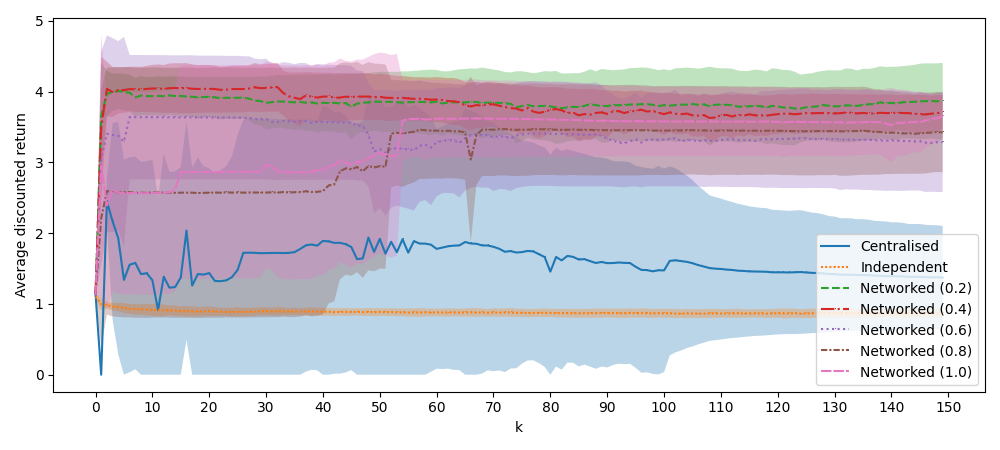}
        \caption{`Beach bar' game.}
        \label{}
    \end{subfigure}
    \begin{subfigure}[b]{0.49\textwidth}
        \centering
        \includegraphics[width=\textwidth]{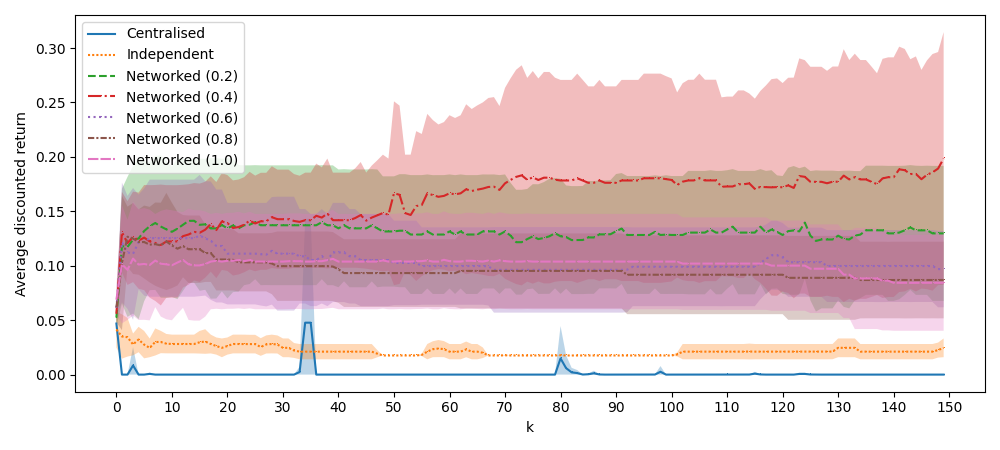}
        \caption{`Shape formation' game.}
        \label{}
    \end{subfigure}
    \caption{Ablation study for observation of true/estimated global average reward $\hat{r}$/$\tilde{\hat{r}}^i_t$, where all agents, including centralised ones, only have access to $r^i_t$, where in the central-agent case $i=1$. $C_e=C_p=1$. The greatest effect of this is on the {central-agent} (blue) populations, which perform much worse in the `target selection' game, and with higher variance in the `cluster' and `beach bar' games, i.e. they suffer without access to the global average reward. The networked agents appear more robust to the loss of the (estimated) average reward, pointing to an additional benefit of the policy communication scheme, though do experience a slight performance decrease, mostly among populations with the largest broadcast radii (pink, 1.0; brown, 0.8), i.e. those most similar to the central-agent case in terms of $\tilde{\hat{r}}^i_t$, as might be expected. In particular, note the greater variance of pink (1.0) in the `target selection' game; slower learning and higher variance of pink (1.0) and brown (0.8) in the `beach bar' game; lower returns for pink (1.0) and brown (0.8) in the `shape formation' game; and slower learning and convergence of the smallest radii (green, 0.2; red, 0.4) in the `target coverage' game. This all demonstrates the usefulness and efficacy of our novel Alg. \ref{average_reward_alg} for decentralised estimation of the global average reward.}
    \label{no_average_reward_fig}
\end{figure*}

\newpage

\begin{figure*}[t]
    \centering
    \begin{subfigure}[b]{0.49\textwidth}
        \centering
        \includegraphics[width=\textwidth]{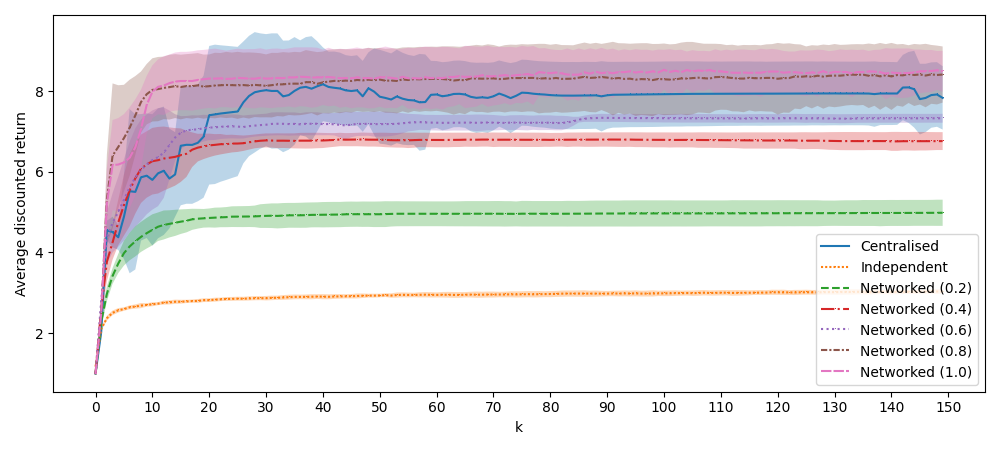}
        \caption{`Cluster' game.}
        \label{}
    \end{subfigure}
    \begin{subfigure}[b]{0.49\textwidth}
        \centering
        \includegraphics[width=\textwidth]{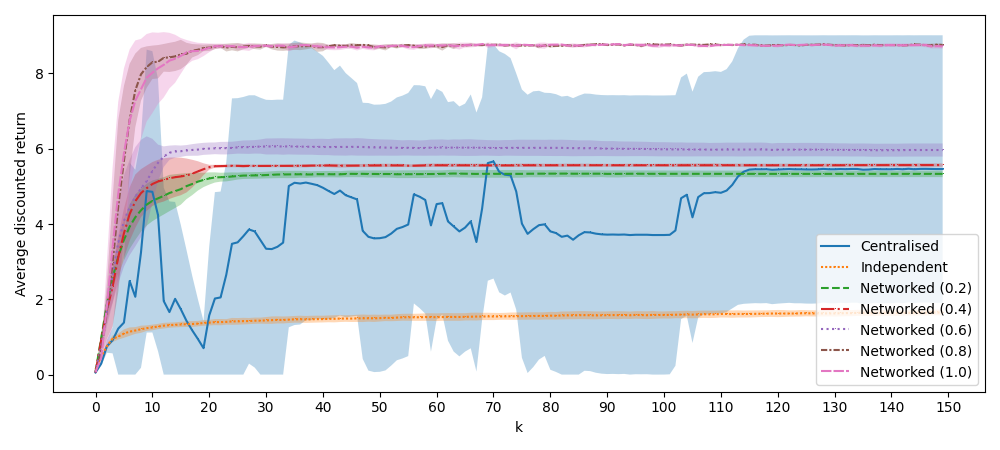}
        \caption{`Target selection' game.}
        \label{}
    \end{subfigure}
    \begin{subfigure}[b]{0.49\textwidth}
        \centering
        \includegraphics[width=\textwidth]{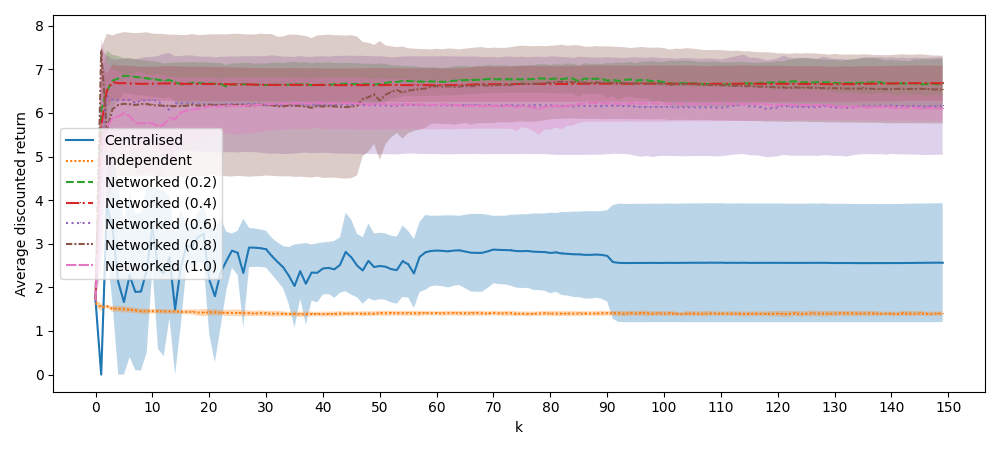}
        \caption{`Disperse' game.}
        \label{}
    \end{subfigure}
    \begin{subfigure}[b]{0.49\textwidth}
        \centering
        \includegraphics[width=\textwidth]{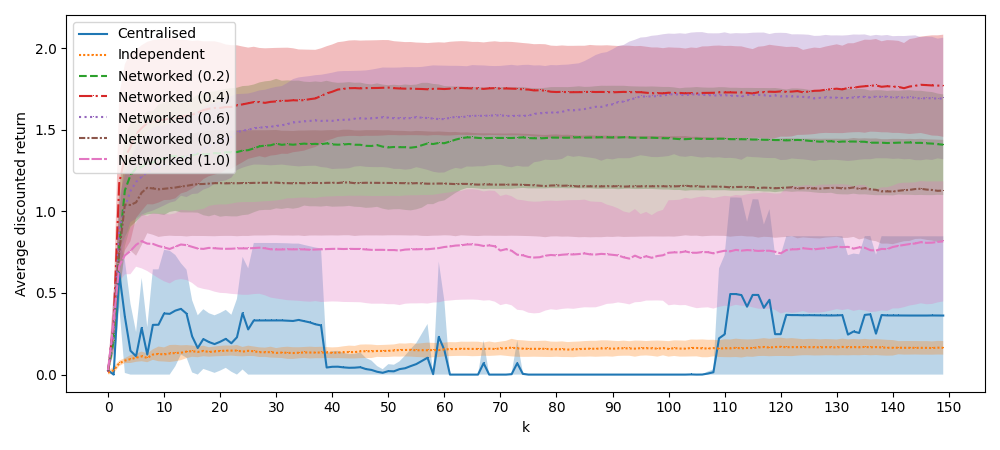}
        \caption{`Target coverage' game.}
        \label{}
    \end{subfigure}
    
    \begin{subfigure}[b]{0.49\textwidth}
        \centering
        \includegraphics[width=\textwidth]{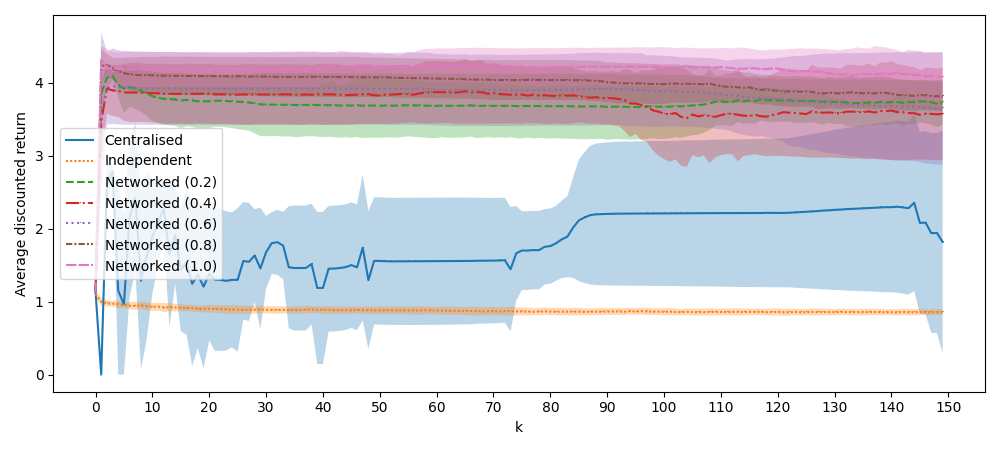}
        \caption{`Beach bar' game.}
        \label{}
    \end{subfigure}
    \begin{subfigure}[b]{0.49\textwidth}
        \centering
        \includegraphics[width=\textwidth]{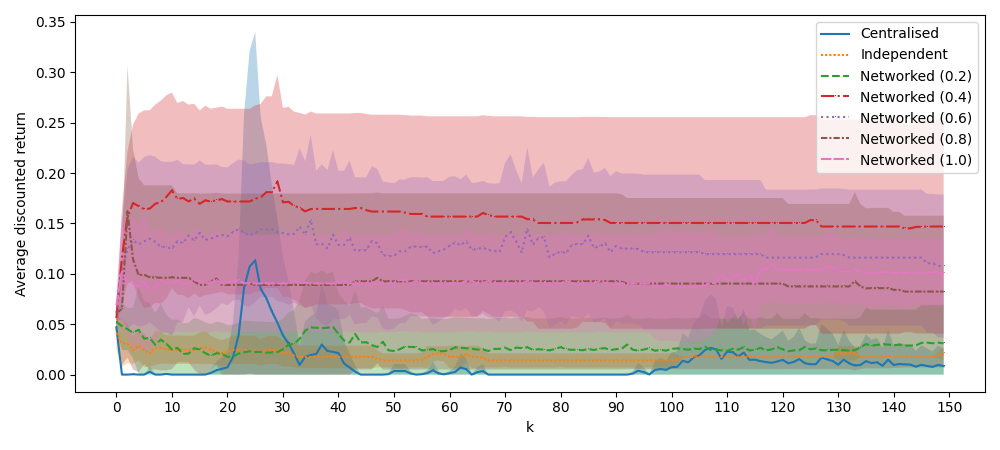}
        \caption{`Shape formation' game.}
        \label{}
    \end{subfigure}
    \caption{Ablation study for Alg. \ref{average_reward_alg} for estimating the true global average reward. All agents, including both networked and independent ones, directly receive the true global average reward such that $\tilde{\hat{r}}^i_t = \hat{r}$. $C_e=C_p=1$. Access to the true average reward does not help networked agents to improve their returns, demonstrating that our novel Alg. \ref{average_reward_alg} already affords networked populations robustness against the lack of access to this global information (having this global information would be an unrealistic assumption in practice). Access to the true average reward also does not help independent agents to improve their returns, suggesting the \textit{policy} communication scheme is the dominant factor in improving the performance of decentralised agents. 
    }
    \label{all_global_average_fig}
\end{figure*}

\newpage

\begin{figure*}[t]
    \centering
    \begin{subfigure}[b]{0.49\textwidth}
        \centering
        \includegraphics[width=\textwidth]{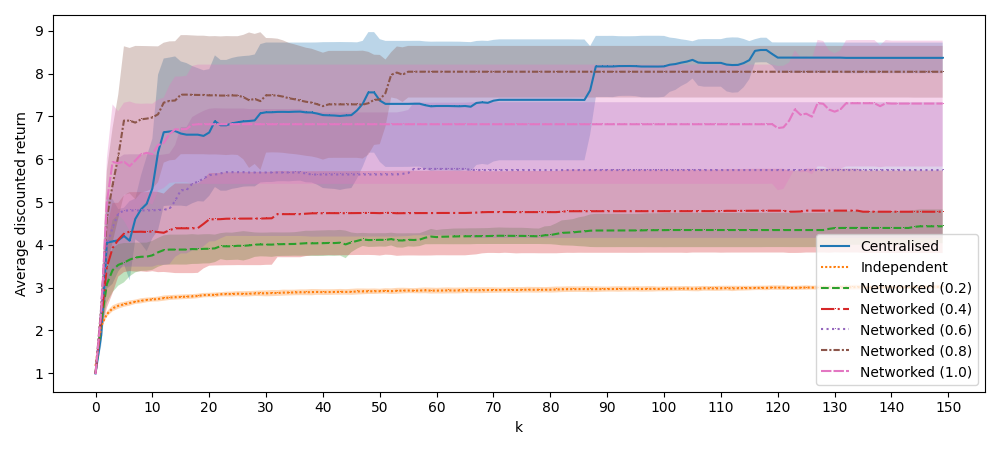}
        \caption{`Cluster' game.}
        \label{}
    \end{subfigure}
    \begin{subfigure}[b]{0.49\textwidth}
        \centering
        \includegraphics[width=\textwidth]{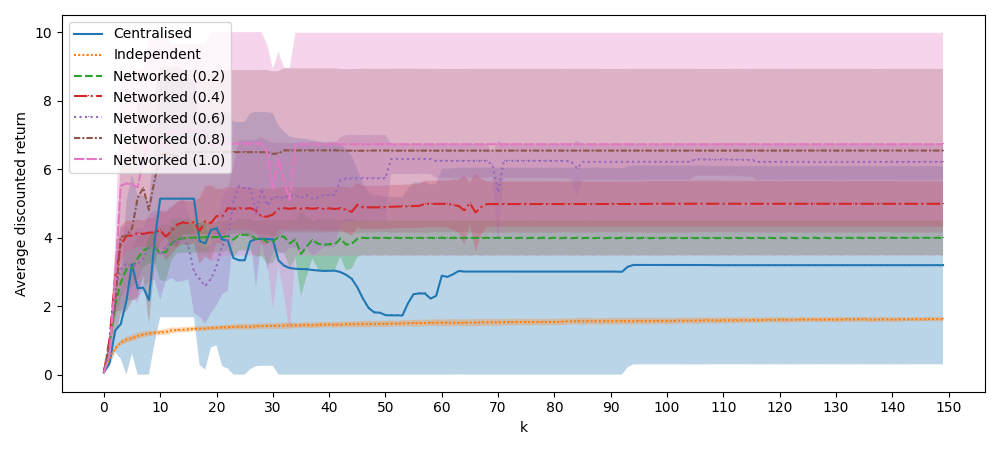}
        \caption{`Target selection' game.}
        \label{}
    \end{subfigure}
    \begin{subfigure}[b]{0.49\textwidth}
        \centering
        \includegraphics[width=\textwidth]{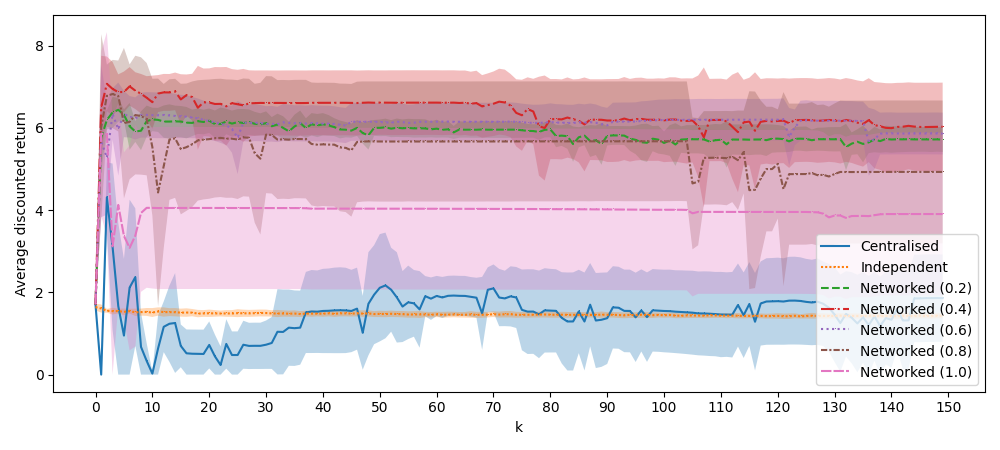}
        \caption{`Disperse' game.}
        \label{}
    \end{subfigure}
    \begin{subfigure}[b]{0.49\textwidth}
        \centering
        \includegraphics[width=\textwidth]{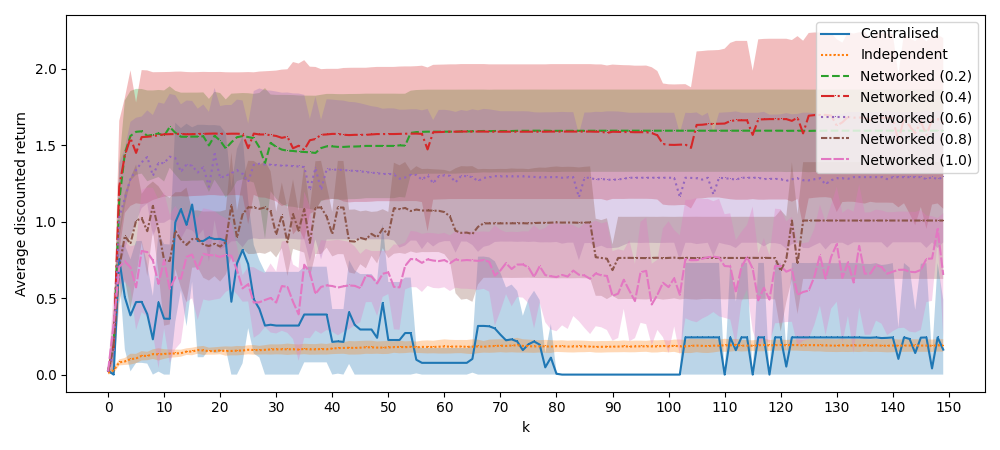}
        \caption{`Target coverage' game.}
        \label{}
    \end{subfigure}
    
    \begin{subfigure}[b]{0.49\textwidth}
        \centering
        \includegraphics[width=\textwidth]{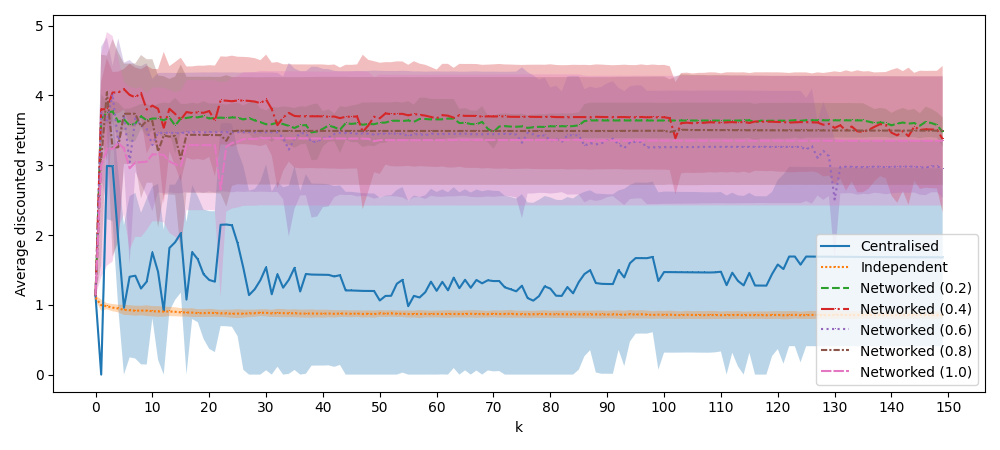}
        \caption{`Beach bar' game.}
        \label{}
    \end{subfigure}
    \begin{subfigure}[b]{0.49\textwidth}
        \centering
        \includegraphics[width=\textwidth]{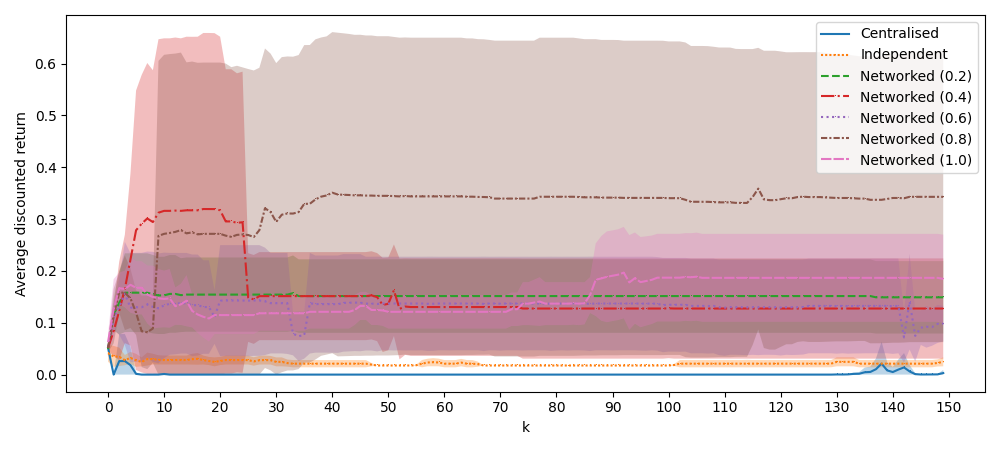}
        \caption{`Shape formation' game.}
        \label{}
    \end{subfigure}
    \caption{Ablation study of the choice of $\tau^{comm}_k$. Here $\forall k \;\tau^{comm}_k =$ 1e-18 (i.e. $\tau^{comm}_k$ is close to 0, turning the softmax into a max function), rather than linearly increasing from 0.001 to 1 across the $K$ iterations as in all other experiments (see Table \ref{Hyperparameters}). $C_e=C_r=C_p=1$. In this setting, networked agents continue to outperform the central-agent (blue) and independent (orange) populations in all games except the `cluster' game, but otherwise generally appear to receive lower average returns and with greater variance. This is because Assumption \ref{approximation_ordering_assumption} on the quality of the finite-step approximations \{$\sigma^i_{k+1}$\}$_{i=1}^{N}$ = $\{\hat{V}^i(\boldsymbol{\pi}_{k+1},\mu_t; E)\}_{i=1}^{N}$ may not always apply in practice, especially as policies get more deterministic closer to convergence. This means the policy estimated to perform the best may not actually be among the best updates, such that enforcing the adoption of this policy can lead to noisy, unstable learning. Using a higher temperature value smooths out this noise (while having a lower temperature at the beginning of training ensures faster learning when the difference in the quality of nascent policies is likely to be more stark, hence our inverse annealing scheme). Moreover, using $\tau^{comm}_k$ close to 0 more effectively enforces consensus on a single policy in the networked case, which in anti-coordination games may also reduce the average return (see Sec. \ref{results_and_discussion}). This all provides empirical evidence for our scheme for $\tau^{comm}_k$, but further optimising the choice might lead to additional performance increase.}
    \label{tau_ablation_fig}
\end{figure*}

\newpage 

\;

\newpage

\;

\newpage

\;

\newpage

\;

\newpage

\;

\newpage

\;

\newpage

\;

\newpage



\section{Conclusion}\label{conclusion_section}

We provided the first algorithms for decentralised training in MFC, as well as the first for online learning in MFC from a single non-episodic run of the empirical system. We did so by modifying existing algorithms for the MFG setting, and contributing a novel algorithm for estimating the global average reward via local communication. We proved theoretically that networked communication accelerates learning over both the independent and central-agent architectures. We supported this with extensive numerical results, accompanied by ablation studies and discussion of the empirical effects of communication radii. 

Our work follows the gold standard in MFGs by presenting experiments on grid world toy environments. Nevertheless future work includes experiments in other games, including non-stationary games, more realistic environments and ones where both the transition function and the reward function depend on the mean field. Please note, however, that \citet{benjamin2024networkedapproximation} already demonstrated in the MFG setting that the communication scheme affords faster learning when both the transition and reward functions depend on the mean field.

In Sec. \ref{theory_section} we give theoretical results showing that our networked algorithm can outperform the central-agent and independent alternative. We leave more general theoretical results, such as proofs of convergence and sample complexity, for future work. 

Our algorithms contain numerous inner loops and thus require synchronisation between communicating agents. Our ablation studies of the sub-routines and our experiment on robustness to communication failures (Fig. \ref{comm_failure_fig}) indicate that synchronisation failure is not necessarily a problem in practice, but future work nevertheless lies in simplifying the nested loops of our algorithms. 

In grid-world settings such as those in our experiments, passing the (estimated or true global) mean-field distribution as a flat vector to the Q-network ignores the geometric structure of the problem. \citet{perrin2022generalization} therefore proposes to create an embedding of the distribution by first passing the vector to a convolutional neural network, essentially treating the categorical distribution as an image. This technique is also followed in \citet{wu2024populationaware} (for their additional experiments, but not in the main body of their paper). As future work, we can test whether such a method increases the usefulness of observing the mean field in population-dependent policies, and therefore increases the importance of be able to accurately estimate the global mean field via Alg. \ref{alg:mean_field_estimation_specific}.

\subsubsection*{Broader Impact Statement}
We identified no specific ethical concerns regarding our work, which explores new game theoretical and machine learning algorithms in general settings.

\bibliography{main}
\bibliographystyle{tmlr}


\end{document}